%% file: main.tex
\documentclass[11pt,a4paper]{article}

\usepackage[margin=1in]{geometry}

\usepackage{theorem,latexsym,graphicx,amssymb}
\usepackage{amsmath,enumerate}
\usepackage{wrapfig}
\usepackage{subfigure}
\usepackage{float}
\usepackage{graphicx}
\usepackage{epsfig}
\usepackage{xspace}
\usepackage{paralist}
\usepackage{times}
\usepackage[compact]{titlesec}
\usepackage{algpseudocode}
\usepackage{algorithm}

\usepackage{enumerate}
\usepackage{cases}
\usepackage{caption}
\usepackage{url}
\usepackage{tikz}
\usetikzlibrary{arrows}
\usetikzlibrary{topaths,calc}
\usepackage{booktabs}
\usepackage{mathtools}

\usepackage{color}


\usepackage{multicol}

\newenvironment{proof}{{\bf Proof:  }}{\hfill\rule{2mm}{2mm}}
\newenvironment{proofof}[1]{{\bf Proof of #1:  }}{\hfill\rule{2mm}{2mm}}

\numberwithin{figure}{section}
\numberwithin{equation}{section}

\newtheorem{theorem}{Theorem}[section]
\newtheorem{definition}[theorem]{Definition}
\newtheorem{remark}[theorem]{Remark}
\newtheorem{corollary}[theorem]{Corollary}
\newtheorem{lemma}[theorem]{Lemma}
\newtheorem{claim}[theorem]{Claim}
\newtheorem{fact}[theorem]{Fact}

\newtheorem{hypothesis}{Hypothesis}

\newtheorem{example}[theorem]{Example}

\newtheorem{proposition}{\hskip\parindent Proposition}[section]

\newcommand{\defeq}{:=}
\newcommand{\etal}{\emph{et al.}\xspace}

\newcommand{\N}{\mathbb{N}}

\newcommand{\Z}{\mathbb{Z}}
\newcommand{\R}{\mathbb{R}}
\newcommand{\E}{\mathbb{E}}

\newcommand{\W}{\ensuremath{\mathsf{W}}\xspace}
\newcommand{\Wm}{\ensuremath{\mathsf{W}^{-1}}\xspace}

\newcommand{\Wh}{\ensuremath{\mathsf{W}^{\frac{1}{2}}}\xspace}

\newcommand{\Wmh}{\ensuremath{\mathsf{W}^{-\frac{1}{2}}}\xspace}

\newcommand{\Lo}{\ensuremath{\mathsf{L}\xspace}}
\newcommand{\Lc}{\ensuremath{\mathcal{L}}\xspace}

\newcommand{\Mo}{\ensuremath{\mathsf{M}\xspace}}
\newcommand{\Mc}{\ensuremath{\mathcal{M}}\xspace}

\newcommand{\D}{\ensuremath{\mathsf{D}\xspace}}
\newcommand{\Dc}{\ensuremath{\mathcal{D}}\xspace}

\newcommand{\M}{\ensuremath{\mathsf{M}\xspace}}

\newcommand{\I}{\ensuremath{\mathsf{I}\xspace}}

\newcommand{\w}[1]{\ensuremath{w_{#1}}}
\newcommand{\f}[1]{\ensuremath{f_{#1}}}
\newcommand{\g}[1]{\ensuremath{g_{#1}}}

\newcommand{\ray}{\ensuremath{\mathsf{R}}\xspace}
\newcommand{\rayc}{\ensuremath{\mathcal{R}}\xspace}

\newcommand{\ve}{\varepsilon}
\newcommand{\vp}{\varphi}

\newcommand{\T}{\ensuremath{\mathsf{T}}\xspace}
\newcommand{\ra}{\rightarrow}

\newcommand{\spn}{\operatorname{span}}

\newcommand{\eps}{\epsilon}
\newcommand{\ceil}[1]{\ensuremath{\lceil #1 \rceil}}

\newcommand{\floor}[1]{\ensuremath{\lfloor #1 \rfloor}}

\newcommand{\Ito}{It\={o}\xspace}

\newcounter{note}[section]
\renewcommand{\thenote}{\thesection.\arabic{note}}

\newcommand{\initOneLiners}{%
    \setlength{\itemsep}{0pt}
    \setlength{\parsep }{0pt}
    \setlength{\topsep }{0pt}
}

\newcommand{\ignore}[1]{}

\newcommand*\samethanks[1][\value{footnote}]{\footnotemark[#1]}

\newcommand{\shortv}[1]{}

\input{macro.tex}

\usepackage{stmaryrd}
\renewcommand{\vec}[1]{\overset{{}_{\shortrightarrow}}{\mathbf{#1}}}

\newcommand{\hubert}[1]{\refstepcounter{note}$\ll${\sf Hubert's
Comment~\thenote:} {\sf \textcolor{blue}{#1}}$\gg$\marginpar{\tiny\bf HC~\thenote}}

\title{Spectral Properties of Hypergraph Laplacian and Approximation Algorithms
\footnote{A preliminary version of this paper appeared in STOC 2015~\cite{Louis15} and the current paper is the result of
a merge with~\cite{ChanTZ15}.}}

\author{T-H. Hubert Chan\thanks{Department of Computer Science, the University of Hong Kong. 
} 
\and Anand Louis\thanks{Princeton University.
Supported by the Simons Collaboration on Algorithms and Geometry. Part of the work was done while the author was a student at Georgia Tech and supported by Santosh Vempala's NSF award CCF-1217793.
} 
\and Zhihao Gavin Tang\samethanks[2] \and Chenzi Zhang\samethanks[2]}
\date{}

\begin{document}

\begin{titlepage}

\maketitle

\input{0_Abstract.tex}

\thispagestyle{empty}
\end{titlepage}

\tableofcontents

\input{intro.tex}

\input{overview.tex}

\input{laplacian.tex}
\input{stochastic.tex}

\input{diameter.tex}

\input{cheeger.tex}

\input{vertex_expansion.tex}

\input{approxminimizer.tex}

\input{sparsest_cut.tex}

\input{ack.tex}

{
\bibliographystyle{alpha}
\bibliography{bibfile}
}

\appendix

\input{tensor.tex}

\input{example.tex}

\end{document}

%% file: macro.tex
\usepackage{tabularx}
\usepackage{amsmath,amstext,amssymb,amsfonts}
\usepackage[mathscr]{euscript}

\newcommand{\sdp}{{\sf SDP }}
\newcommand{\sdpval}{{\sf SDPval}}

\newcommand{\opt}{{\sf OPT}}
\newcommand{\OPT}{{\sf OPT}}

\newcommand{\sse}{{\sf SSE}}
\newcommand{\SSEH}{Small-Set Expansion Hypothesis\xspace}
\newcommand{\ratio}{k\log k \log\log k \cdot \sqrt{\log r}}

\newcommand{\subjectto}{\text{subject to}}

\renewcommand{\vec}[1]{{\bm{#1}}}

\newcommand{\paren}[1]{\left(#1 \right )}

\newcommand{\Brac}[1]{\left[#1\right]}

\newcommand{\set}[1]{\left\{#1\right\}}

\newcommand{\Abs}[1]{\left\lvert#1\right\rvert}

\newcommand{\norm}[1]{\left\lVert#1\right\rVert}

\newcommand{\cL}{\mathcal L}

\newcommand{\cN}{\mathcal N}

\newcommand{\Ex}[1]{\E\Brac{#1}}

\newcommand{\bigO}{\mathcal{O}}
\newcommand{\bigo}[1]{\bigO\left(#1\right)}
\newcommand{\tbigO}{\tilde{\mathcal{O}}}
\newcommand{\tbigo}[1]{\tbigO\left(#1\right)}

\newcommand{\argmax}{{\sf argmax}}

\newcommand{\supp}{{\sf supp}}

\newcommand{\U}{\bar{u}}
\newcommand{\V}{\bar{v}}

\newcommand{\dr}{{\sf dr}}
\newcommand{\diam}{{\sf diam }}

\newcommand{\Nin}{N^{\sf in}}
\newcommand{\Nout}{N^{\sf out}}
\newcommand{\phiv}{\phi^{\sf V}}
\newcommand{\linf}{\lambda_{\infty}}
\newcommand{\mvert}{M^{\sf vert}}
\newcommand{\lapvert}{L^{\sf vert}}

\newcommand{\yes}{\textsc{Yes}\xspace}    
\newcommand{\no}{\textsc{No}\xspace}	

\newcommand{\vol}{{\sf vol }}
\newcommand{\one}{\textbf{1}}

\newcommand{\mper}{\,.}

\newcommand{\eig}{\gamma}
\newcommand{\lh}{\eig_2}
\newcommand{\markov}{heat\xspace}
\newcommand{\Markov}{Heat\xspace}

\newcommand{\rmin}{r_{\min}}
\newcommand{\rmax}{r_{\max}}

\newcommand{\eset}{\emptyset}

\newcommand{\normo}[1]{\norm{#1}_{\scriptstyle 1}}

\newcommand{\normb}[1]{\norm{#1}_{\scriptstyle \square}}

\newcommand{\Ind}[1]{\Isymb\Brac{#1}}

\newcommand{\Isymb}{\mathbb{I}}

\newcommand{\inprod}[1]{\left\langle #1\right\rangle}

\newcommand{\tmix}[1]{{\sf t}^{\sf mix}_{\delta}\paren{#1}}

\newtheorem{theorem*}{Theorem}
\newtheorem{SDP}[theorem]{SDP}
\newtheorem{reduction}[theorem]{Reduction}

\usepackage{boxedminipage}

\newenvironment{mybox}
{\center \noindent\begin{boxedminipage}{1.0\linewidth}}
{\end{boxedminipage}
\noindent
}

\usepackage{prettyref}
\newcommand{\savehyperref}[2]{\texorpdfstring{\hyperref[#1]{#2}}{#2}}

\newrefformat{eq}{\savehyperref{#1}{\textup{(\ref*{#1})}}}
\newrefformat{lem}{\savehyperref{#1}{Lemma~\ref*{#1}}}
\newrefformat{def}{\savehyperref{#1}{Definition~\ref*{#1}}}
\newrefformat{thm}{\savehyperref{#1}{Theorem~\ref*{#1}}}
\newrefformat{cor}{\savehyperref{#1}{Corollary~\ref*{#1}}}
\newrefformat{cha}{\savehyperref{#1}{Chapter~\ref*{#1}}}
\newrefformat{sec}{\savehyperref{#1}{Section~\ref*{#1}}}
\newrefformat{app}{\savehyperref{#1}{Appendix~\ref*{#1}}}
\newrefformat{tab}{\savehyperref{#1}{Table~\ref*{#1}}}
\newrefformat{fig}{\savehyperref{#1}{Figure~\ref*{#1}}}
\newrefformat{hyp}{\savehyperref{#1}{Hypothesis~\ref*{#1}}}
\newrefformat{alg}{\savehyperref{#1}{Algorithm~\ref*{#1}}}
\newrefformat{sdp}{\savehyperref{#1}{SDP~\ref*{#1}}}
\newrefformat{rem}{\savehyperref{#1}{Remark~\ref*{#1}}}
\newrefformat{item}{\savehyperref{#1}{Item~\ref*{#1}}}
\newrefformat{step}{\savehyperref{#1}{step~\ref*{#1}}}
\newrefformat{conj}{\savehyperref{#1}{Conjecture~\ref*{#1}}}
\newrefformat{fact}{\savehyperref{#1}{Fact~\ref*{#1}}}
\newrefformat{prop}{\savehyperref{#1}{Proposition~\ref*{#1}}}
\newrefformat{claim}{\savehyperref{#1}{Claim~\ref*{#1}}}
\newrefformat{relax}{\savehyperref{#1}{Relaxation~\ref*{#1}}}
\newrefformat{red}{\savehyperref{#1}{Reduction~\ref*{#1}}}
\newrefformat{part}{\savehyperref{#1}{Part~\ref*{#1}}}
\newrefformat{prob}{\savehyperref{#1}{Problem~\ref*{#1}}}
\newrefformat{ass}{\savehyperref{#1}{Assumption~\ref*{#1}}}

%% file: 0_Abstract.tex
\begin{abstract}
The celebrated Cheeger's Inequality \cite{am85,a86}  
establishes a bound on the edge expansion of a graph via its spectrum.
This inequality is central to a 
rich spectral theory of graphs, based on studying the eigenvalues and
eigenvectors of the adjacency matrix (and other related matrices) of graphs. 
It has remained open to define a suitable spectral model for hypergraphs whose 
spectra can be used to estimate various combinatorial properties of the hypergraph.

In this paper we introduce a new hypergraph
Laplacian operator generalizing the Laplacian matrix of graphs.
In particular, the operator is induced by a diffusion process on the hypergraph,
such that within each hyperedge, measure flows from vertices having maximum weighted measure to those having minimum.  
Since the operator is non-linear, we have to exploit other properties of the diffusion process to recover a spectral property concerning the ``second eigenvalue'' of the resulting Laplacian.  Moreover, we show that higher order spectral properties cannot hold in general using the current framework.

We consider a stochastic diffusion process, in which each vertex also experiences Brownian noise from outside the system.  We show a relationship between the second eigenvalue and the convergence behavior of the process.

We show that various hypergraph parameters like multi-way expansion and diameter
can be bounded using this operator's spectral properties.
Since higher order spectral properties do not hold for
the Laplacian operator, we instead use the concept of
procedural minimizers to consider higher order Cheeger-like inequalities.
For any $k \in \N$, we give a polynomial time algorithm to compute 
an $O(\log r)$-approximation to the $k$-th procedural minimizer,
where $r$ is the maximum cardinality of a hyperedge. We show that this approximation factor is optimal under the \sse~ hypothesis (introduced by \cite{rs10}) 
for constant values of $k$.

Moreover, using the factor preserving reduction from vertex expansion in graphs to hypergraph expansion,
we show that all our results for hypergraphs extend to vertex expansion in graphs.

\end{abstract}

%% file: intro.tex
\section{Introduction}
\label{sec:intro}

There is a rich spectral theory of graphs, based on studying the eigenvalues and
eigenvectors of the adjacency and other related matrices of graphs
\cite{am85,a86,ac88,abs10,lrtv11,lrtv12,lot12}.
We refer the reader to \cite{chung97,mt06} for a comprehensive survey on Spectral Graph Theory.
A fundamental graph parameter is its  expansion or conductance defined for a graph $G = (V,E)$ as:
\[  \phi_G \defeq \min_{S \subset V} \frac{ \Abs{\partial S} }{  \min \set{ \vol(S), \vol(\bar{S}) } },    \]
where by $\vol(S)$ we denote the sum of degrees of the vertices in $S$, and
$\partial S$ is the set of edges in the cut induced by $S$.
Cheeger's inequality \cite{am85,a86}, a central inequality in Spectral Graph Theory, 
establishes a bound on expansion via the spectrum of the graph:
\[  \frac{\lambda_2}{2} \leq \phi_G \leq \sqrt{2 \lambda_2}, \]
where $\lambda_2$ is the second smallest eigenvalue of the normalized 
Laplacian matrix $\cL_G \defeq \W^{-1/2} (\W- A) \W^{-1/2}$,
and $A$ is the adjacency matrix of the graph and $\W$ is the diagonal matrix whose
$(i,i)$-th entry is the degree of vertex $i$.
This theorem and its many (minor) variants have played a major role in the design of
algorithms as well as in understanding the limits of computation \cite{js89,ss96,d07,arv09,abs10}.
We refer the reader to \cite{hlw06} for a comprehensive survey.

Edge expansion can be generalized to edge-weighted hypergraphs.
In a hypergraph $H = (V,E)$, an edge $e \in E$ is a non-empty subset of $V$.
The edges have non-negative weights indicated by 
$w : E \rightarrow \R_+$.
We say that $H$ is an $r$-graph (or $r$-uniform) if every edge contains exactly $r$ vertices.  (Hence,
a normal graph is a 2-graph.)
Each vertex $v \in V$ has weight $w_v := \sum_{e \in E: v \in e} w_e$. A subset $S$ of vertices has
weight $w(S) := \sum_{v \in S} w_v$, and the edges it cuts is $\partial S := \{e \in E : $ $e$
intersects both $S$ and $V \setminus S \}$.
The \emph{edge expansion} of $S \subset V$
is defined as $\phi(S) := \frac{w(\partial S)}{w(S)}$.  The expansion
of $H$ is defined as:

\begin{equation}
\phi_H := \min_{\eset \subsetneq S \subsetneq V} \max\{\phi(S), \phi(V \setminus S)\}.
\label{eq:hyper_exp}
\end{equation}

It has remained open to define a spectral model of hypergraphs, whose spectra 
can be used to estimate hypergraph parameters. 
Hypergraph expansion
and related hypergraph partitioning problems are of 
immense practical importance, having applications in parallel and distributed computing \cite{ca99}, VLSI circuit 
design and computer architecture \cite{kaks99,gglp00}, scientific computing \cite{dbh06} and other areas.
Inspite of this, hypergraph expansion problems haven't been studied as well as their graph counterparts
(see Section~\ref{sec:hyper-related} for a brief survey). 
Spectral graph partitioning
algorithms are widely used in practice for their efficiency and the high quality of solutions that they often
provide \cite{bs94,hl95}. 
Besides being of natural theoretical interest, a spectral theory of hypergraphs might also 
be relevant for practical applications.

The various spectral models for hypergraphs considered in the literature haven't been without shortcomings.
An important reason for this is that there is no canonical matrix representation 
of hypergraphs.
For an $r$-uniform hypergraph $H=(V,E)$ on the vertex set $V$ and having edge set $E \subseteq {V \choose r}$,
one can define the canonical $r$-tensor form $A$ as follows:
\[ A_{(i_1, \ldots, i_r)} \defeq \begin{cases} 1 & \set{i_1, \ldots, i_r} \in E \\ 0 & \textrm{otherwise}   \end{cases} \mper \]

This tensor form and its minor variants have been explored in the literature
(see Section~\ref{sec:hyper-related} for a brief survey),
but have not been understood very well.
Optimizing over tensors is NP-hard \cite{hl13}; even getting good approximations  
might be intractable \cite{bv09}.
Moreover,
the spectral properties of tensors seem to be unrelated to combinatorial properties
of hypergraphs (See Appendix~\ref{app:hyper-tensor}).

Another way to study a hypergraph, say $H = (V,E)$, is to replace each hyperedge $e \in E$ by a complete $2$-graph
or a low degree expander on the vertices of $e$ to obtain a $2$-graph $G = (V,E')$. If we let $r$ denote the size of the largest hyperedge
in $E$, then it is easy to see that the combinatorial properties  of $G$ and $H$, 
like min-cut, sparsest-cut, among others,  could be separated by a factor of $\Omega(r)$. 
Therefore, this approach will not be useful when $r$ is {\em large}.

In general, one cannot hope to have a linear operator for hypergraphs whose spectra captures hypergraph expansion in 
a Cheeger-like manner. 
This is because the existence of such an operator will imply the existence of a polynomial
time algorithm obtaining a $\bigo{\sqrt{\opt}}$ bound on hypergraph expansion, but we rule this out
by giving a lower bound of $\Omega(\sqrt{\opt\, \log r})$ for computing hypergraph expansion,
where $r$ is the size of the largest hyperedge 
(Theorem~\ref{thm:hyper-expansion-hardness-informal}).

Our main contribution is the definition of a new Laplacian operator for hypergraphs, obtained by generalizing the 
random-walk operator on graphs. 
Our operator does not require the hypergraph to be uniform (i.e. does not require all the hyperedges
to have the same size). 
We describe this operator in Section~\ref{sec:laplacian}   
(see also Figure~\ref{fig:hyper_diffusion}).
We present our main results about this hypergraph operator in Section~\ref{sec:laplacian} and  Section~\ref{sec:cheeger}.
Most of our results are independent of~$r$ (the size of the hyperedges), some of our bounds have a logarithmic
dependence on~$r$, and none of our bounds have a polynomial dependence on~$r$. 
All our bounds are generalizations of the corresponding bounds for $2$-graphs.

\subsection{Related Work}
\label{sec:hyper-related}
Freidman and Wigderson \cite{fw95} studied the canonical tensors of hypergraphs. They bounded the
second eigenvalue of such tensors for hypergraphs drawn randomly from various distributions 
and showed their connections to randomness dispersers. 
Rodriguez \cite{r09} studied the eigenvalues of a graph obtained by replacing each hyperedge 
by a clique (Note that this step incurs a loss of $\bigO(r^2)$, where $r$ is the size of the hyperedge).
Cooper and Dutle \cite{cd12} studied the roots of the characteristic polynomial of hypergraphs
and related it to its chromatic number. 
\cite{hq13,hq14} also studied the canonical tensor form of the hypergraph and related
its eigenvectors to some configured components of that hypergraph.
Lenz and Mubayi \cite{lm12,lm15,lm13b} related the eigenvector corresponding to the second largest eigenvalue of the canonical
tensor to hypergraph quasi-randomness.
Chung \cite{c93} defined a notion of Laplacian for hypergraphs and studied the relationship
between its eigenvalues and a very different notion of hypergraph cuts and homologies. 
\cite{prt12,pr12,p13,kkl14,skm14} studied the relation of simplicial complexes with
rather different notions of Laplacian forms,
and considered isoperimetric inequalities, homologies and mixing times.
Ene and Nguyen \cite{en14} studied the hypergraph multiway partition problem (generalizing the graph multiway partition problem)
and gave a $\frac{4}{3}$-approximation algorithm for 3-uniform hypergraphs.
Concurrent to this work, 
\cite{lm14b} gave approximation algorithms for hypergraph expansion, and more generally, hypergraph small set
expansion; they gave an $\tbigo{k\sqrt{\log n}}$-approximation algorithm and an $\tbigo{k\sqrt{\OPT\, \log r}}$
approximation bound for the problem of computing the set of vertices of size at most $\Abs{V}/k$ in a hypergraph 
$H = (V,E)$,  having the least expansion.

Bobkov, Houdr\'e and Tetali \cite{bht00} defined a Poincair\'e-type functional graph parameter called 
$\linf$ and showed that it relates to the vertex expansion of a graph in a Cheeger-like manner, i.e.
it satisfies $\frac{\linf}{2} \leq \phiv = \bigo{\sqrt{\linf}}$ where $\phiv$ is the vertex expansion of the 
graph (see Section~\ref{sec:vertex-results} for the definition of vertex expansion of a graph).
\cite{lrv13} gave an $\bigo{\sqrt{\OPT \log d}}$-approximation bound for computing the vertex
expansion in graphs having the largest vertex degree $d$. 
Feige \etal \cite{fhl08} gave an $\bigo{\sqrt{\log n}}$-approximation algorithm for computing the vertex expansion
of graphs (having arbitrary vertex degrees).

Peres \etal \cite{pssw09} study a ``tug of war'' Laplacian operator on graphs that is 
similar to our hypergraph \markov operator and use it to prove that 
every bounded real-valued Lipschitz function $F$ on a subset $Y$ of a length space $X$ admits a 
unique absolutely minimal extension to $X$. Subsequently a variant of this operator was used 
for analyzing the rate of convergence of local dynamics in bargaining 
networks \cite{cdp10}. 
\cite{lrtv11,lrtv12,lot12,lm14} study higher eigenvalues of graph Laplacians
and relate them to graph multi-partitioning parameters (see Section~\ref{sec:higher-cheeger}).

\section{Notation}
\label{sec:hyper-notation}

Recall that we consider an edge-weighted hypergraph $H = (V,E,w)$, where $V$ is the vertex set,
$E$ is the set of hyperedges and 
$w : E \to \R_+$ gives the edge weights.
We let $n := |V|$ and $m := |E|$.
The weight of a vertex $v \in V$ is
$w_v \defeq \sum_{e \in E: v \in e} w(e)$.
Without loss of generality, we assume that all vertices
have positive weights, since any vertex with zero weight can be removed.
We use $\R^V$ to denote the set of column vectors.
Given $f \in \R^V$,
we use $f_u$ or $f(u)$ (if we need to use the subscript to distinguish between different vectors) to indicate the coordinate
corresponding to $u \in V$.
We use $A^\T$ to denote the transpose of a matrix $A$.
For a positive
integer $s$, we denote $[s] := \{1,2,\ldots, s\}$.

We let $\I$ denote the identity matrix and $\W \in \R^{n \times n}$ denote the diagonal matrix whose
$(i,i)$-th entry is $w_i$.
We use $\rmin \defeq \min_{e \in E} \Abs{e}$ to denote the size of the smallest hyperedge
and use $\rmax \defeq \max_{e \in E} \Abs{e}$ to denote the size of the largest hyperedge.
Since, most of our bounds will only need $\rmax$, we use $r \defeq \rmax$ for brevity. 
We say that a hypergraph is {\em regular} if all its vertices have the same degree. 
We say that a hypergraph is {\em uniform} if all its hyperedges have the same cardinality.
Recall that the expansion $\phi_H$ of a hypergraph $H$ 
is defined in (\ref{eq:hyper_exp}).
We drop the subscript whenever the
hypergraph is clear from the context.

\noindent \emph{Hop-Diameter.}
A list of edges $e_1, \ldots, e_l$ such that $e_i \cap e_{i+1} \neq \eset$ for $i \in [l-1]$
is referred as a {\em path}. The length of a path is the number of edges in it. 
We say that a path $e_1, \ldots, e_l$ connects two vertices $u,v \in V$
if $u \in e_1$ and $v \in e_l$.
We say that the hypergraph is {\em connected} if for each pair of vertices $u,v \in V$, there
exists a path connecting them. 
The {\em hop-diameter} of a hypergraph, denoted by $\diam(H)$, is the smallest value $l \in \N$, such that
each pair of vertices $u,v \in V$ have a path of length at most $l$ connecting them.

For an $x \in \R$, we define $x^+ \defeq \max \set{x,0}$ and $x^- \defeq \max \set{-x,0}$.
For a vector $u$, we use $\norm{u} := \norm{u}_2$ to
denote its Euclidean norm; if $\norm{u} \neq 0$,
we define $\tilde{u} \defeq \frac{u}{\norm{u}}$.
We use $\one \in \R^V$ to denote the vector having $1$ in every coordinate. 
For a vector $x \in \R^V$, we define its support as the set of coordinates at which $x$ is non-zero, i.e.
$\supp(x) \defeq \set{i : x_i \neq 0}$.
We use $\Ind{\cdot}$ to denote the indicator variable, i.e. $\Ind{\mathcal{E}}$ is equal to $1$
if event $\mathcal{E}$ occurs, and is equal to $0$ otherwise.
We use $\chi_S \in \R^V$ to denote the indicator vector of the set $S \subset V$, i.e.
\[ \chi_S(v) = \begin{cases} 1 & v \in S \\ 0 & \textrm{otherwise}  \end{cases} \mper \]

In classical spectral graph theory, the edge expansion
is related to the \emph{discrepancy ratio}, which is defined as 

\[\D_w(f):=\frac{\sum_{e\in E} \w{e} \max_{u,v\in e}{(\f{u}-\f{v})^2}}{\sum_{u \in V} \w{u} \f{u}^2},\]

\noindent for each non-zero vector $f \in \R^V$.
Note that $0 \leq \D_w(f) \leq 2$, where the upper bound 
can be achieved, say, by a complete bipartite graph with $f$ having 1's on one side and $-1$'s on the other side.
Observe that if $f = \chi_S$ is the indicator vector for a subset $S \subset V$,
then $\D_w(f) = \phi(S)$.
In this paper, we use three isomorphic spaces described as follows.
As we shall see, sometimes it is more convenient to use one space to describe the results.

\noindent \textbf{Weighted Space.} This is the space
associated with the discrepancy ratio $\D_w$ to consider
edge expansion.
For $f,g \in \R^V$, the inner product
is defined as $\langle f, g \rangle_w := f^\T \W g$,
and the associated norm is $\| f \|_w := \sqrt{\langle f, f \rangle_w}$.  We use $f \perp_w g$ to denote $\langle f, g \rangle_w = 0$.

\noindent \textbf{Normalized Space.} 
Given $f \in \R^V$ in the weighted space,
the corresponding vector in the normalized space is $x := \Wh f$.
The normalized discrepancy ratio is
$\Dc(x) := \D_w(\Wmh x) = \D_w(f)$.

In the normalized space, the usual $\ell_2$ inner product and norm are used.
Observe that if $x$ and $y$ are the corresponding
normalized vectors for $f$ and $g$ in the weighted space,
then $\langle x, y \rangle = \langle f, g \rangle_w$.

A well-known result~\cite{chung97} is that the \emph{normalized
Laplacian} for a 2-graph can be defined as
$\Lc := \I - \Wmh A \Wmh$ (where $A$ is the symmetric matrix
giving the edge weights)
such that $\Dc(x)$ coincides with the 
\emph{Rayleigh quotient} of the Laplacian defined as follows:
$$\rayc(x) := \frac{\langle x, \Lc x \rangle}{\langle x, x \rangle}.$$

\noindent \textbf{Measure Space.}  This is the space
associated with the diffusion process that we shall define later.
Given $f$ in the weighted space,
the corresponding vector in the measure space
is given by $\vp := \W f$.  
Observe that a vector in the measure space can have negative coordinates.
We do not consider inner product explicitly in this space, and 
so there is no special notation for it.
However, we use the $\ell_1$-norm, which is not
induced by an inner product.
For vectors $\vp_i = \Wh x_i$, we have
\[\sqrt{w_{\min}} \cdot \|x_1 - x_2\|_2  \leq \|\vp_1 - \vp_2\|_1 \leq \sqrt{w(V)} \cdot \|x_1 - x_2\|_2,\]

\noindent where the upper bound comes from
the Cauchy-Schwarz inequality.

In the diffusion process, we consider how $\vp$ will move in the future.
Hence, unless otherwise stated,
all derivatives considered are actually right-hand-derivatives
$\frac{d \vp(t)}{dt} := \lim_{\Delta t \ra 0^+} \frac{\vp(t + \Delta t) - \vp(t)}{\Delta t}$.

\noindent \textbf{Transformation between Different Spaces.}
We use the Roman letter $f$ for vectors in
the weighted space, $x$ for vectors in the
normalized space, and Greek letter $\vp$
for vectors in the measure space.
Observe that 
an operator defined on one space
induces operators on
the other two spaces.
For instance, if $\Lo$ is an operator defined on the measure space,
then $\Lo_w := \Wm \Lo \W$  is the corresponding
operator on the weighted space 
and $\Lc := \Wmh \Lo \Wh$ is the one on the normalized space.
Moreover, all three operators have the same eigenvalues.
Recall that the Rayleigh quotients are defined
as $\ray_w(f) := \frac{\langle f, \Lo_w f \rangle_w}{\langle f, f \rangle_w}$
and $\rayc(x) := \frac{\langle x, \Lc x \rangle}{\langle x, x \rangle}$. For $\Wh f = x$, we have $\ray_w(f) = \rayc(x)$.

Given a set $S$ of vectors in the normalized
space, $\Pi_S$ is the orthogonal projection operator
onto the subspace spanned by $S$.  The orthogonal
projection operator $\Pi^w_S$ can also be defined for
the weighted space.

%% file: overview.tex
\section{Overview of Results}
\label{sec:hyper-defs}

A major contribution of this paper
is to define a hypergraph Laplacian operator $\Lc$
whose spectral properties are related to the
expansion properties of the underlying hypergraph.

\subsection{Laplacian and Diffusion Process}
\label{sec:lap_overview}

In order to gain insights on how to define the Laplacian
for hypergraphs,
we first illustrate that the Laplacian for $2$-graphs
can be related to a diffusion process.
Suppose edge weights $w$ of a $2$-graph
are given by the (symmetric) matrix $A$.

Suppose $\vp \in \R^V$ is some measure on the vertices,
which, for instance, can represent a probability distribution
on the vertices.  A random walk on the graph
can be characterized by the transition matrix $\M := A \Wm$.
Observe that each column of $\M$ sums to 1,
because we apply $\M$ to the column vector $\vp$
to get the distribution $\M \vp$ after one step of the random walk.

We wish to define a continuous diffusion process.
Observe that, at this moment, the measure vector $\vp$ is moving
in the direction of $\M \vp - \vp = (\M - \I) \vp$.
Therefore, if we define an operator $\Lo := \I - \M$
on the measure space, we have
the differential equation $\frac{d \vp}{d t} = - \Lo \vp$.


Using the transformation into
the weighted space $f = \Wm \vp$
and the normalized space $x = \Wmh \vp$,
we can define the corresponding operators
$\Lo_w := \Wm \Lo \W = \I - \Wm A$
and $\Lc := \Wmh \Lo \Wh = \I - \Wmh A \Wmh$,
which is exactly the normalized Laplacian for $2$-graphs.
In the literature, the (weighted) Laplacian is defined
as $\W - A$, which is $\W \Lo_w$ in our notation.  Hence,
to avoid confusion, we only consider the normalized Laplacian in this paper.

\noindent \emph{Interpreting the Diffusion Process.}
In the above diffusion process,  we
consider more carefully the
rate of change for the measure at a certain vertex $u$:

\begin{equation}
\frac{d \vp_u}{d t} = \sum_{v : \{u,v\} \in E} w_{uv} (f_v - f_u),
\label{eq:flow}
\end{equation} 

\noindent where $f = \Wm \vp$ is the weighted measure.
Observe that for a stationary distribution of the random walk,
the measure at a vertex $u$ should be proportional
to its (weighted) degree $w_u$.
Hence, given an edge $e = \{u,v\}$, 
equation (\ref{eq:flow}) indicates that
there should be a contribution of measure flowing from the vertex with higher $f$ value
to the vertex with smaller $f$ value.
Moreover, this contribution has rate
given by $c_e := w_e \cdot |f_u - f_v|$.

\noindent \emph{Generalizing Diffusion Rule to Hypergraphs.}
Suppose in a hypergraph $H = (V,E,w)$
the vertices have measure $\vp \in \R^V$ (corresponding to $f = \Wm \vp$).
For $e \in E$,
we define $I_e(f) \subseteq e$ as the vertices $u$
in $e$ whose $f_u = \frac{\vp_u}{w_u}$ values are minimum, $S_e(f) \subseteq e$ as those whose corresponding values
are maximum, and $\Delta_e(f) := \max_{u,v \in E} (f_u - f_v)$
as the discrepancy within edge $e$.  Then,
inspired from the case of $2$-graphs,
the diffusion process should satisfy the following rules.

\begin{compactitem}
\item[\textsf{(R1)}] When the measure distribution is at state $\vp$ (where $f = \Wm \vp$), there can be a positive rate of measure flow
from $u$ to $v$ due to edge $e \in E$ only if  $u \in S_e(f)$ and $v \in I_e(f)$.
\item[\textsf{(R2)}] For every edge $e \in E$,
the total rate of measure flow \textbf{due to $e$} from vertices
in $S_e(f)$ to $I_e(f)$ is $c_e := w_e \cdot \Delta_e(f)$.
\end{compactitem}

We shall later
elaborate how the rate $c_e$ of flow due to edge $e$ is distributed
among the pairs in $S_e(f) \times I_e(f)$.  Figure~\ref{fig:hyper_diffusion}
summarizes this framework.

\begin{figure}[H]
\begin{tabularx}{\columnwidth}{|X|}
\hline

\vspace{5pt}

Given a hypergraph $H=(V,E,w)$, 
we define the (normalized) Laplacian operator
as follows.  Suppose $x \in \R^V$ is in the
normalized space with the corresponding
$\vp := \Wh x$ in the measure space
and $f := \Wm \vp$ in the weighted space.

\begin{enumerate}
	\item \label{step:step1} For each hyperedge $e \in E$,
	let $I_e(f) \subseteq e$ be the set of vertices $u$
in $e$ whose $f_u = \frac{\vp_u}{w_u}$ values are minimum
and $S_e(f) \subseteq e$ be the set of vertices in $e$ whose corresponding values
are maximum. Let $\Delta_e(f) := \max_{u,v \in E} (f_u - f_v)$.

\item \label{step:step2} \emph{Weight Distribution.}
For each $e \in E$,
the weight $w_e$ is ``somehow'' distributed among
pairs in $S_e(f) \times I_e(f)$ satisfying \textsf{(R1)} and \textsf{(R2)}.
Observe that if $I_e = S_e$, then $\Delta_e = 0$, and it does not matter how
the weight $w_e$ is distributed.

For each $(u,v) \in S_e(f) \times I_e(f)$,
there exists $a^e_{uv} = a^e_{uv}(f)$
such that $\sum_{(u,v) \in S_e \times I_e} a^e_{uv} = w_e$,
and the rate of flow from $u$ to $v$ (due to $e$) is $a^e_{uv} \cdot \Delta_e$. 

For ease of notation, we let $a^e_{uv} = a^e_{vu}$.
Moreover, for other pairs $\{u',v'\}$ that
do not receive any weight from $e$, we let $a^e_{u'v'} = 0$.

\item \label{step:step3} 
The distribution of hyperedge weights induces a symmetric matrix $A_f$
as follows.
For $u \neq v$, $A_f(u,v) = a_{uv} := \sum_{e \in E} a^e_{uv} (f)$; the diagonal entries are chosen such that entries in  the row corresponding to vertex $u$ sum to $w_u$.  Observe that $A_f$ depends on $\vp$ because $f = \Wm \vp$.
\end{enumerate}

Then, the operator $\Lo(\vp) := (\I - A_f \Wm) \vp$
is defined on the measure space,
and the diffusion process is described by
$\frac{d \vp}{d t} = - \Lo \vp$.

This induces the (normalized) Laplacian $\Lc := \Wmh \Lo \Wh$,
and the operator $\Lo_w := \Wm \Lo \W$ on the weighted space.
\\
\hline 
\end{tabularx}
\caption{Defining Laplacian via Diffusion Framework}
\label{fig:hyper_diffusion}
\end{figure}

\noindent \textbf{How to distribute the weight $w_e$
in Step~(\ref{step:step2}) in Figure~\ref{fig:hyper_diffusion}?}
In order to satisfy rule \textsf{(R1)}, it turns out that
the weight cannot be distributed arbitrarily.  We show that the following
straightforward approaches will not work.

\begin{compactitem}
\item \emph{Assign the weight $w_e$ to just one pair $(u,v) \in S_e \times I_e$.}
For the case $|S_e| \geq 2$, after infinitesimal time, 
among vertices in $S_e$, only $\vp_u$ (and $f_u$) will decrease due to $e$.
This means $u$ will no longer be in $S_e$ after infinitesimal time,
and we will have to pick another vertex in $S_e$ immediately. However,
we will run into the same problem again if we try to pick another vertex
from $S_e$, and the diffusion process cannot continue.

\item \emph{Distribute the weight $w_e$ evenly among pairs
in $S_e \times I_e$.}\footnote{Through personal communication, Jingcheng Liu and Alistair Sinclair
have informed us that they also noticed that distributing the weight
of a hyperedge uniformly will not work, and discovered independently a
similar method for resolving ties.}
In Example~\ref{eg:Louis},
there is an edge $e_5 = \{a, b, c\}$
such that the vertex in $I_{e_5} = \{c\}$
receives measure from the vertices in $S_{e_5} = \{a,b\}$.
However, vertex $b$ also gives some measure to vertex 
$d$ because of the edge $e_2 = \{b, d\}$.
In the example, all vertices have the same weight.
Now, if $w_{e_5}$ is distributed evenly among $\{a,c\}$ and $\{b, c\}$,
then the measure of $a$ decreases more slowly than that of $b$
because $b$ loses extra measure due to $e_2$.
Hence, after infinitesimal time, $b$ will no longer be in $S_{e_5}$.
This means that the measure of $b$ should not have been decreased at all due to $e_5$,
contradicting the choice of distributing $w_{e_5}$ evenly.
\end{compactitem}

\noindent \textbf{What properties should the Laplacian operator have?}
Even though the weight distribution in Step~\ref{step:step2} does not
satisfy rule \textsf{(R1)}, some operator could still be defined.
The issue is whether such an operator would have any desirable properties.
In particular, the spectral properties of the Laplacian should
have be related to the expansion properties of the hypergraph.
Recall that the normalized discrepancy ratio $\Dc(x)$ is defined for
non-zero $x \in \R^V$, and is related to hypergraph edge expansion.

\begin{definition}[Procedural Minimizers]
\label{defn:proc_min}
Define $x_1 := \Wh \vec{1}$, where $\vec{1} \in \R^V$ is the all-ones vector; $\gamma_1 := \Dc(x_1) = 0$.
Suppose $\{(x_i, \gamma_i)\}_{i \in [k-1]}$ have
been constructed.  Define $\gamma_k := \min \{ \Dc(x) : \vec{0} \neq x \perp \{x_i:
i \in [k-1]\}\}$, and $x_k$ to be any such minimizer that attains $\gamma_k = \Dc(x_k)$.
\end{definition}

\noindent \emph{Properties of Laplacian in 2-graphs.}
For the case of 2-graphs, it is known that the discrepancy ratio $\Dc(x)$
coincides with the Rayleigh quotient $\rayc(x) := \frac{\langle  x , \Lc x \rangle}{\langle x, x \rangle}$
of the normalized Laplacian $\Lc$,
which can be interpreted as a symmetric matrix.
Hence, it follows that the sequence $\{ \gamma_i \}$ obtained by the procedural minimizers
also gives the eigenvalues of $\Lc$.
Observe that for a $2$-graph, 
the sequence $\{ \gamma_i \}$ is uniquely defined,
even though the minimizers $\{x_i\}$ might not be unique (even modulo scalar multiple)
in the case of repeated eigenvalues.
On the other hand, for hypergraphs, $\gamma_2$ is uniquely defined, but
we shall see in Example~\ref{ex:gamma_nonunique} that $\gamma_3$ could depend
on the choice of minimizer $x_2$.

\begin{theorem}[Diffusion Process and Laplacian]
\label{th:main1}
Given an edge-weighted hypergraph, a diffusion process
satisfying rules~\textsf{(R1)} and~\textsf{(R2)}
can be defined and uniquely induces a
normalized Laplacian $\Lc$ (that
is not necessarily linear) on the normalized
space having the following properties.

\begin{compactitem}
\item[1.] For all $\vec{0} \neq x \in \R^V$, the Rayleigh quotient $\frac{\langle x, \Lc x \rangle}{\langle x, x \rangle}$ coincides with the discrepancy ratio $\Dc(x)$.  This implies that
all eigenvalues of $\Lc$ are non-negative.

\item[2.] There is an operator $\Lo := \Wh \Lc \Wmh$ on the measure space
such that the diffusion process
can be described by the differential equation $\frac{d \vp}{d t} = - \Lo \vp$.

\item[3.] Any procedural minimizer $x_2$ attaining $\gamma_2 := \min_{\vec{0} \neq
x \perp \Wh \vec{1}} \Dc(x)$ satisfies $\Lc x_2 = \gamma_2 x_2$.
\end{compactitem}

However, there exists a hypergraph (Example~\ref{eg:gamma3}) such that
for all procedural minimizers $\{x_1, x_2\}$,
any procedural minimizer $x_3$ attaining
$\gamma_3 := \min_{\vec{0} \neq
x \perp \{x_1, x_2\}} \Dc(x)$ is not
an eigenvector of $\Pi_{\{x_1, x_2\}^\perp} \Lc$.
\end{theorem}

The first three statements are proved in Lemmas~\ref{lemma:ray_disc}, 
\ref{lemma:define_lap} and Theorem~\ref{th:hyper_lap}.
Example~\ref{eg:gamma3} suggests that the current approach cannot be generalized
to consider higher order eigenvalues of the Laplacian $\Lc$,
since any diffusion process satisfying rules~\textsf{(R1)} and~\textsf{(R2)}
uniquely determines the 
Laplacian $\Lc$.

We remark that for hypergraphs, the Laplacian $\Lc$ is non-linear.
In general,
non-linear operators can have more or fewer than $n$ eigenvalues.
Theorem~\ref{th:main1} implies that
apart from $x_1 = \Wh \vec{1}$, the Laplacian
has another eigenvector $x_2$, which is a procedural minimizer
attaining $\gamma_2$.
It is not clear if $\Lc$ has any other eigenvalues. 
We leave as an open problem the task of investigating if other eigenvalues exist.

\noindent \textbf{Diffusion Process and Steepest Descent.}  We
can interpret the above diffusion process in terms of
deepest descent with respect to the
following quadratic potential function on the weighted space:
$$\mathsf{Q}_w(f) := \frac{1}{2} \sum_{e \in E} w_e \max_{u,v \in e} (f_u - f_v)^2.$$

Specifically, we can imagine a diffusion process in which
the motion is leading to a decrease in the potential function.
For $2$-graphs, one can check that in fact we have
$\frac{df}{dt} = - \Wm \nabla_f \mathsf{Q}_w(f)$.
Hence, we could try to define $\Lo_w f$ as  $\Wm \nabla_f \mathsf{Q}_w(f)$.
Indeed,
Lemma~\ref{lemma:deriv} confirms that
our diffusion process implies that $\frac{d}{dt} \mathsf{Q}_w(f) = - \norm{\Lo_w f}^2_w$.
However, because of the maximum operator in the definition of
$\mathsf{Q}_w(\cdot)$, one eventually has to consider the issue of resolving ties
in order to give a meaningful definition of $\nabla_f \mathsf{Q}_w(f)$.

\noindent \textbf{Comparison to other operators.}
One could ask if there can be a ``better'' operator?
A natural operator that one would be tempted to try is the {\em averaging} operator,
which corresponds to a diffusion process that attempts
to transfer measure between \emph{all} vertices in a hyperedge
to approach the stationary distribution.
However, for each hyperedge $e \in E$,
the averaging operator will yield information about
$\E_{i,j \in e} \paren{f_i - f_j}^2$, instead of
$\max_{i,j \in e} \paren{f_i - f_j}^2$ that is related to edge expansion.
In particular, the averaging operator will have a gap of factor $\Omega(r)$ between the hypergraph 
expansion
and the square root of its second smallest eigenvalue.

\input{diffusion_overview.tex}

\subsection{Cheeger Inequalities}
\label{sec:higher-cheeger}

We generalize the Cheeger's inequality~\cite{am85,a86} to hypergraphs.

\begin{theorem}[Hypergraph Cheeger Inequalities]
\label{thm:hyper-cheeger}
Given an edge-weighted hypergraph $H$,
its expansion $\phi_H$ is defined as in~(\ref{eq:hyper_exp}).  Then, we have
the following:
\[ \frac{\lh}{2} \leq \phi_H \leq 2 \sqrt{ \lh} \mper  \]
\end{theorem}

However, to consider higher-order Cheeger inequalities for hypergraphs,
at this moment, we cannot use the spectral properties of the Laplacian $\Lc$.
Moroever, the sequence $\{\gamma_i\}$ generated by procedural
minimizers might not be unique.  We consider the following parameters.

\noindent \textbf{Orthogonal Minimaximizers.}  
Define $\xi_k := \min_{x_1, \ldots x_k} \max_{i \in [k]} \Dc(x_i)$
and $\zeta_k := \min_{x_1, \ldots x_k}  \max \{\Dc(x) : x \in \spn\{x_1, \ldots x_k\}\}$, where the minimum is over $k$ non-zero mutually orthogonal vectors $x_1, \ldots, x_k$ in
the normalized space.  (All involved minimum and maximum 
can be attained because $\Dc$ is continuous and all vectors could be
chosen from the surface of a unit ball, which is compact.)

For 2-graphs, the three parameters $\xi_k = \gamma_k = \zeta_k$
coincide with the eigenvalues of the normalized Laplacian~$\Lc$.  
Indeed,
most proofs in the literature concerning expansion and Cheeger inequalities  (e.g.,~\cite{lot12,DBLP:conf/stoc/KwokLLGT13})
just need to use the underlying properties of $\gamma_k$, $\xi_k$ and $\zeta_k$ with
respect to the discrepancy ratio,
without explicitly using the spectral properties of the Laplacian.
However, the three parameters
can be related to one another in the following lemma,
whose proof is in Section~\ref{sec:min}.

\begin{lemma}[Comparing Discrepancy Minimizers]
\label{lemma:min}
Suppose $\{\gamma_k\}$ is some sequence produced by the procedural minimizers.
For each $k \geq 1$, $\xi_k \leq \gamma_k \leq \zeta_k \leq k \xi_k$.  
In particular, $\gamma_2 = \zeta_2$, but
it is possible that $\xi_2 < \gamma_2$.
\end{lemma}


Given a parameter $k \in \N$, the multi-way small-set expansion problem asks to compute
$k$ disjoint sets $S_1, S_2, \ldots, S_k$ that all have small expansion.
This problem has a close connection with the Unique
Games Conjecture~\cite{rs10,abs10}. 
In recent works, higher eigenvalues of Laplacians were used to bound small-set expansion in $2$-graphs~\cite{lrtv12,lot12}.  In particular, the
following result is achieved.

\begin{fact}[Higher-Order Cheeger Inequalities for 2-Graphs]
\label{thm:graph-higher-cheeger}
There exists an absolute constant $c > 0$ such that for any $2$-graph $G=(V,E,w)$  and any integer $k < \Abs{V}$, there exist
 $\Theta(k)$  non-empty disjoint sets $S_1, \ldots, S_{\floor{ck}} \subset V$ such that 
\[ \max_{i \in [ck]} \phi(S_i) = \bigo{ \sqrt{\gamma_{k} \log k } } \mper \]
Moreover, for any $k$ disjoint non-empty sets $S_1, \ldots, S_k \subset V$
\[ \max_{i \in [k]} \phi(S_i) \geq \frac{\gamma_k}{2} \mper  \]
\end{fact}


\ignore{

Moreover, it was shown that a graph's $\lambda_k$ 
is small if and only if the graph has roughly $k$ sets each having small expansion.
This fact can be viewed as a generalization of the Cheeger's inequality to higher 
eigenvalues and partitions.

}

We prove the following generalizations to hypergraphs
(see Theorems~\ref{thm:hyper-sse} 
and~\ref{thm:hyper-higher-cheeger} for formal statements).

%
%

\begin{theorem}[Small Set Expansion]
\label{thm:hyper-sse-informal}
Given hypergraph $H = (V,E,w)$ and parameter $k < \Abs{V}$, 
suppose $f_1, f_2, \ldots, f_k$ are $k$ orthonormal vectors in the weighted space
	such that $\max_{s \in [k]} \D_w(f_s) \leq \xi$.
Then,
there exists a set $S \subset V$ 
such that $\Abs{S} = \bigo{ \Abs{V}/k}$ satisfying 
\[ \phi(S) = \bigo{ \ratio \cdot \sqrt{ \xi}  } , \]
where $r$ is the size of the largest hyperedge in $E$.
\end{theorem}

\begin{theorem}[Higher-Order Cheeger Inequalities for Hypergraphs]
\label{thm:hyper-higher-cheeger-informal}
There exist absolute constants $c > 0$ such that the following holds.
Given a hypergraph $H=(V,E,w)$  and any integer $k < \Abs{V}$,
suppose $f_1, f_2, \ldots, f_k$ are $k$ orthonormal vectors in the weighted space
	such that $\max_{s \in [k]} \D_w(f_s) \leq \xi$.
Then,
 there exists
 $\Theta(k)$  non-empty disjoint sets $S_1, \ldots, S_{\floor{ck}} \subset V$ such that 
\[ \max_{i \in [ck]} \phi(S_i) = \bigo{ k^2 \log k \log \log k \cdot \sqrt{\log r} \cdot  \sqrt{\xi} } \mper \]

\noindent Moreover, for any $k$ disjoint non-empty sets $S_1, \ldots, S_k \subset V$
\[ \max_{i \in [k]} \phi(S_i) \geq \frac{\zeta_k}{2} \mper  \]

\end{theorem}

\subsection{Hardness via Vertex Expansion in $2$-Graphs}
\label{sec:vertex-results}

Given a graph $G = (V,E,w)$ having maximum vertex degree $d$ and a 
set $S \subset V$, its internal boundary $\Nin(S)$, and external boundary $\Nout(S)$
is defined as follows:

$\Nin(S) \defeq \set{v \in S : \exists u \in \bar{S} \textrm{such that } \set{u,v} \in E }$ and

$\Nout(S) \defeq \set{v \in \bar{S} : \exists u \in S \textrm{such that } \set{u,v} \in E }$.

The vertex expansion $\phiv(S)$ of a set $S$  is defined as 
\[ \phiv(S) \defeq \frac{ \Abs{\Nin(S)} + \Abs{\Nout(S)}  }{\Abs{S}} \mper \]
Vertex expansion is a fundamental graph parameter that has 
applications both as an algorithmic primitive and as a tool for proving communication lower bounds
\cite{lt80,l80,bt84,ak95,sm00}.


Bobkov \etal~ \cite{bht00} defined a Poincair\'e-type  functional graph parameter as follows.
Given an undirected graph $G = (V,E)$, denote $v \sim u$ if $\{v,u\}\in E$,
and define

\[ \linf \defeq \min_{ f \in \R^V} 
\frac{ \sum_{u \in V} \max_{v \sim u} \paren{f_u - f_v}^2 }{\sum_{u \in V} f_u^2 - \frac{1}{n} \paren{\sum_{u \in V} f_u }^2 } \mper \]

Observe that the expression to be minimized does not change
if the same constant is added to every coordinate.
Hence, without loss of generality, we can assume that
the above minimization is taken over all non-zero vectors such that $f \perp \vec{1}$.
Therefore, we can equivalently write 

\begin{equation}
\linf = \min_{0 \neq f \perp \vec{1}} \D^V(f),
\label{eq:linf}
\end{equation}

\noindent where $\D^V(\cdot)$ is the discrepancy ratio for vertex expansion:

\[ \D^V(f) := \frac{ \sum_{u \in V} \max_{v \sim u} \paren{f_u - f_v}^2 }{\sum_{u \in V} f_u^2 }.\]

If $\chi_S$ is the characteristic vector of the susbet $S$ of vertices,
then it follows that $\phiv(S) = \D^V(\chi_S)$.
We can see that there are many similarities with edge expansion,
and indeed
a Cheeger-type Inequality for vertex expansion in graphs was proved
in~\cite{bht00}.

\begin{fact}[\cite{bht00}]
For an un-weighted graph $G = (V,E)$,
\[  \frac{\linf}{2} \leq  \phiv_G \leq \sqrt{2 \linf } \mper  \]
\end{fact}

Given the similarities between vertex expansion in $2$-graphs
and hyperedge expansion, one could imagine that a diffusion process
can be defined with respect to vertex expansion in order
to construct a similar Laplacian operator, which
would have $\linf$ as an eigenvalue.  However, instead of repeating
the whole argument and analysis, we remark that
there is a well known reduction from vertex expansion in 
$2$-graphs to hyperedge expansion.

%

\begin{mybox}
\begin{reduction}~
\label{red:hyper-vert}

{\sf Input}: Undirected $2$-graph $G=(V,E)$.

{\sf Output}:We construct hypergraph $H = (V,E')$ as follows. For every vertex $v \in V$, we add the (unit-weighted)
hyperedge $\set{v} \cup \Nout(\set{v})$ to $E'$.
\end{reduction}
\end{mybox}

\begin{fact}[\cite{lm14b}]
\label{thm:hyper-vert-exp}
Given a graph $G=(V,E,w)$ of maximum degree $d$ and minimum degree $c_1 d$ (for some constant $c_1$), 
the hypergraph $H = (V,E')$ obtained from Reduction~\ref{red:hyper-vert}
has hyperedges of cardinality at most $d+1$ and,
\[ c_1 \phi_H(S) \leq  \frac{1}{d+1} \cdot \phiv_G(S) \leq \phi_H(S)  \qquad \forall S \subset V \mper \]
\end{fact}

\begin{remark}
The dependence on the degree in Fact~\ref{thm:hyper-vert-exp} is only because vertex expansion and hypergraph
expansion are normalized differently. The vertex expansion of a set $S$ is defined as the number of vertices in the boundary 
of $S$ divided by the cardinality of $S$, whereas the hypergraph expansion of a set $S$ is defined as the number 
hyperedges crossing $S$ divided by the sum of the degrees of the vertices in $S$. 
\end{remark}

Using Fact~\ref{thm:hyper-vert-exp},
we can apply our results for hypergraph edge expansion
to vertex expansion in $d$-regular 2-graphs.
In particular, we relate $\linf$ with the
parameter $\gamma_2$ associated with the hypergraph
achieved in Reduction~\ref{red:hyper-vert}.

\begin{theorem}
\label{thm:linf-eig2}
Let $G = (V,E)$ be a undirected $d$-regular $2$-graph with parameter~$\linf$, and let $H = (V,E')$ be the hypergraph obtained
in Reduction~\ref{red:hyper-vert} having parameter~$\eig_2$. Then,
\[  \frac{\eig_2}{4} \leq  \frac{\linf}{d} \leq \eig_2 \mper    \]
\end{theorem}

The computation of $\linf$ is not known to be tractable.
For graphs having maximum vertex degree $d$, 
\cite{lrv13} gave a $\bigo{\log d}$-approximation algorithm for computing $\linf$, and 
showed that there exists an absolute constant $C$ such that is $\sse$-hard 
to get better than a $C \log d$-approximation to $\linf$.
Indeed, such a hardness result implies that
the hyperedge expansion and the spectral gap $\gamma_2$ cannot be efficiently
approximated.
See Section~\ref{sec:vert-exp} for a definition of \sse~hypothesis.
Specifically, we show the following hardness results
for computing hyperedge expansion (see Theorem~\ref{thm:hyper-expansion-hardness})
and $\gamma_2$ (see Theorem~\ref{thm:hyper-eigs-lower}).

\begin{theorem}[Informal Statement]
\label{thm:hyper-expansion-hardness-informal}
Given a hypergraph $H$, it is \sse-hard to get better than an $\bigo{\sqrt{\phi_H \cdot \frac{\log r}{r}}}$ bound
on hypergraph expansion in polynomial time.
(Note that this is non-trivial only when $\phi_H \leq \frac{\log r}{r}$.)
\end{theorem}

\begin{theorem}[Informal Statement]
\label{thm:hyper-eigs-lower-informal}
When $\gamma_2 \leq \frac{1}{r}$,
it is \sse-hard to output a number $\widehat{\gamma}$ in polynomial time
such that $\gamma_2 \leq \widehat{\gamma} = \bigo{\gamma_2 \log r}$.
\end{theorem}

\ignore{

\hubert{I propose to remove the following Theorem,
as the corrected version is quite clumsy and
 it is just rephrasing Theorem~\ref{thm:hyper-expansion-hardness-informal}.}

\begin{theorem}
\label{thm:hyper-nonlinear}
There exists absolute constants $c_1, c_2 > 0$ such that the following holds.
Given a hypergraph $H=(V,E,w)$, assuming the \sse~ hypothesis, there exists no
polynomial time algorithm to return $\lambda$, such that 
\[  c_1 \lambda \leq \phi_H \leq c_2 \sqrt{\lambda}.\]
\end{theorem}

}

\subsection{Approximation Algorithms}

We do not know how to efficiently find
orthonormal vectors $f_1, f_2, \ldots, f_k$ in
the weighted space that attain $\xi_k$.
In view of Theorems~\ref{thm:hyper-sse-informal}
and~\ref{thm:hyper-higher-cheeger-informal},
we consider
approximation algorithms
to find $k$ such vectors to minimize $\max_{i \in [k]} \D_w(f_i)$.

\noindent \textbf{Approximate Procedural Minimizers.}
Our approximation algorithms are based on the following
result on finding approximate procedural minimizers.

\begin{theorem}
\label{thm:hyper-eigs-alg}
Suppose for $k \geq 2$,
$\{f_i\}_{i \in [k-1]}$ is a set of orthonormal vectors
in the weighted space, and
define $\gamma := \min \{ \D_w(f) : \vec{0} \neq f \perp_w \{f_i:
i \in [k-1]\}\}$.
Then, there is a randomized procedure that
produces a non-zero vector $f$ that
is orthogonal to $\{f_i\}_{i \in [k-1]}$
in polynomial time,
such that with high probability,
$\D_w(f) = \bigo{\gamma \log r}$, 
where $r$ is the size of the largest hyperedge.
\end{theorem}

Using the procedure in Theorem~\ref{thm:hyper-eigs-alg}
as a subroutine for generating procedural minimizers,
we can show that the resulting vectors provide
an $\bigo{k \log r}$-approximation to $\xi_k$.

\begin{theorem}[Approximating $\xi_k$]
\label{th:approx_xik_informal}
There exists a randomized polynomial time algorithm that, given a hypergraph $H = (V,E,w)$
and a parameter $k < \Abs{V}$, outputs $k$ orthonormal vectors $f_1, \ldots, f_k$
in the weighted space
such that with high probability, for each $i \in [k]$,
\[ \D_w(f_i) \leq \bigo{i \log r\,\cdot \xi_i}.  \]
\end{theorem}

\noindent \textbf{Algorithmic Applications.}
Applying Theorem~\ref{th:approx_xik_informal},
we readily have approximation algorithms
for the problems in Theorems~\ref{thm:hyper-cheeger},~\ref{thm:hyper-sse-informal}
and~\ref{thm:hyper-higher-cheeger-informal}.

\begin{corollary}[Hyperedge Expansion] 
\label{cor:hyper-sparsest-informal}
There exists a randomized polynomial time algorithm that given a hypergraph $H = (V,E,w)$,
outputs a set $S \subset V$ such that $\phi(S) = \bigo{\sqrt{\phi_H \log r}}$
with high probability, where $r$ is the size of the largest hyperedge in $E$.
\end{corollary}
We note that Corollary~\ref{cor:hyper-sparsest-informal} also follows directly from \cite{lm14b}.

Many theoretical and practical applications require multiplicative approximation
guarantees for hypergraph sparsest cut. 
In a seminal work,
Arora, Rao and Vazirani \cite{arv09} gave a $\bigo{\sqrt{\log n}}$-approximation
algorithm for the (uniform) sparsest cut problem in graphs. 
\cite{lm14b} gave a $\bigo{\sqrt{\log n }}$-approximation algorithm for hypergraph expansion.

\begin{corollary}[Small Set Expansion]
\label{cor:hyper-sse-informal}
There exists a randomized polynomial time algorithm that given hypergraph $H = (V,E,w)$ and parameter $k < \Abs{V}$, 
produces a set $S \subset V$ 
such that with high probability, $\Abs{S} = \bigo{\frac{n}{k}}$ and 
\[ \phi(S) = \bigo{ k^{1.5} \log k \log \log k \cdot \log r \cdot \sqrt{ \xi_k}  } , \]
where $r$ is the size of the largest hyperedge in $E$.
\end{corollary}

In contrast, a polynomial-time algorithm
is given in~\cite{lm14b}
that returns a subset $S$ with size $\bigo{ \frac{n}{k}}$
whose expansion is at most $\bigo{k \log k \log \log k \cdot \sqrt{\log n}}$ times
the smallest expansion over all vertex sets of size at most 
$\frac{n}{k}$.

\begin{corollary}[Multi-way Hyperedge Expansion]
\label{cor:hyper-higher-cheeger-informal}
There exist absolute constants $c, c' > 0$ such that the following holds.
There exists a randomized polynomial time algorithm that given hypergraph $H = (V,E,w)$ 
and parameter $k < \Abs{V}$,
produces
 $\Theta(k)$  non-empty disjoint sets $S_1, \ldots, S_{\floor{ck}} \subset V$ such that 
with high probability,
\[ \max_{i \in [ck]} \phi(S_i) = \bigo{ k^{2.5} \log k \log \log k \cdot {\log r} \cdot  \sqrt{\xi_k} } \mper \]
\end{corollary}

In contrast, for $2$-graphs, 
a polynomial-time
bi-criteria approximation algorithm~\cite{lm14}
outputs $(1 - \eps)k$ disjoint subsets
such that each subset has expansion at most $O_\eps(\sqrt{\log n \log k})$ times the optimal value.

\subsection{Sparsest Cut with General Demands}
\label{sec:sparest_overview}

An instance of the problem consists of a hypergraph $H = (V,E,w)$
with edge weights $w$ and
a collection $T = \{(\{s_i,t_i\}, D_i): i \in [k]\}$
of demand pairs, where each pair $\{s_i, t_i\}$ has demand $D_i$.
For a subset $S \subset V$,
its expansion with respect to $T$ is 

\[\Phi(S) := \frac{w(\partial S)}{\sum_{i \in [k]} D_i \Abs{  \chi_S(s_i) - \chi_S(t_i)}}.\]

The goal is to find $S$ to minimize $\Phi(S)$.  We denote
$\Phi_H := \min_{S \subset V} \Phi(S)$.

Arora, Lee and Naor \cite{aln05}  gave a 
$\bigo{\sqrt{\log k} \log \log k }$-approximation algorithm for the sparsest cut in 2-graphs with general demands.
We give a similar bound for the sparsest cut in hypergraphs with general demands.

\begin{theorem}
\label{thm:hyper-sparsest-nonuniform}
There exists a randomized polynomial time algorithm that given 
an instance of the hypergraph Sparsest Cut problem with hypergraph  $H=(V,E,w)$ and $k$ demand pairs in 
$T = \{(\{s_i,t_i\}, D_i): i \in [k]\}$, 
outputs a set $S \subset V$ such that with high probability,
\[ \Phi(S) \leq \bigo{\sqrt{\log k \log r} \log \log k } \Phi_H ,\]
where $r = \max_{e \in E} \Abs{e} $.
\end{theorem}

\subsubsection{Discussion}
We stress that none of our bounds have a polynomial dependence on $r$, the size of the largest hyperedge
(Theorem~\ref{thm:hyper-sse-informal} has a dependence on $\tbigo{\min \set{r,k}}$). 
In many of the practical
applications, the typical instances have $r = \Theta(n^{\alpha})$ for some $\alpha = \Omega(1)$;
in such cases having bounds of ${\sf poly}(r)$ would not be of any practical utility.
All our results generalize the corresponding results for 2-graphs.

\subsection{Organization}

We formally define the diffusion process and our Laplacian operator in
Section~\ref{sec:laplacian}.
We prove the existence of a non-trivial eigenvalue
for the Laplacian operator in Theorem~\ref{th:hyper_lap}.

In Section~\ref{sec:diffusion},
we define the stochastic diffusion process,
and prove our bounds on the mixing time (Theorem~\ref{thm:hyperwalk-upper} and Theorem~\ref{thm:hyperwalk-lower-informal}).
We define a discrete diffusion operator and give a 
bound on the hypergraph diameter (Theorem~\ref{thm:hyper-diam}) in Section~\ref{sec:hyper-diam}.

In Section~\ref{sec:cheeger},
we prove the basic hypergraph Cheeger inequality (Theorem~\ref{thm:hyper-cheeger})
and also the higher-order variants
(Theorem~\ref{thm:hyper-sse-informal} and 
Theorem~\ref{thm:hyper-higher-cheeger-informal}).

In Section~\ref{sec:vert-exp},
we explore the relationship between hyperedge expansion and vertex expansion
in 2-graphs.  Using hardness results for vertex expansion,
we prove our hardness results for computing hypergraph eigenvalues (Theorem~\ref{thm:hyper-eigs-lower-informal}) 
and for hypergraph expansion (Theorem~\ref{thm:hyper-expansion-hardness-informal}).

In Section~\ref{sec:hyper-eigs-poly-alg},
we give our approximation algorithm for procedural minimizers (Theorem~\ref{thm:hyper-eigs-alg}).
We present our algorithm for sparsest cut with general demands 
(Theorem~\ref{thm:hyper-sparsest-nonuniform}) in Section~\ref{sec:hyper-sparsestcut}.

%% file: diffusion_overview.tex
\subsection{Diffusion Processes}

The diffusion process described
in Figure~\ref{fig:hyper_diffusion}
is given by
the differential equation
$\frac{d \vp}{d t} = - \Lo \vp$,
where $\vp \in \R^V$ is in the measure space.
The diffusion process is
deterministic,
and no measure enters or leaves the system.  We believe that
it will be of independent interest to consider
the case when each vertex can experience independent noise
from outside the system, for instance, in risk management 
applications~\cite{merton1969lifetime,merton1971optimum}.  Since
the diffusion process is continuous in nature,
we consider Brownian noise.

For some $\eta \geq 0$, we assume that
the noise experienced by each vertex $u$
follows the Brownian motion whose rate of variance
is $\eta w_u$.  Then, the measure $\Phi_t \in \R^V$
of the system is an \Ito process
defined by the stochastic
differential equation $d \Phi_t = - \Lo \Phi_t \, dt + \sqrt{\eta} \cdot \Wh \, d B_t$.
For $\eta = 0$, this reduces to the deterministic diffusion process in a closed system.

We consider the transformation into the normalized
space $X_t := \Wmh \Phi_t$, and obtain the corresponding
equation 
$d X_t = - \Lc X_t \, dt + \sqrt{\eta} \, d B_t$,
where $\Lc$ is the normalized Laplacian.
Observe that the random noise
in the normalized space is spherically symmetric.

\noindent \textbf{Convergence Metric.}
Given a measure vector $\vp \in \R^V$,
denote $\vp^* := \frac{\langle \vec{1}, \vp \rangle}{w(V)} \cdot \W \vec{1}$, which is the corresponding \emph{stationary} measure vector obtained by distributing
the total measure $\sum_{u \in V} \vp_u = \langle \vec{1}, \vp \rangle$ among the vertices such that each vertex $u$ receives an amount
proportional to its weight $w_u$.

For the normalized vector $x = \Wmh \vp$,
observe that $x^* := \Wmh \vp^* = \frac{\langle \vec{1}, \vp \rangle}{w(V)} \cdot \Wh \vec{1}$
is the projection of $x$ into the 
subspace spanned by $x_1 := \Wh \vec{1}$.
We denote by $\Pi$ the orthogonal projection
operator into the subspace orthogonal to $x_1$.

Hence, given $x = \Wmh \vp$,
we have $x = x^* + \Pi x$, where $x^*$ is the stationary component
and $\Pi x$ is the transient component.
Moreover, $\vp - \vp^* = \Wh \Pi x$.

%

We derive a relationship
between $\gamma_2$ and the system's convergence behavior.


\begin{theorem}[Convergence and Spectral Gap]
\label{th:stoch_dom}
Suppose $\gamma_2 = \min_{0 \neq x \perp x_1} \rayc(x)$.
Then, in the stochastic diffusion process described above,
for each $t \ge 0$,
the random variable $\|\Pi X_t \|_2$
is stochastically dominated by 
$\|\widehat{X}_t \|_2$,
where $\widehat{X}_t$
has distribution 
$e^{-\gamma_2 t} \Pi {X}_0 + \sqrt{\frac{\eta}{2 \gamma_2} \cdot(1 - e^{- 2 \gamma_2 t})} \cdot N(0,1)^V$,
 and $N(0,1)^V$ is the standard $n$-dimensional Guassian distribution with independent coordinates.
\end{theorem}

\noindent \textbf{Mixing Time for Deterministic Diffusion Process.}
For the special case $\eta = 0$,
one can consider an initial probability measure $\vp_0 \in \R_+^V$
such that $\langle \vec{1}, \vp_0 \rangle = 1$.
We denote the the stationary distribution $\vp^* := \frac{1}{w(V)} \cdot \W \vec{1}$.
For $\delta > 0$,
the \emph{mixing time}
$\tmix{\vp_0}$ is the smallest time $\tau$ such that 
for all $t \geq \tau$,
$\normo{\vp_t - \vp^*} \leq \delta$.

\begin{theorem}[Upper Bound for Mixing Time]
\label{thm:hyperwalk-upper}
Consider the deterministic diffusion process with some
initial probability measure $\vp_0 \in \R_+^V$.
Then, for all $\delta > 0$, the mixing time
$\tmix{\vp_0} \leq \frac{1}{\gamma_2} \log \frac{1}{\delta \sqrt{\vp^*_{\min}}}$,
where $\vp^*_{\min} := \min_{u \in V} \vp^*(u)$.
\end{theorem}

Observe that for a regular hypergraph (i.e., $w_u$ is the same for all $u \in V$),
Theorem~\ref{thm:hyperwalk-upper} says that
the mixing time can be $\bigo{\log n}$.
We believe that this fact might have applications in counting/sampling problems on hypergraphs
\`{a} la MCMC (Markov chain monte carlo) algorithms on graphs.

\noindent \textbf{Towards Local Clustering Algorithms for Hypergraphs}
We believe that the hypergraph diffusion process  
has applications in computing combinatorial properties and
sampling problems in hypergraphs.
As a concrete example, we show that the diffusion process can be useful  
towards computing sets of vertices having small expansion.
We show that if the diffusion process mixes slowly, then
the hypergraph must contain a set of vertices having small expansion. 
This is analogous to the corresponding fact for graphs,
and can be used as a tool to certify an upper bound on hypergraph expansion.

\begin{theorem}
\label{thm:hyper-walk-cut}
Given a hypergraph $H = (V,E,w)$ and a probability distribution $\vp_0 : V \to [0,1]$, 
let $\vp_t$ denote the probability distribution at time $t$ according to the 
diffusion process (Figure~\ref{fig:hyper_diffusion}) and $\vp^*$ be the stationary distribution.

Let $\delta > 0$.
Suppose initially $\norm{\vp_0 - \vp^*}_1 > \delta$
and for some time $T > 0$,
$\norm{\vp_T - \vp^*}_1 > \delta$.
Then,
there exists a set $S \subset V$ such that $\vp^*(S) \leq \frac{1}{2}$
and

\[ \phi(S) \leq \bigo{\frac{1}{T} \ln \frac{\norm{\vp_0 - \vp^*}_1}{\sqrt{\vp^*_{\min}} \cdot \delta}}.\]
\end{theorem}
As in the case of graphs, 
this upper bound might be better than the guarantee obtained using an \sdp relaxation 
(\ref{cor:hyper-sparsest-informal}) in certain settings.

One could ask if the converse of the statement of Theorem~\ref{thm:hyper-walk-cut} is true, i.e.,
if the hypergraph $H=(V,E,w)$ has a ``sparse cut'', then is there a polynomial time computable probability 
distribution $\vp_0 : V \to [0,1]$ such that the diffusion process initialized with this
$\vp_0$ mixes ``slowly''? 
Theorem~\ref{thm:hyperwalk-lower-informal} shows that there exists such a distribution 
$\vp_0$, but it is not known if such a distribution can be computed in polynomial time. We leave this as 
an open problem.

\begin{theorem}[Lower bound on Mixing Time]
\label{thm:hyperwalk-lower-informal}
Given a hypergraph $H=(V,E,w)$, there exists a probability measure $\vp_0$ on $V$
such that $\normo{\vp_0 -\vp^*} \geq \frac{1}{2}$,
and for small enough $\delta$, 
\[ \tmix{\vp_0} = \Omega \paren{ \frac{1  }{\lh}  \ln \frac{\vp^*_{\min}}{\delta} } \mper \]
\end{theorem}
See Theorem~\ref{thm:hyperwalk-lower} for the formal statement of Theorem~\ref{thm:hyperwalk-lower-informal}.
We view the condition in Theorem~\ref{thm:hyperwalk-lower-informal} that the starting distribution 
$\vp_0$ satisfy $\normo{\vp_0 -\vp^*} \geq \frac{1}{2}$ as the analogue of a random walk in a graph starting from
some vertex.


\paragraph{Discretized Diffusion Operator and Hypergraph Diameter}
A well known fact about regular $2$-graphs is that the diameter
of a graph $G$ is at most 
$\bigo{\log n/ \paren{\log (1/(1 - \gamma_2)) }}$.

We define a discretized diffusion operator
$\Mo := \I - \frac{1}{2} \Lo$ on the measure space
in Section~\ref{sec:hyper-diam},
and use it to
prove an upper bound on the hop-diameter of a hypergraph.

\begin{theorem}
\label{thm:hyper-diam}
Given a hypergraph $H = (V,E,w)$, its hop-diameter $\diam(H)$ is  at most $\bigo{\frac{\log N_w}{\gamma_2}}$,
where $N_w := \max_{u \in V} \frac{w(V)}{w_u}$.
\end{theorem}


%% file: laplacian.tex
\section{Defining Diffusion Process and Laplacian for Hypergraphs}
\label{sec:laplacian}

A classical result in spectral graph theory
is that for a $2$-graph whose edge weights
are given by the adjacency matrix $A$,
the parameter $\gamma_2 := \min_{\vec{0} \neq x \perp  \Wh \vec{1}} \Dc(x)$ is an eigenvalue of the normalized
Laplacian \mbox{$\Lc := \I - \Wmh A \Wmh$}, 
where a corresponding minimizer $x_2$ is an eigenvector
of $\Lc$.
Observe that $\gamma_2$ is also an eigenvector
on the operator $\Lo_w := \I - \Wm A$ induced on the weighted space.
However, 
in the literature, the (weighted) Laplacian is defined
as $\W - A$, which is $\W \Lo_w$ in our notation.  Hence,
to avoid confusion, we only consider the normalized Laplacian in this paper.

In this section, we generalize the result to hypergraphs.
Observe that any result for the normalized space has an equivalent counterpart in the weighted space, and vice versa.

\begin{theorem}[Eigenvalue of Hypergraph Laplacian]
\label{th:hyper_lap}
For a hypergraph with edge weights $w$,
there exists a normalized Laplacian $\Lc$ such that the
normalized discrepancy ratio $\Dc(x)$ coincides
with the corresponding Rayleigh quotient $\rayc(x)$.
Moreover, 
the parameter $\gamma_2 := \min_{\vec{0} \neq x \perp  \Wh \vec{1}} \Dc(x)$ is an eigenvalue of $\Lc$,
where any minimizer $x_2$ is a corresponding eigenvector.
\end{theorem}

However, 
we show in Example~\ref{eg:gamma3}
that the above result for our Laplacian does not hold for $\gamma_3$.

\noindent \textbf{Intuition from Random Walk and Diffusion Process.}
We further elaborate the intuition described in Section~\ref{sec:lap_overview}.  
Given a $2$-graph whose edge weights
$w$ are given by the (symmetric) matrix $A$,
we illustrate the relationship between the Laplacian
and a diffusion process in an underlying measure space,
in order to gain insights on how to define the Laplacian
for hypergraphs.

Suppose $\vp \in \R^V$ is some measure on the vertices,
which, for instance, can represent a probability distribution
on the vertices.  A random walk on the graph
can be characterized by the transition matrix $\M := A \Wm$.
Observe that each column of $\M$ sums to 1,
because we apply $\M$ to the column vector $\vp$
to get the distribution $\M \vp$ after one step of the random walk.

We wish to define a continuous diffusion process.
Observe that, at this moment, the measure vector $\vp$ is moving
in the direction of $\M \vp - \vp = (\M - \I) \vp$.
Therefore, if we define an operator $\Lo := \I - \M$
on the measure space, we have
the differential equation $\frac{d \vp}{d t} = - \Lo \vp$.

To be mathematically precise, we are considering how $\vp$
will move in the future.  Hence, unless otherwise stated,
all derivatives considered are actually right-hand-derivatives
$\frac{d \vp(t)}{dt} := \lim_{\Delta t \ra 0^+} \frac{\vp(t + \Delta t) - \vp(t)}{\Delta t}$.

Using the transformation into
the weighted space $f = \Wm \vp$
and the normalized space $x = \Wmh \vp$,
we can define the corresponding operators
$\Lo_w := \Wm \Lo \W = \I - \Wm A$
and $\Lc := \Wmh \Lo \Wh = \I - \Wmh A \Wmh$,
which is exactly the normalized Laplacian for $2$-graphs.

\noindent \textbf{Generalizing the Diffusion Rule from 2-Graphs to
Hypergraphs.}  We
consider more carefully the
rate of change for the measure at a certain vertex $u$:
$\frac{d \vp_u}{d t} = \sum_{v : \{u,v\} \in E} w_{uv} (f_v - f_u)$, where $f = \Wm \vp$ is the weighted measure.
Observe that for a stationary distribution of the random walk,
the measure at a vertex $u$ should be proportional
to its (weighted) degree $w_u$.
Hence, given an edge $e = \{u,v\}$, by comparing the values $f_u$ and $f_v$, measure should move from the vertex with higher $f$ value
to the vertex with smaller $f$ value, at the rate given by $c_e := w_e \cdot |f_u - f_v|$.

To generalize this to a hypergraph $H = (V,E)$,
for $e \in E$ and measure $\vp$ (corresponding to $f = \Wm \vp$),
we define $I_e(f) \subseteq e$ as the vertices $u$
in $e$ whose $f_u = \frac{\vp_u}{w_u}$ are minimum, $S_e(f) \subseteq e$ as those whose corresponding values
are maximum, and $\Delta_e(f) := \max_{u,v \in E} (f_u - f_v)$
as the discrepancy within edge $e$.  Then,
the diffusion process obeys the following rules.

\begin{compactitem}
\item[\textsf(R1)] When the measure distribution is at state $\vp$ (where $f = \Wm \vp$), there can be a positive rate of measure flow
from $u$ to $v$ due to edge $e \in E$ only if  $u \in S_e(f)$ and $v \in I_e(f)$.
\item[\textsf(R2)] For every edge $e \in E$,
the total rate of measure flow \textbf{due to $e$} from vertices
in $S_e(f)$ to $I_e(f)$ is $c_e := w_e \cdot \Delta_e(f)$.
In other words,
the weight $w_e$ is distributed among $(u,v) \in S_e(f) \times I_e(f)$ such that 
for each such $(u,v)$,
there exists $a^e_{uv} = a^e_{uv}(f)$
such that $\sum_{(u,v) \in S_e \times I_e} a^e_{uv} = w_e$,
and
the rate of flow from $u$ to $v$ (due to $e$) is $a^e_{uv} \cdot \Delta_e$. (For ease of notation, we write $a^e_{uv} = a^e_{vu}$.)
Observe that if $I_e = S_e$, then $\Delta_e = 0$ and it does not matter how
the weight $w_e$ is distributed.
\end{compactitem}

Observe that the distribution of hyperedge weights will induce a symmetric matrix $A_f$
such that for $u \neq v$, $A_f(u,v) = a_{uv} := \sum_{e \in E} a^e_{uv} (f)$, and the diagonal entries are chosen such that entries in  the row corresponding to vertex $u$ sum to $w_u$.
Then, the operator $\Lo(\vp) := (\I - A_f \Wm) \vp$
is defined on the measure space to obtain
the differential equation $\frac{d \vp}{d t} = - \Lo \vp$.
As in the case for $2$-graph, we show in Lemma~\ref{lemma:ray_disc} that
the corresponding operator $\Lo_w$ on the weighted space
and the normalized Laplacian $\Lc$ are induced such that
$\D_w(f) = \ray_w(f)$ and $\Dc(x) = \rayc(x)$,
which hold no matter how the weight $w_e$ of hyperedge $e$
is distributed among edges in $S_e(f) \times I_e(f)$.

\begin{lemma}[Rayleigh Quotient Coincides with Discrepancy Ratio]
\label{lemma:ray_disc}
Suppose $\Lo_w$ on the weighted space is defined such that
rules \textsf{(R1)} and \textsf{(R2)} are obeyed.
Then, 
the Rayleigh quotient associated with $\Lo_w$ satisfies
that for any $f$ in the weighted space,
$\ray_w(f) = \D_w(f)$.
By considering the isomorphic normalized space,
we have for each $x$, $\rayc(x) = \Dc(x)$.
\end{lemma}

\begin{proof}
It suffices to show that $\langle f, \Lo_w f \rangle_w =
\sum_{e \in E} w_e \max_{u,v \in e} (f_u - f_v)^2$.

Recall that $\vp = \W f$,
and $\Lo_w = \I - \Wm A_f$,
where $A_f$ is chosen as above
to satisfy rules~\textsf{(R1)} and~\textsf{(R2)}.

Hence,
it follows that

$\langle f, \Lo_w f \rangle_w = f^\T (\W - A_f) f
= \sum_{uv \in {V \choose 2}} a_{uv} (f_u - f_v)^2$

$= \sum_{uv \in {V \choose 2}} \sum_{e \in E: \{uv, vu\} \cap S_e \times I_e \neq \emptyset} a^e_{uv} (f_u - f_v)^2
= \sum_{e \in E} w_e \max_{u,v \in e} (f_u - f_v)^2,
$
as required.
\end{proof}

\subsection{Defining Diffusion Process to Construct Laplacian}
\label{sec:disp}

Recall that $\vp \in \R^V$ is the measure vector, where each coordinate
contains the ``measure'' being dispersed.  Observe that
we consider a closed system here, and hence $\langle \vec{1}, \vp \rangle$ remains invariant.
To facilitate the analysis, we also consider the weighted measure $f := \Wm \vp$.

Our goal is to define a diffusion process
that obeys rules~\textsf{(R1)} and~\textsf{(R2)}.
Then, the operator on the measure space
is given by $\Lo \vp := - \frac{d \vp}{d t}$.
By observing that the weighted space is achieved
by the transformation $f = \Wm \vp$,
the operator on the weighted space is
given by $\Lo_w f := - \frac{d f}{d t}$.

In Figure~\ref{fig:define_r}, we give a procedure
that takes $f \in \R^V$ and returns $r = \frac{d f}{d t} \in \R^V$.
This defines $\Lo_w f = - r$,
and the Laplacian is induced $\Lc := \Wh \Lo_w \Wmh$
on the normalized space $x = \Wh f$.

Suppose we have the measure vector $\vp \in \R^V$
and the corresponding weighted vector $f = \Wm \vp$.
Observe that even though we call $\vp$ a measure vector,
$\vp$ can still have negative coordinates.
We shall construct a vector $r \in \R^V$ that is
supposed to be $\frac{df}{dt}$.
For $u \in V$ and $e \in E$,
let $\rho_u(e)$ be the rate of change of
the measure $\vp_u$ due to edge $e$.  Then, $\rho_u :=
\sum_{e \in E} \rho_u(e)$ gives the rate of change of $\vp_u$.

We show that $r$ and $\rho$ must satisfy certain
constraints because of
rules~\textsf{(R1)} and~\textsf{(R2)}.
Then, it suffices to show that there exists 
a unique $r \in \R^V$ that satisfies all the constraints.

First, since $\frac{df}{dt} = \Wm \frac{d \vp}{dt}$,
we have for each vertex $u \in V$,
$r_u = \frac{\rho_u}{w_u}$.

Rule \textsf{(R1)} implies the following constraint:

for $u \in V$ and $e \in E$,
$\rho_u(e) < 0$ only if $u \in S_e(f)$,
and $\rho_u(e) > 0$ only if $u \in I_e(f)$.

Rule \textsf{(R2)} implies the following constraint:

for each $e \in E$, we have
$\sum_{u \in I_e(f)} \rho_u(e) = - \sum_{u \in S_e(f)} \rho_u(e) = w_e \cdot \Delta_e(f)$.

\noindent \textbf{Construction of $A_f$.}
Observe that for each $e \in E$,
once all the $\rho_u(e)$'s are determined,
the weight $w_e$ can be distributed among
edges in $S_e \times I_e$ by considering
a simple flow problem on the complete bipartite graph,
where each $u \in S_e$ is a source with
supply $-\frac{\rho_u(e)}{\Delta_e}$,
and each $v \in I_e$ is a sink with demand
$\frac{\rho_v(e)}{\Delta_e}$.  Then, from any feasible flow,
we can set $a^e_{uv}$ to be the flow along the edge $(u,v) \in S_e \times I_e$.

\noindent \textbf{Infinitesimal Considerations.}
In the previous discussion, we argue that
if a vertex $u$ is losing measure due to edge $e$,
then it should remain in $S_e$ for infinitesimal time,
which holds only if the rate of change of $f_u$ is the maximum among vertices in $S_e$.
A similar condition should hold if the vertex
$u$ is gaining measure due to edge $e$.  This translates
to the following constraints.

\noindent Rule \textsf{(R3)} First-Order Derivative Constraints:
\begin{compactitem}
\item If $\rho_u(e) < 0$, then $r_u \geq r_v$ for all $v \in S_e$.
\item If $\rho_u(e) > 0$, then $r_u \leq r_v$ for all $v \in I_e$.
\end{compactitem}

We remark that rule \textsf{(R3)}
is only a necessary condition in order
for the diffusion process to satisfy
rule~\textsf{(R1)}.  Even though $A_f$ might not be unique,
we shall show that these rules are sufficient
to define a unique $r \in \R^V$, which is returned
by the procedure in Figure~\ref{fig:define_r}.

Moreover, observe that if $f = \alpha g$ for some $\alpha > 0$, then
in the above flow problem to determine the symmetric matrix, we can still have $A_f = A_g$.
Hence, even though the resulting
$\Lo_w(f) := (\I - \Wm A_f) f$ might not be linear,
we still have \mbox{$\Lo_w (\alpha g) = \alpha \Lo_w(g)$}.

\begin{figure}[H]
\begin{tabularx}{\columnwidth}{|X|}
\hline

\vspace{5pt}

Given a hypergraph $H=(V,E,w)$
and a vector $f \in \R^V$ in the weighted
space, we describe a procedure to return
$r \in \R^V$ that is supposed to be $r = \frac{df}{dt}$
in the diffusion process.

\begin{enumerate}
	\item  Define an equivalence relation on $V$
	such that $u$ and $v$ are in the same equivalence class \emph{iff} $f_u = f_v$.  
	
	\item We consider
	each such equivalence class $U \subset V$
	and define the $r$ values for vertices in $U$.
	Denote $E_U := \{e \in E: \exists u \in U, u \in I_e \cup S_e\}$.

	Recall that $c_e := w_e \cdot \max_{u,v \in E} (f_u - f_v)$.  For $F \subset E$, denote $c(F) := \sum_{e \in F} c_e$.
	
	For $X \subset U$,
	define $I_X := \{e \in E_U : I_e \subseteq X\}$
	and $S_X := \{e \in E_U: S_e \cap X \neq \emptyset\}$.

	Denote $C(X) := c(I_X) - c(S_X)$
	and $\delta(X) := \frac{C(X)}{w(X)}$.

\item Find any $P \subset U$ such that $\delta(P)$ is maximized.

For all $u \in P$, set $r_u := \delta(P)$.

\item Recursively, find the $r$ values
for the remaining points $U' := U \setminus P$ using $E_{U'} := E_U \setminus (I_P \cup S_P)$.
\end{enumerate}

\\
\hline 
\end{tabularx}
\caption{Determining the Vector $r = \frac{df}{dt}$}
\label{fig:define_r}
\end{figure}

\noindent \textbf{Uniqueness of Procedure.}
In step~(3) of Figure~\ref{fig:define_r},
there could be more than one choice of $P$
to maximize $\delta(P)$.  
In Section~\ref{sec:densest}, we give
an efficient algorithm to find such a $P$.
Moreover, we shall show that
the procedure will return the same $r \in \R^V$ no
matter what choice the algorithm makes.
In Lemma~\ref{lemma:define_lap},
we prove that rules \textsf{(R1)}-\textsf{(R3)}
imply that $\frac{df}{dt}$ must equal to such an $r$.

\subsection{A Densest Subset Problem}
\label{sec:densest}

In step~(3) of Figure~\ref{fig:define_r},
we are solving the following variant of the densest subset
problem restricted to some set $U$ of vertices,
with multi-sets $I := \{e \cap U: e \in E, I_e(f) \cap U \neq \emptyset\}$
and $S := \{e \cap U: e \in E, S_e(f) \cap U \neq \emptyset\}$.

\begin{definition}[Densest Subset Problem]
The input is a hypergraph $H_U = (U, I \cup S)$,
where we allow multi-hyperedges in
$I \cup S$.  Each $v \in U$ has weight $w_v > 0$,
and each $e \in I \cup S$ has value $c_e > 0$.

For $X \subset U$,
define $I_X := \{e \in I: e \subset X\}$
and $S_X := \{e \in S: e \cap X \neq \emptyset\}$.

The output is a non-empty $P \subset U$
such that $\delta(P) := \frac{c(I_P) - c(S_P)}{w(P)}$ is maximized,
and we call such $P$ a \emph{densest subset}.
\end{definition}

We use an LP similar to the one given by 
Charikar~\cite{Charikar00} used for the basic densest subset problem.

\begin{equation*}
\begin{array}{ll@{}ll}
\text{maximize}  & c(x) := \sum_{e \in I}{c_e x_e} - \sum_{e \in S}{c_e x_e} \quad & &\\
\text{subject to}& \sum_{v\in U} w_v y_v  = 1 & & \\
				 & x_e \leq y_v & \forall e \in I, v \in e & \\
         & x_e \geq y_v & \forall e \in S, v \in e & \\

                 & y_v, x_e \geq 0 & \forall v\in U, e \in I \cup S &
\end{array}
\end{equation*}

We analyze this LP using a similar approach given in~\cite{BalalauBCGS15}.
Given a subset $P \subset U$,
we define the following feasible solution
$z^P = (x^P, y^P)$.

$x^P_{e} =
\begin{cases}
\frac{1}{w(P)} & \mbox{if $e \in I_P \cup S_P$} \\
0 & \mbox{otherwise.}
\end{cases}
$

$y^P_v =
\begin{cases}
\frac{1}{w(P)} & \mbox{if $v \in P$} \\
0 & \mbox{otherwise.}
\end{cases}
$

Feasibility of $z^P$ can be verified easily
and it can be checked that the objective value is $c(x^P) = \delta(P)$.

Given a feasible solution $z = (x,y)$,
we say that a non-empty $P$ is a \emph{level set} of $z$
if there exists $r > 0$ such that $P = \{v \in U: y_v \geq r \}$.

The following lemma has a proof similar to~\cite[Lemma 4.1]{BalalauBCGS15}.

\begin{lemma}
\label{lemma:convex_comb}
Suppose $z^* = (x^*, y^*)$ is an optimal (fractional) solution of the LP.
Then, every (non-empty) level set $P$ of $z^*$
is a densest set and $\delta(P) = c(x^*)$.
\end{lemma}

\begin{proof}
Suppose $z^* = (x^*,y^*)$ is an optimal solution.  
We prove the result by induction on the number $k$
of level sets of $z^*$, which is also the number of distinct non-zero values found in the coordinates of $y^*$.  
For the base case when $k=1$, $z^*$ only has one level set $P = \supp(y^*)$.
Because $\sum_{v \in U} w_v y^*_v = 1$,
it follows that we must have $z^* = z^P$, and hence
$P$ must be a densest set and the result holds for $k=1$.

For the inductive step, suppose $y^*$ has $k \geq 2$ non-zero 
distinct values in its coordinates.
Let $P := \supp(y^*)$ and $\alpha := \min\{y^*_v: v \in P\}$. Observe
$P$ is a level set of $z^*$ and $\alpha \cdot w(P) \leq  \sum_{v \in U} w_v y^*_v = 1$.
Moreover, observe that if $x^*_e > 0$, then $x^*_e \geq \alpha$.

Define $\widehat{z} = (\widehat{x}, \widehat{y})$ as follows.

$\widehat{x}_{e} =
\begin{cases}
\frac{x^*_{e} - \alpha}{1 - \alpha \cdot w(P)} & \mbox{if $x^*_{e} > 0$} \\
0 & \mbox{otherwise.}
\end{cases}
$

$\widehat{y}_v =
\begin{cases}
\frac{y^*_v - \alpha}{1 - \alpha \cdot w(P)} & \mbox{if $v \in P$} \\
0 & \mbox{otherwise.}
\end{cases}
$

Hence, $z^* = \alpha \cdot w(P) \cdot z^S + (1 - \alpha \cdot w(P)) \widehat{z}$,
and the number of level sets of $\widehat{z}$ is exactly $k-1$.
In particular,
the level sets of $z^*$ are $P$ together with those of $\widehat{z}$.

Hence, to complete
the inductive step, it suffices to show that $\widehat{z}$ is
a feasible solution to the LP.
To see why this is enough, observe that
the objective function is
linear, $c(x^*) = \alpha \cdot w(P) \cdot  c(x^P) + (1 - \alpha \cdot w(P)) \cdot c(\widehat{x})$.
Hence, if both $z^P$ and $\widehat{z}$ are feasible, then
both must be optimal.  Then, the inductive hypothesis on $\widehat{z}$ can be
used to finish the inductive step.

Hence, it remains to check the feasibility 
of $\widehat{z}$.

First, $\sum_{v \in U} w_v \widehat{y}_v = \sum_{v \in P} w_v \frac{y^*_v - \alpha}{1 - \alpha \cdot w(P)} = 1$.

Observe in the objective value,
we want to increase $x_e$ for $e \in I$ and
decrease $x_e$ for $e \in S$. Hence, the optimality of $z^*$ implies that

$x^*_{e} =
\begin{cases}
\min_{v \in e} y^*_v & \mbox{if $e \in I$} \\
\max_{v \in e} y^*_v & \mbox{if $e \in S$.}
\end{cases}
$

For $x^*_e = 0$, then $\widehat{x}_e = 0$ and the corresponding inequality
is satisfied.

Otherwise,  $x^*_e \geq \alpha$, we have

$\widehat{x}_{e} = \frac{x^*_{e} - \alpha}{1 - \alpha \cdot w(P)} =
\begin{cases}
\frac{\min_{v \in e} y^*_v - \alpha}{1 - \alpha \cdot w(P)} = \min_{v \in e} \widehat{y}_v & \mbox{if $e \in I$} \\
\frac{\max_{v \in e} y^*_v - \alpha}{1 - \alpha \cdot w(P)} = \max_{v \in e} \widehat{y}_v & \mbox{if $e \in S$.}
\end{cases}
$

Therefore,
$\widehat{z}$ is feasible and this completes the inductive step.
\end{proof}

Given two densest subsets $P_1$ and $P_2$,
it follows that
that $\frac{z^{P_1} + z^{P_2}}{2}$ is an optimal
LP solution.  Hence, by considering its level sets,
Lemma~\ref{lemma:convex_comb} implies
the following corollary.

\begin{corollary}[Properties of Densest Subsets]
\begin{enumerate}
\item Suppose $P_1$ and $P_2$ are both densest subsets.  Then,
$P_1 \cup P_2$ is also a densest subset. Moreover,
if $P_1 \cap P_2$ is non-empty, then it is also a densest subset.
\item The maximal densest subset is unique and contains
all densest subsets.
\end{enumerate}
\end{corollary}

The next two lemmas show that the procedure defined
in Figure~\ref{fig:define_r} will return the same $r \in \R^V$,
no matter which densest subset is returned in step~(3).
Lemma~\ref{lemma:densest_remain} implies that if $P$
is a maximal densest subset in the given instance,
then the procedure will assign $r$ values to the vertices 
in $P$ first and each $v \in P$ will receive $r_v := \delta(P)$.

\begin{lemma}[Remaining Instance]
\label{lemma:densest_remain}
Suppose in an instance $(U, I \cup S)$ with density function $\delta$,
some densest subset $X$ is found,
and the remaining instance $(U', I' \cup S')$
is defined with $U' := U \setminus X$,
$I' := \{e \cap U': e \in I \setminus I_X\}$, 
$S' := \{e \cap U': e \in S \setminus S_X\}$
and the corresponding density function $\delta'$.
Then, for any $Y \subset U'$, $\delta'(Y) \leq \delta(X)$,
where equality holds \emph{iff} $\delta(X \cup Y) = \delta(X)$.
\end{lemma}

\begin{proof}
Denote $\delta_M := \delta(X) = \frac{c(I_X) - c(S_X)}{w(X)}$.

Observe that
$c(I'_Y) = c(I_{X \cup Y}) - c(I_X)$
and
$c(S'_Y) = c(S_{X \cup Y}) - c(S_X)$.

Hence, we have 
$\delta'(Y) = \frac{c(I'_Y) - c(S'_Y)}{w(Y)} = \frac{\delta(X \cup Y) \cdot w(X \cup Y) - \delta_M \cdot w(X)}{w(X \cup Y) - w(X)}$.

Therefore, for each $\bowtie \, \in \{ <, = , > \}$,
we have $\delta'(Y) \bowtie \delta_M$ \emph{iff} $\delta(X \cup Y) \bowtie \delta_M$.

We next see how this implies the lemma.
For $\bowtie$ being ``$>$'', we know $\delta'(Y) > \delta(X)$ is impossible, because
this implies that $\delta(X \cup Y) > \delta(X)$, violating the
assumption that $X$ is a densest subset.

For $\bowtie$ being ``$=$'', this gives
$\delta'(Y) = \delta(X)$ \emph{iff} $\delta(X \cup Y) = \delta(X)$, as required.
\end{proof}

\begin{corollary}[Procedure in Figure~\ref{fig:define_r} is well-defined.]
\label{cor:fig_define}
The procedure defined
in Figure~\ref{fig:define_r} will return the same $r \in \R^V$,
no matter which densest subset is returned in step~(3).
In particular, if $P$
is the (unique) maximal densest subset in the given instance,
then the procedure will assign $r$ values to the vertices 
in $P$ first and each $v \in P$ will receive $r_v := \delta(P)$.
Moreover, after $P$ is removed from the instance,
the maximum density in the remaining instance is strictly less than $\delta(P)$.
\end{corollary}

\subsection{Densest Subset Procedure Defines Laplacian}

We next show that rules~\textsf{(R1)} to \textsf{(R3)}
imply that in the diffusion process, $\frac{df}{dt}$
must equal to the vector $r \in \R^V$ returned
by the procedure described in Figure~\ref{fig:define_r}.

We denote
$r_S(e) := \max_{u \in S_e} r_u$ and $r_I(e) := \min_{u \in I_e} r_u$.

\begin{lemma}[Defining Laplacian from Diffusion Process]
\label{lemma:define_lap}
Given a measure vector $\vp \in \R^V$ (and the corresponding $f = \Wm \vp$
in the weighted space),
rules~\textsf{(R1)} to \textsf{(R3)}
uniquely determine $r = \frac{df}{dt} \in \R^V$ (and $\rho = \W r$),
which can be found by the procedure
described in Figure~\ref{fig:define_r}.
This defines the operators $\Lo_w f := -r$
and \mbox{$\Lo \vp := - \W r$}.
The normalized Laplacian is also induced
$\Lc := \Wmh \Lo \Wh$.

Moreover,
$\sum_{e \in E} c_e (r_I(e) - r_S(e)) = \sum_{u \in V} \rho_u r_u = \|r\|^2_w$.
\end{lemma}

\begin{proof}
As in Figure~\ref{fig:define_r},
we consider each equivalence class $U$,
where all vertices in a class have the same $f$ values.

For each such equivalence
class $U \subset V$,
define $I_U := \{e \in E: \exists u \in U, u \in I_e\}$
and $S_U := \{e \in E: \exists u \in U, u \in S_e\}$.
Notice that each $e$ is in exactly one such $I$'s and one
such $S$'s.

As remarked in Section~\ref{sec:disp},
for each $e \in E$, once all $\rho_u(e)$ is defined
for all $u \in S_e \cup I_e$,
it is simple to determine $a^e_{uv}$ for $(u,v) \in S_e \times I_e$ by considering a flow problem on the bipartite graph $S_e \times I_e$.  The ``uniqueness'' part of the proof
will show that $r = \frac{df}{dt}$ must be some unique value,
and the ``existence'' part of the proof shows that
this $r$ can determine the $\rho_u(e)$'s.

\noindent \textbf{Considering Each Equivalence Class $U$.}
We can consider each equivalence class $U$ independently by analyzing $r_u$ and
$\rho_u(e)$ for $u \in U$ and $e \in I_U \cup S_U$ that satisfy rules~\textsf{(R1)} to \textsf{(R3)}.

\noindent \textbf{Proof of Uniqueness.}
 We next show that rules~\textsf{(R1)} to \textsf{(R3)}
 imply that $r$ must take a unique value that
 can be found by the procedure in Figure~\ref{fig:define_r}.

For each $e \in I_U \cup S_U$,
recall that $c_e := w_e \cdot \Delta_e(f)$,
which is the rate of flow due to $e$ into $U$ (if $e \in I_U$) or
out of $U$ (if $e \in S_U$).
For $F \subseteq I_U \cup S_U$, denote $c(F) := \sum_{e \in F} c_e$.

Suppose $T$ is the set of vertices that have
the maximum $r$ values within the equivalence class,
i.e., for all $u \in T$, $r_u = \max_{v \in U} r_v$.
Observe that to satisfy rule~\textsf{(R3)}, for $e \in I_U$,
there is positive rate $c_e$ of measure flow into $T$ due to $e$
\emph{iff} $I_e \subseteq T$; otherwise, the entire rate $c_e$ will flow
into $U \setminus T$.
On the other hand, for $e \in S_U$,
if $S_e \cap T \neq \emptyset$, then there is a rate $c_e$
of flow out of $T$ due to $e$; otherwise,
the rate $c_e$ flows out of $U \setminus T$.

Based on this observation,
we define for $X \subset U$,
$I_X := \{e \in I_U: I_e \subseteq X\}$
and $S_X := \{e \in S_U: S_e \cap X \neq \emptyset\}$.
Note that these definitions are consistent with $I_U$ and $S_U$.
We denote $C(X) := c(I_X) - c(S_X)$.

To detect which vertices in $U$ should have the largest $r$ values,
we define $\delta(X) := \frac{C(X)}{w(X)}$,
which, loosely speaking, is the average weighted (with respect to $\W$) measure rate going into vertices in $X$. 
Observe that if $r$ is feasible, then the definition of $T$
implies that for all $v \in T$, $r_v = \delta(T)$.

Corollary~\ref{cor:fig_define} implies
that the procedure in Figure~\ref{fig:define_r}
will find the unique maximal densest subset $P$
with $\delta_M := \delta(P)$.

We next show that $T = P$.  Observe that for 
all edges $e \in I_P$ have $I_e \subset P$,
and hence,
there must be at least rate of $c(I_P)$ going into $P$;
similarly, there is at most rate of $c(S_P)$ going out of $P$.
Hence, we have $\sum_{u \in P} w_u r_u \geq c(I_P) - c(S_P) =  w(P) \cdot \delta(P)$.
Therefore, there exists $u \in P$ such that
$\delta(P) \leq r_u \leq \delta(T)$, where
the last inequality holds because every vertex $v \in T$ 
is supposed to have the maximum rate
 $r_v = \delta(T)$.
This implies that $\delta(T) = \delta_M$, $T \subseteq P$
and the maximum $r$ value is $\delta_M = \delta(T) = \delta(P)$.  Therefore,
the above inequality becomes 
$w(P) \cdot \delta_M \geq \sum_{u \in P} w_u r_u \geq w(P) \cdot \delta(P)$,
which means equality actually holds.  This implies that
every vertex $u \in P$ has the maximum rate $r_u = \delta_M$,
and so $T = P$.

\noindent \emph{Recursive Argument.} Hence, it follows that
the set $T$ can be uniquely identified
in Figure~\ref{fig:define_r}
as the set of vertices have maximum $r$ values,
which is also the unique maximal densest subset.
Then, the uniqueness argument can be applied
recursively for the smaller instance with
$U' := U \setminus T$, $I_{U'} := I_U \setminus I_T$,
$S_{U'} := S_U \setminus S_T$.

\noindent \textbf{Proof of Existence.}  We show
that once $T$ is identified in Figure~\ref{fig:define_r},
it is possible to
assign for each $v \in T$ and edge $e$ where $v\in I_e \cup S_e$,
the values $\rho_v(e)$ such that $\delta_M = r_v = \sum_{e} \rho_v(e)$.

Consider an arbitrary configuration $\rho$
in which edge $e \in I_T$ supplies a rate of $c_e$ to vertices in $T$,
and each edge $e \in S_T$ demands a rate of $c_e$ from vertices in $T$.
Each vertex $v \in T$ is supposed to gather a net rate of $w_v \cdot \delta_M$,
where any deviation is known as the \emph{surplus} or \emph{deficit}.

Given configuration $\rho$, define a directed graph $G_\rho$ with vertices in $T$
such that there is an arc $(u,v)$ if non-zero measure rate can be transferred from $u$ to $v$.
This can happen in one of two ways:
(i) there exists $e \in I_T$ containing both $u$ and $v$ such that
$\rho_u(e) > 0$, or
(ii) there exists $e \in S_T$ containing both $u$ and $v$ such that
$\rho_v(e) < 0$.

Hence, if there is a directed path from a vertex $u$ with non-zero surplus
to a vertex $v$ with non-zero deficit, then the surplus at vertex $u$ (and the
deficit at vertex $v$) can be decreased.

We argue that a configuration $\rho$ with minimum surplus must have zero surplus.
(Observe that the minimum can be achieved because $\rho$ comes from a compact set.)
Otherwise, suppose there is at least one vertex with positive surplus,
and let $T'$ be all the vertices that are reachable from some vertex with positive surplus in the
directed graph $G_\rho$. Hence, it follows that
for all $e \notin I_{T'}$, for all $v \in T'$, $\rho_v(e) = 0$,
and for all $e \in S_{T'}$, for all $u \notin T'$, $\rho_u(e)=0$.
This means that the rate going into $T'$ is $c(I_{T'})$ and all comes from $I_{T'}$,
and the rate going out of $T'$ is $c(S_{T'})$.  Since no vertex
in $T'$ has a deficit and at least one has positive surplus,
it follows that $\delta(T') > \delta_M$, which is a contradiction.

After we have shown that a configuration $\rho$ with zero surplus exists,
it can be found by a standard flow problem, in which each $e \in I_{T}$
has supply $c_e$, each $v \in {T}$ has demand $w_v \cdot \delta_M$,
and each $e \in S_{T}$ has demand $c_e$.  Moreover, in the flow network,
there is a directed edge $(e,v)$ if $v \in I_e$ and $(v,e)$ if $v \in S_e$.
Suppose in a feasible solution, there is a flow
with magnitude $\theta$ along a directed edge.
If the flow is in the direction $(e,v)$, then
$\rho_v(e) = \theta$; otherwise, if it is in the direction $(v,e)$,
then $\rho_v(e) = - \theta$.

\noindent \emph{Recursive Application.}
The feasibility argument
can be applied recursively to the smaller instance
defined on $(U', I_{U'}, S_{U'})$
with the corresponding density function $\delta'$.
Indeed, Corollary~\ref{cor:fig_define}
implies that
that $\delta_M' := \max_{\emptyset \neq Q \subset U'} \delta'(Q) < \delta_M$.

\vspace{10pt}

\noindent \textbf{Claim.} $\sum_{e \in E} c_e (r_I(e) - r_S(e)) = \sum_{u \in V} \rho_u r_u$.

Consider $T$ defined above with $\delta_M = \delta(T) = r_u$ for $u \in T$.

Observe that $\sum_{u \in T} \rho_u r_u = (c(I_T) - c(S_T)) \cdot \delta_M
= \sum_{e \in I_T} c_e \cdot r_I(e) - \sum_{e \in S_T} c_e \cdot r_S(e)$,
where the last equality is due to rule~\textsf{(R3)}.

Observe that every $u \in V$ will be in exactly one such $T$, and
every $e \in E$ will be accounted for exactly once in each of $I_T$ and $S_T$, ranging over all $T$'s.  Hence, summing over all $T$'s gives the result.
\end{proof}

\noindent \textbf{Comment on the Robustness of Diffusion Process.} 
Recall that in Section~\ref{sec:lap_overview},
we mention that if the weight distribution
is not carefully designed in Figure~\ref{fig:hyper_diffusion},
then the diffusion process cannot actually continue.
The following lemma implies that our diffusion process
resulting from the procedure in Figure~\ref{fig:define_r}
will be robust.

\begin{lemma}
In the diffusion process resulting from Figure~\ref{fig:define_r}
with the differential equation $\frac{df}{dt} = - \Lo_w f$,
at any time $t_0$, there exists some $\eps > 0$
such that $\frac{df}{dt}$
is continuous in $(t_0, t_0 + \eps)$.
\end{lemma}

\begin{proof}
Observe that as long as the equivalence classes
induced by $f$ do not change, then each of them act as a super vertex, and hence the diffusion process goes smoothly.

At the very instant that equivalence classes merge into some $U$,
Figure~\ref{fig:define_r} is actually used
to determine whether the vertices will stay together 
in the next moment.

An equivalence class can be split in two ways.
The first case is that the equivalence class $U$ is peeled off layer by layer
in the recursive manner described above,
because they receive different $r$ values.
In particular, the (unique) maximal densest subset $T$
is such a layer.  

The second case is more subtle, 
because it is possible that vertices within $T$ could
be split in the next moment.
For instance, there could be
a proper subset $X \subsetneq T$ whose $r$ values might be
marginally larger than the rest after infinitesimal time.  

The potential issue is that
if the vertices in $X$ go on their own, then
the vertices $X$ and also the vertices
in $T \setminus X$ might experience a sudden jump in their rate $r$, thereby nullifying the ``work'' performed in Figure~\ref{fig:define_r}

Fortunately, this cannot happen, because
if the set $X$ could go on its own, it must be the case
that $\delta_M = \delta(T) = \delta(X)$.
Corollary~\ref{cor:fig_define} states
that in this case, after $X$ is separated on its own,
then in the remaining instance,
we must still have $\delta'(T \setminus X) = \delta_M$.
Hence, the behavior of the remaining vertices
is still consistent with the $r$ value
produced in Figure~\ref{fig:define_r},
and the
$r$ value cannot suddenly jump.

Hence, we can conclude that if equivalence classes merge or split at time $t_0$,
there exists some $\eps>0$ such that $\frac{df}{dt}$
is continuous in $(t_0, t_0 + \eps)$, until the next time
equivalence classes merge or split.
\end{proof}

\subsection{Spectral Properties of Laplacian}
\label{sec:eigen}

We next consider the spectral properties
of the normalized Laplacian $\Lc$ induced
by the diffusion process defined in Section~\ref{sec:disp}.

\begin{lemma}[First-Order Derivatives]
\label{lemma:deriv}
Consider the diffusion process satisfying rules~\textsf{(R1)}
to~\textsf{(R3)} on 
the measure space with $\vp \in \R^V$, which
corresponds to $f = \Wm \vp$ in the weighted space.
Suppose 
$\Lo_w$ is the induced operator on the weighted space such that
$\frac{d f}{d t} = - \Lo_w f$.
Then, we have the following derivatives.

\begin{compactitem}
\item[1.] $\frac{d \|f\|^2_w}{dt} = - 2 \langle f, \Lo_w f \rangle_w$.
\item[2.] $\frac{d \langle f, \Lo_w f \rangle_w}{dt} 
= - 2 \|\Lo_w f \|^2_w$.
\item[3.] Suppose $\ray_w(f)$ is the Rayleigh quotient
with respect to the operator $\Lo_w$ on the weighted space.
Then, for $f \neq 0$, $\frac{d \ray_w(f)}{dt} = -\frac{2}{\|f\|^4_w} \cdot
(\|f\|^2_w \cdot \|\Lo_w f\|^2_w - \langle f , \Lo_w f \rangle^2_w) \leq 0$,
by the Cauchy-Schwarz inequality
on the $\langle \cdot , \cdot \rangle_w$ inner product, where equality
holds \emph{iff} $\Lo_w f \in \spn(f)$.
\end{compactitem}

\end{lemma}

\begin{proof}
For the first statement,
$\frac{d \|f\|_w^2}{d t} = 2 \langle f, \frac{d f}{d t} \rangle_w
= - 2 \langle f, \Lo_w f \rangle_w$.

For the second statement,
recall from Lemma~\ref{lemma:ray_disc}
that $\langle f, \Lo_w f \rangle_w = \sum_{e \in E} w_e \max_{u,v \in e} (f_u - f_v)^2$.
Moreover, recall also that
$c_e = w_e \cdot  \max_{u,v \in e} (f_u - f_v)$.
Recall that $r = \frac{df}{dt}$,
$r_S(e) = \max_{u \in S_e} r_u$
and $r_I(e) = \min_{u \in I_e} r_u$.

Hence, by the Envelope Theorem, $\frac{d \langle f, \Lo_w f \rangle_w}{dt} =
2 \sum_{e \in E} c_e \cdot (r_S(e) - r_I(e))$.
From Lemma~\ref{lemma:define_lap},
this equals $- 2 \|r\|^2_w = - 2 \|\Lo_w f\|^2_w$.

Finally, for the third statement, we have

$\frac{d}{dt} \frac{\langle f, \Lo_w f \rangle_w}{\langle f, f \rangle_w} = \frac{1}{\|f\|^4_w} (\| f \|^2_w \cdot \frac{d \langle f, \Lo_w f \rangle_w}{ dt} -  \langle f, \Lo_w f \rangle_w \cdot \frac{d \|f\|^2_w}{dt})
= -\frac{2}{\|f\|^4_w} \cdot
(\|f\|^2_w \cdot \|\Lo_w f\|^2_w - \langle f , \Lo_w f \rangle^2_w)$,
where the last equality follows from the first two statements.
\end{proof}

We next prove some properties
of the normalized Laplacian $\Lc$ with respect to orthogonal
projection in the normalized space.

\begin{lemma}[Laplacian and Orthogonal Projection]
\label{lemma:lap_proj}
Suppose $\Lc$ is the normalized Laplacian 
defined in Lemma~\ref{lemma:define_lap}.
Moreover, denote $x_1 := \Wh \vec{1}$, and
let $\Pi$ denote the orthogonal projection
into the subspace that is orthogonal to $x_1$.
Then, for all $x$, we have the following:
\begin{compactitem}
\item[1.] $\Lc(x) \perp x_1$,
\item[2.] $\langle x, \Lc x \rangle = \langle \Pi x, \Lc \Pi x \rangle$.
\item[3.] For all real numbers $\alpha$ and $\beta$,
$\Lc(\alpha x_1 + \beta x) = \beta \Lc(x)$.
\end{compactitem}
\end{lemma}

\begin{proof}
For the first statement, observe that since the diffusion process
is defined on a closed system, the total measure given by $\sum_{u \in V} \vp_u$ does not change.
Therefore, $0 = \langle \vec{1}, \frac{d \vp}{d t} \rangle = \langle \Wh \vec{1}, \frac{d x}{d t} \rangle$,
which implies that $\Lc x = - \frac{d x}{d t} \perp x_1$.

For the second statement,
observe that from Lemma~\ref{lemma:ray_disc},
we have:

$\langle x, \Lc x \rangle = \sum_{e \in E} w_e
\max_{u,v \in e} (\frac{x_u}{\sqrt{w_u}} - \frac{x_v}{\sqrt{x_v}})^2 
= \langle (x + \alpha x_1), \Lc (x + \alpha x_1) \rangle $,
where the last equality holds for all real numbers $\alpha$.
It suffices to observe that $\Pi x = x + \alpha x_1$, for some suitable real $\alpha$.

For the third statement, it is more convenient
to consider transformation into the weighted space $f = \Wmh x$.
It suffices to show that $\Lo_w (\alpha \vec{1} + \beta f) = \beta \Lo_w(f)$.
This follows immediately because
in the definition of the diffusion process,
it can be easily checked that $\Delta_e(\alpha \vec{1} + \beta f) = \beta \Delta_e(f)$.
\end{proof}

\begin{proofof}{Theorem~\ref{th:hyper_lap}}
Suppose $\Lc$ is the normalized Laplacian
induced by the diffusion process in Lemma~\ref{lemma:define_lap}.
Let $\gamma_2 := \min_{\vec{0} \neq x \perp \Wh \vec{1}} \rayc(x)$
be attained by some minimizer $x_2$.
We use the isomorphism between the three spaces:
$\Wmh \vp = x = \Wh f$.

The third statement of Lemma~\ref{lemma:deriv}
can be formulated in terms of the normalized space,
which states that $\frac{d \rayc(x)}{d t} \leq 0$,
where equality holds \emph{iff} $\Lc x \in \spn(x)$.

We claim that $\frac{d \rayc(x_2)}{d t} = 0$.
Otherwise, suppose $\frac{d \rayc(x_2)}{d t} < 0$.
From Lemma~\ref{lemma:lap_proj},
we have $\frac{dx}{dt} = - \Lc x \perp \Wh \vec{1}$.
Hence, it follows that at this moment, the current normalized
vector is at position $x_2$, and is moving
towards the direction given by
$x' := \frac{d x}{dt}|_{x=x_2}$ such that
$x' \perp \Wh \vec{1}$, and $\frac{d \rayc(x)}{d t}|_{x=x_2} < 0$.
Therefore, for sufficiently small $\eps > 0$,
it follows that $x_2' := x_2 + \eps x'$ is a non-zero vector
that is perpendicular to $\Wh \vec{1}$
and $\rayc(x_2') < \rayc(x_2) = \gamma_2$, contradicting the definition of $x_2$.

Hence, it follows that $\frac{d \rayc(x_2)}{d t} = 0$,
which implies that $\Lc x_2 \in \spn(x_2)$.
Since $\gamma_2 = \rayc(x_2) = \frac{\langle x_2, \Lc x_2 \rangle}{\langle x_2,  x_2 \rangle}$,
it follows that $\Lc x_2 = \gamma_2 x_2$, as required.
\end{proofof}

%% file: stochastic.tex
\section{Diffusion Processes}
\label{sec:diffusion}

In Section~\ref{sec:laplacian},
we define a diffusion process in a closed system with respect
to a hypergraph according to the
equation $\frac{d \vp}{d t} = - \Lo \vp$,
where $\vp \in \R^V$ is the measure vector,
and $\Lo$ is the corresponding operator on the measure space.
In this section, we consider related diffusion processes.
In the stochastic diffusion process,
on the top of the diffusion process,
each vertex is subject to independent Brownian noise.
We also consider a discretized diffusion operator,
which we use to analyze the hop-diameter of a hypergraph.

\subsection{Stochastic Diffusion Process}
\label{sec:stochastic}

We analyze the process using \Ito calculus,
and the reader can refer to the textbook by
{\O}ksendal~\cite{oksendalSDE} for relevant background.

\noindent \textbf{Randomness Model.}
We consider the standard multi-dimensional Wiener process
$\{B_t \in \R^V: t \geq 0\}$ with independent Brownian
motion on each coordinate.
Suppose the variance of the Brownian motion 
experienced by each vertex is proportional to its weight.
To be precise, there exists $\eta \geq 0$
such that for each vertex $u \in V$,
the Brownian noise introduced to $u$ till time $t$ is $\sqrt{\eta w_u} \cdot B_t(u)$,
whose variance is $\eta w_u t$.
It follows that the net amount of measure
added to the system till time $t$
is $\sum_{u \in V} \sqrt{\eta w_u} \cdot B_t(u)$,
which has normal distribution  $N(0, \eta t \cdot w(V))$.
Observe that the special case for $\eta = 0$ is just the
diffusion process in a closed system.

This random model induces an \Ito process on the measure space
given by the following stochastic differential equation:

\[d \Phi_t = - \Lo \Phi_t \, dt + \sqrt{\eta} \cdot \Wh \, d B_t, \]
with some initial
measure $\Phi_0$

By the transformation into the normalized space
$x := \Wmh \vp$, we consider the corresponding
stochastic differential equation in the normalized space:

\[d X_t = - \Lc X_t \, dt + \sqrt{\eta} \, d B_t,\]
where $\Lc$ is the normalized Laplacian from Lemma~\ref{lemma:define_lap}.  Observe that the random noise
in the normalized space is spherically symmetric.

\noindent \textbf{Convergence Metric.}
Given a measure vector $\vp \in \R^V$,
denote $\vp^* := \frac{\langle \vec{1}, \vp \rangle}{w(V)} \cdot \W \vec{1}$, which is the measure vector obtained by distributing
the total measure $\sum_{u \in V} \vp_u = \langle \vec{1}, \vp \rangle$ among the vertices such that each vertex $u$ receives an amount
proportional to its weight $w_u$.

For the normalized vector $x = \Wmh \vp$,
observe that $x^* := \Wmh \vp^* = \frac{\langle \vec{1}, \vp \rangle}{w(V)} \cdot \Wh \vec{1}$
is the projection of $x$ into the 
subspace spanned by $x_1 := \Wh \vec{1}$.
We denote by $\Pi$ the orthogonal projection
operator into the subspace orthogonal to $x_1$.

Hence, to analyze how far the measure is from being stationary,
we consider the vector $\Phi_t - \Phi^*_t$,
whose $\ell_1$-norm is
$\|\Phi_t - \Phi^*_t\|_1 \leq \sqrt{w(V)} \cdot \| \Pi X_t \|_2$.
As random noise is constantly delivered to the system,
we cannot hope to argue that these random quantities approach zero as $t$ tends to infinity.  However,
we can show that these random variables are stochastically
dominated by distributions with bounded mean and variance
as $t$ tends to infinity.  The following lemma states that a larger value 
of $\gamma_2$ implies that the measure is closer to being stationary.

\begin{lemma}[Stochastic Dominance]
\label{lemma:stoch_dom}
Suppose $\gamma_2 = \min_{0 \neq x \perp x_1} \rayc(x)$.
Then, in the stochastic diffusion process described above,
for each $t \ge 0$,
the random variable $\|\Pi X_t \|_2$
is stochastically dominated by 
$\|\widehat{X}_t \|_2$,
where $\widehat{X}_t$
has distribution 
$e^{-\gamma_2 t} \Pi {X}_0 + \sqrt{\frac{\eta}{2 \gamma_2} \cdot(1 - e^{- 2 \gamma_2 t})} \cdot N(0,1)^V$,
 and $N(0,1)^V$ is the standard $n$-dimensional Guassian distribution with independent coordinates.
\end{lemma}

\begin{proof}
Consider the function $h : \R^V \ra \R$ given by
$h(x) := \| \Pi x \|_2^2 = \| x -  x^*\|_2^2$,
where $x^* := \frac{\langle x_1, x \rangle}{w(V)} \cdot x_1$ and $x_1 := \Wh \vec{1}$.
Then, one can check that
the gradient is $\nabla h(x) = 2 \, \Pi x$,
and the Hessian is 
$\nabla^2 h(x) = 2 (\I - \frac{1}{w(V)} \cdot \Wh J \Wh)$,
where $J$ is the matrix where every entry is 1.

Define the \Ito process $Y_t := h(X_t) = \langle \Pi X_t,  \Pi X_t \rangle$.
By the \Ito's lemma,
we have 

$d Y_t = \langle \nabla h(X_t), d X_t \rangle
+ \frac{1}{2} (d X_t)^\T \nabla^2 h(X_t) \, (d X_t)$.

To simplify the above expression, we make the substitution
$d X_t = - \Lc X_t \, dt + \sqrt{\eta} \, d B_t$.
From Lemma~\ref{lemma:lap_proj},
we have for all $x$, $\Lc x \perp x_1$ and
$\langle x , \Lc x \rangle = \langle \Pi x, \Lc \Pi x \rangle$.

Moreover, the convention for the product of differentials is
$0 = dt \cdot dt = dt \cdot d B_t(u) 
=  d B_t(u) \cdot d B_t(v)$ for $u \neq v$,
and $d B_t(u)\cdot d B_t(u) = dt$.
Hence, only the diagonal entries of the Hessian are relevant.

We have
$d Y_t = - 2 \langle \Pi X_t, \Lc \Pi X_t \rangle \, dt
+ \eta \sum_{u \in V} (1 - \frac{w_u}{w(V)}) \, dt
+ 2 \sqrt{\eta} \cdot \langle \Pi X_t, dB_t \rangle$.
Observing that $\Pi X_t \perp x_1$, from the definition of
$\gamma_2$, we have 
$\langle \Pi X_t, \Lc \Pi X_t \rangle \geq 
\gamma_2 \cdot \langle \Pi X_t, \Pi X_t \rangle$.
Hence, we have the following inequality:
$d Y_t \leq - 2 \gamma_2 Y_t \, dt
+ \eta n \, dt
+ 2 \sqrt{\eta} \cdot \langle \Pi X_t, dB_t \rangle$.

We next define another \Ito process $\widehat{Y_t} := \langle \widehat{X}_t, \widehat{X}_t \rangle$ with
initial value $\widehat{X}_0 := \Pi X_0$ and
stochastic differential equation:
$d \widehat{Y}_t = - 2 \gamma_2 \widehat{Y}_t \, dt
+ \eta n \, dt
+ 2 \sqrt{\eta} \cdot \langle \widehat{X}_t, d\widehat{B}_t \rangle$.

We briefly explain why $Y_t$ is stochastically dominated by
$\widehat{Y}_t$ by using a simple coupling argument.
If $Y_t < \widehat{Y_t}$, then we can choose $d B_t$ and $d\widehat{B}_t$ to be independent.
If $Y_t = \widehat{Y_t}$, observe that
$\langle \Pi X_t, dB_t \rangle$ and 
$\langle \widehat{X}_t, d\widehat{B}_t \rangle$ have the same distribution, because both $d B_t$ and $d\widehat{B}_t$
are spherically symmetric.  Hence,
in this case, we can choose a coupling between $d B_t$ and $d\widehat{B}_t$ such that
$\langle \Pi X_t, dB_t \rangle = \langle \widehat{X}_t, d\widehat{B}_t \rangle$.

Using \Ito's lemma, one can verify that the
above stochastic differential equation can be derived from
the following equation involving $\widehat{X}_t$:
$d \widehat{X}_t = - \gamma_2 \widehat{X}_t \, dt 
+ \sqrt{\eta} \, d \widehat{B}_t$.

Because $d \widehat{B}_t$ has independent coordinates,
it follows that the equation can be solved independently
for each vertex $u$.  Again, using
the \Ito lemma,
one can verify that $d(e^{\gamma_2 t} X_t) = \sqrt{\eta} \cdot e^{\gamma_2 t}
\, d \widehat{B}_t$.
Therefore, we have the solution
$\widehat{X}_t = e^{-\gamma_2 t} \widehat{X}_0 + \sqrt{\eta} \cdot e^{- \gamma_2 t} \int_{0}^t e^{\gamma_2 s} \, d \widehat{B}_s$,
which has the same distribution 
as: 

$e^{-\gamma_2 t} \widehat{X}_0 + \sqrt{\frac{\eta}{2 \gamma_2} \cdot(1 - e^{- 2 \gamma_2 t})} \cdot N(0,1)^V$, as required.
\end{proof}

\begin{corollary}[Convergence and Laplacian]
In the stochastic diffusion process, as $t$ tends to infinity,
$\|\Phi_t - \Phi_t^*\|_1^2$ is stochastically dominated by $\frac{\eta \cdot w(V)}{2 \gamma_2} \cdot \chi^2(n)$,
where $\chi^2(n)$ is the
chi-squared distribution with $n$ degrees of freedom.
Hence, $\lim_{t \ra \infty} E[ \|\Phi_t - \Phi_t^*\|_1]
\leq \sqrt{\frac{\eta n \cdot w(V)}{2 \gamma_2}}$.
\end{corollary}

\noindent \textbf{Remark.}
Observe that the total measure introduced
into the system is $\sum_{u \in V} \sqrt{\eta w_u} \cdot B_t(u)$,
which has standard deviation $\sqrt{\eta t \cdot w(V)}$.
Hence, as $t$ increases, the ``error rate''
is at most $\sqrt{\frac{n}{2 \gamma_2 t}}$.

\begin{proof}
Observe that, as $t$ tends to infinity,
$\widehat{Y_t} = \| \widehat{X}_t \|_2^2$
converges to the distribution $\frac{\eta}{2 \gamma_2} \cdot \chi^2(n)$, where $\chi^2(n)$ is the
chi-squared distribution with $n$ degrees of freedom (having
mean $n$ and standard deviation $\sqrt{2n}$).

Finally,
observing that $\|\Phi_t - \Phi_t^*\|_1^2 \leq w(V) \cdot \|\Pi X_t\|_2^2$,
it follows that as $t$ tends to infinity,
$\|\Phi_t - \Phi_t^*\|_1^2$ is stochastically dominated by
the distribution $\frac{\eta \cdot w(V)}{2 \gamma_2} \cdot \chi^2(n)$,
which has mean $\frac{\eta n \cdot w(V)}{2 \gamma_2}$
and standard deviation $\frac{\eta \sqrt{n} \cdot w(V)}{\sqrt{2} \gamma_2}$.
\end{proof}

\begin{corollary}[Upper Bound for Mixing Time for $\eta = 0$]
\label{cor:mixing}
Consider the deterministic diffusion process with $\eta = 0$,
and some initial probability measure $\vp_0 \in \R_+^V$
such that $\langle \vec{1}, \vp_0 \rangle = 1$.
Denote $\vp^* := \frac{1}{w(V)} \cdot \W \vec{1}$,
and $\vp^*_{\min} := \min_{u \in V} \vp^*(u)$.
Then, for any $\delta > 0$ and $t \geq \frac{1}{\gamma_2} \log \frac{1}{\delta \sqrt{\vp^*_{\min}}}$, we have
$\|\Phi_t - \vp^*\|_1 \leq \delta$.
\end{corollary}

\begin{proof}
In the deterministic process with $\eta = 0$, stochastic
dominance becomes
$\|\Pi X_t \|_2 \leq e^{\gamma_2 t} \cdot \|\Pi X_0\|_2$.

Relating the norms, we have
$\|\Phi_t - \vp^*\|_1 \leq \sqrt{w(V)} \cdot \|\Pi X_t\|_2
\leq \sqrt{w(V)} \cdot e^{- \gamma_2 t} \cdot \|\Pi X_0\|_2$.

Observe that
$\|\Pi X_0\|_2^2
\leq \langle X_0, X_0 \rangle
= \langle \vp_0, \Wm \vp_0 \rangle \leq \frac{1}{\min_u w_u}$.

Hence, it follows that
$\|\Phi_t - \vp^*\|_1 \leq \frac{1}{\sqrt{\vp^*_{\min}}} \cdot e^{- \gamma_2 t}$, which is at most $\delta$,
for $t \geq \frac{1}{\gamma_2} \log \frac{1}{\delta \sqrt{\vp^*_{\min}}}$.
\end{proof}

\subsection{Bottlenecks for the Hypergraph Diffusion Process}

In this section we prove that if the hypergraph diffusion process mixes slowly, then
it must have a set of vertices having small expansion (Theorem~\ref{thm:hyper-walk-cut}).

\begin{theorem}[Restatement of Theorem~\ref{thm:hyper-walk-cut}]
Given a hypergraph $H = (V,E,w)$ and a probability distribution $\vp_0 : V \to [0,1]$, 
let $\vp_t$ denote the probability distribution at time $t$ according to the 
diffusion process (Figure~\ref{fig:hyper_diffusion}) and $\vp^*$ be the stationary distribution.

Let $\delta > 0$.
Suppose initially $\norm{\vp_0 - \vp^*}_1 > \delta$
and for some time $T > 0$,
$\norm{\vp_T - \vp^*}_1 > \delta$.
Then,
there exists a set $S \subset V$ such that $\vp^*(S) \leq \frac{1}{2}$
and

\[ \phi(S) \leq \bigo{\frac{1}{T} \ln \frac{\norm{\vp_0 - \vp^*}_1}{\sqrt{\vp^*_{\min}} \cdot \delta}}.\]
\end{theorem}

\begin{proof}
We consider the transformation
$x_t := \Wmh \vp_t$.
We denote by $\Pi$ the orthogonal projection
operator into the subspace orthogonal to $x_1 := \Wh \vec{1}$.
Consider
the projection $\widehat{x}_t := \Pi x_t$
onto the subspace orthogonal to $x_1$.
Denote $x^* := \Wmh \vp^* = \frac{1}{w(V)} \cdot \Wh \vec{1}$,
which is the projection of $x_0$ into the 
subspace spanned by $x_1 := \Wh \vec{1}$.

Observe that $x_t = x^* + \widehat{x}_t$,
where $x^*$ is the stationary component
and $\widehat{x}_t$ is the transient component.
Moreover, $\vp_t - \vp^* = \Wh \widehat{x}_t$.

The diffusion process on the measure space
induces the differential equation on $\widehat{x}_t$
as follows:

$\frac{d \widehat{x}_t}{d t} = - \Lc \widehat{x}_t$.

By expressing Lemma~\ref{lemma:deriv}~(1)
in terms of the normalized space,
we have

$\frac{d \norm{\widehat{x}_t}^2}{d t} = - 2 \rayc(\widehat{x}_t) \cdot \norm{\widehat{x}_t}^2$.

Integrating from $t=0$ to $T$ and simplifying,
we have

$\ln \frac{\norm{\widehat{x}_0}}{\norm{\widehat{x}_T}}
= \int^T_0 \rayc(\widehat{x}_t) dt \geq T \cdot \rayc(\widehat{x}_T)$,

\noindent where the last inequality holds
because $\rayc(\widehat{x}_t)$ is decreasing
according to Lemma~\ref{lemma:deriv}~(3).

Since the norms are related by
$\sqrt{w_{\min}} \cdot \|x\|_2  \leq \|\vp\|_1 \leq \sqrt{w(V)} \cdot \|x\|_2$,
we have

$\rayc(\widehat{x}_T) \leq \frac{1}{T} \ln \frac{\norm{\widehat{x}_0}}{\norm{\widehat{x}_T}}
\leq \frac{1}{T} \ln (\frac{1}{\sqrt{\vp^*_{\min}}} \cdot \frac{\norm{\vp_0 - \vp^*}_1}{\norm{\vp_T - \vp^*}_1})
\leq \frac{1}{T} \ln \frac{\norm{\vp_0 - \vp^*}_1}{\sqrt{\vp^*_{\min}} \cdot \delta}$.

Finally, observing that $\widehat{x}_T \perp x_1$,
Proposition~\ref{prop:hyper-sweep-rounding} implies that
there exists a set $S \subset V$ such that $\vp^*(S) \leq \frac{1}{2}$,
and $\phi(S) \leq \bigo{\sqrt{\rayc(\widehat{x}_T)}}
\leq \bigo{\frac{1}{T} \ln \frac{\norm{\vp_0 - \vp^*}_1}{\sqrt{\vp^*_{\min}} \cdot \delta}}$.
\end{proof}

\subsection{Lower Bounds on Mixing Time}

Next we prove Theorem~\ref{thm:hyperwalk-lower-informal}.

\begin{theorem}[Formal statement of Theorem~\ref{thm:hyperwalk-lower-informal}]
\label{thm:hyperwalk-lower}
Given a hypergraph $H=(V,E,w)$,
suppose there exists a vector $y \perp x_1$
in the normalized space such that $\rayc(y) \leq \gamma$.
Then, there exists an initial
probability distribution $\vp_0 \in \R_+^V$
in the measure space
such that $\norm{\vp_0 - \vp^*}_1 \geq \frac{1}{2}$.
Moreover, for any $\delta > 0$ and $t \leq \frac{1}{4 \gamma} \ln \frac{\sqrt{\vp^*_{\min}}}{2 \delta}$, at time $t$ of the diffusion process, we have

$\norm{\vp_t - \vp^*}_1 \geq \delta$.
\end{theorem}

We consider the diffusion process from the perspective of the
normalized space. Recall that $x_1 := \Wh \vec{1}$
is an eigenvector of the normalized Laplacian $\Lc$ with eigenvalue 0.
From Lemma~\ref{lemma:lap_proj}~(1),
$\Lc(x) \perp x_1$ for all $x \in \R^V$.
Therefore, the diffusion process has no effect on the subspace
spanned by $x_1$, and we can
focus on its orthogonal space.

\begin{lemma}
\label{lemma:lb_diff_norm}
Suppose $y \in \R^V$ is a non-zero vector in the
normalized space such that $y \perp x_1$ and
$\rayc(y) = \gamma$.
If we start the diffusion process with $y_0 := y$,
then after time $t \geq 0$,
we have $\norm{y_t}_2 \geq e^{- \gamma t} \cdot \norm{y_0}_2$.
\end{lemma}

\begin{proof}
By Lemma~\ref{lemma:deriv}~(1) interpreted for the normalized space,
we have

$\frac{d \norm{y_t}^2}{dt} = - 2 \rayc(y_t) \cdot \norm{y_t}^2
\geq - 2 \gamma \cdot \norm{y_t}^2$,
where the last inequality holds
because from Lemma~\ref{lemma:deriv}~(3),
$t \mapsto \rayc(y_t)$ is a decreasing funtion,
which implies that $\rayc(y_t) \leq \rayc(y_0) = \gamma$.

Integrating the above gives

$\norm{y_t}^2 \geq e^{-2 \gamma t} \cdot \norm{y_0}^2$.
\end{proof}

The next lemma shows that
given a vector in the normalized space
that is orthogonal to $x_1$,
a corresponding probability distribution
in the measure space that has large distance
from the stationary distribution $\vp^* := \frac{\W \vec{1}}{w(V)}$ can be constructed.

\begin{lemma}
\label{lemma:init_prob_dist}
Suppose $y \in \R^V$ is a non-zero vector in the
normalized space such that $y \perp x_1$ and
$\rayc(y) = \gamma$.  Then,
there exists $\widehat{y} \perp x_1$
such that $\rayc(\widehat{y}) \leq  4\gamma$
and $\vp_0 := \vp^* + \Wh \widehat{y}$ is a probability distribution 
(i.e., $\vp_0 \geq 0$), and $\norm{\Wh \widehat{y}}_1 \geq \frac{1}{2}$.
\end{lemma}

\begin{proof}
One could try to consider $\vp^* +  \Wh (\alpha y)$ for some $\alpha \in \R$,
but the issue is that to ensure that every coordinate is non-negative,
the scalar $\alpha$ might need to have very small magnitude,
leading to a very small $\norm{\Wh (\alpha y)}_1$.

We construct the desired vector in several steps.  We first
consider $z := y + c x_1$ for an appropriate scalar $c \in \R$
such that both $w(\supp(z^+))$ and $w(\supp(z^-))$ are at most $\frac{1}{2} \cdot w(V)$,
where $z^+$ is obtained from $z$ by keeping only the positive coordinates,
and $z^-$ is obtained similarly from the negative coordinates. Observe that
we have $z = z^+ + z^-$.

We use $\Pi$ to denote the projection operator into the space orthogonal
to $x_1$ in the normalized space.
Then, we have $y = \Pi z = \Pi z^+ + \Pi z^-$.
Without loss of generality, by replacing $z$ with $-z$,
we can assume that $\norm{\Pi z^+} \geq \frac{1}{2} \norm{y}$.

Observe that $\langle \Pi z^+, \Lc \Pi z^+ \rangle
= \langle z^+, \Lc z^+ \rangle \leq \langle z, \Lc z \rangle
= \langle y, \Lc y \rangle$,

\noindent where the middle inequality follows because
$\langle z, \Lc z \rangle = \sum_{e \in E} w_e \max_{u,v \in e} (\frac{z_u}{\sqrt{w_u}} - \frac{z_v}{\sqrt{w_v}})^2$.

Hence, we have $\rayc(\Pi z^+) \leq 4 \rayc(y)$,
and we consider an appropriate
scaled vector $\widehat{y} :=  \Pi \widehat{z}$,
where $\widehat{z} = c z^+$ for some $c > 0$ such that
$\langle \vec{1}, \Wh \widehat{z} \rangle = 1$.

Hence, it follows that
$\widehat{y} = \widehat{z} - \frac{\langle \Wh \vec{1}, \widehat{z} \rangle}{w(V)} \cdot \Wh \vec{1}$,
which implies that $\Wh \widehat{y} = \Wh \widehat{z} - \vp^*$.

Therefore, we have $\vp_0 := \vp^* + \Wh \widehat{y} = \Wh \widehat{z} \geq 0$.

Moreover,
$\norm{\Wh \widehat{y}}_1 \geq \langle \vec{1}, \Wh \widehat{z} \rangle - \frac{w(\supp(z^+))}{w(V)} + \frac{w(\supp(z^-))}{w(V)} \geq \frac{1}{2}$,
where the last inequality follows from $w(\supp(z^+)) \leq \frac{1}{2} w(V)$.
\end{proof} 

\begin{proofof}{Theorem~\ref{thm:hyperwalk-lower}}
Using Lemma~\ref{lemma:init_prob_dist},
we can construct $\widehat{y}$ from $y$
such that $\widehat{y} \perp x_1$
and $\rayc(\widehat{y}) \leq 4 \gamma$.

Then, we can define the initial probability distribution
$\vp_0 := \vp^* + \Wh \widehat{y}$ in the measure space
with the corresponding $y_0 := \widehat{y}$ vector in
the normalized subspace orthogonal to~$x_1$.

By Lemma~\ref{lemma:lb_diff_norm},
at time $t$ of the difffusion process, we have
$\norm{y_t}_2 \geq e^{- 4 \gamma t} \cdot \norm{y_0}_2$.

Relating the norms of the measure space and the normalized space,
we have

$\norm{\vp_t - \vp^*}_1 \geq \sqrt{w_{\min}} \cdot \norm{y_t}_2
\geq \sqrt{w_{\min}} \cdot e^{- 4 \gamma t} \cdot \norm{y_0}_2
\geq \sqrt{\vp^*_{\min}} \cdot e^{- 4 \gamma t} \cdot \norm{\vp_0 - \vp^*}_1
\geq \sqrt{\vp^*_{\min}} \cdot e^{- 4 \gamma t} \cdot \frac{1}{2}$.

Hence, for $t \leq \frac{1}{4 \gamma} \ln \frac{\sqrt{\vp^*_{\min}}}{2 \delta}$, we have
$\norm{\vp_t - \vp^*}_1 \geq \delta$, as required.
\end{proofof}

\begin{remark}
Observe that we do not know how to efficiently find $x_2 \perp x_1$
to attain $\rayc(x_2) = \gamma_2$.  However, 
the approximation algorithm in Theorem~\ref{thm:hyper-eigs-alg_formal}
allows us to efficiently compute some $y$ such that
$\rayc(y) \leq O(\log r) \cdot \gamma_2$.

Hence, we can  compute a probability distribution $\vp_0$ in polynomial time such 
\[ \normo{\vp_0 - \vp^*} \geq \frac{1}{2} \qquad \textrm{and} \qquad \tmix{\vp_0} \geq 
\Omega(\frac{1}{\gamma_2 \log r} \log \frac{\vp^*_{\min}}{\delta}).
\]
\end{remark}

%% file: diameter.tex
\subsection{Hypergraph Diameter }
\label{sec:hyper-diam}
In this section we prove Theorem~\ref{thm:hyper-diam}.

\begin{theorem}[Restatement of Theorem~\ref{thm:hyper-diam}]
Given a hypergraph $H = (V,E,w)$, 
its hop-diameter is 
\[ \diam(H) 
= \bigo{ \frac{\log N_w}{\eig_2} },\]
where $N_w := \max_{u \in V} \frac{w(V)}{w_u}$ and $\eig_2$
is the eigenvalue of the normalized Laplacian as
defined in Theorem~\ref{th:hyper_lap}.
\end{theorem}


We start by defining the notion of discretized
diffusion operator.

\begin{definition}[Discretized Diffusion Operator]
Recalling that a diffusion process in the measure space
is defined in Section~\ref{sec:disp}
by $\frac{d \vp}{d t} = -\Lo \vp$,
we define a discretized diffusion operator on
the measure space by
$\Mo := \I - \frac{1}{2} \cdot \Lo$.

Moreover, using the isomorphism between
the measure space and the normalized space,
we define the corresponding operator
on the normalized space $\Mc 
:= \I - \frac{1}{2} \cdot \Lc$. 
%
\end{definition}

When we consider the diffusion process,
it is more convenient to think
in terms of the measure space.  However,
the normalized space is more convenient
for considering orthogonality.

Next, we bound the norm of the discretized diffusion operator.
\begin{lemma}
\label{lem:hyper-higher-norms}
For a vector $x$ in the
normalized space such that $x \perp x_1 := \Wh \vec{1}$,
we have $\norm{\Mc x}_2 \leq \sqrt{1 - \frac{\gamma_2}{2}}
\cdot \norm{x}_2$.
%
%
%
\end{lemma}

\begin{proof}
Fix $x \perp x_1 := \Wh \vec{1}$.
Observe that $\Mc x = \widehat{M} x$
for some symmetric matrix
$\widehat{M} := \I - \frac{1}{2} \cdot \widehat{L}$,
where the matrix
$\widehat{L}$ depends on $x$
and has the form 
$\widehat{L} := \I - \Wmh \widehat{A} \Wmh$.
The precise definition of $\widehat{A}$
(depending on $x$) is given in Section~\ref{sec:disp},
but the important property
is that $\widehat{A}$ is a non-negative symmetric
matrix such that sum of entries
in row $u$ is $w_u$.

Standard spectral graph theory and linear algebra
state that $\R^V$ has a basis
consisting of orthonormal eigenvectors
$\{v_1, v_2, \ldots, v_n\}$ of $\widehat{L}$,
whose eigenvalues are in $[0,2]$.
Hence, the matrix $\widehat{M}$
has the same eigenvectors;
suppose the eigenvalue of $v_i$ is $\lambda_i \in [0,1]$.

We write $x := \sum_{i=1}^n c_i v_i$ for some real $c_i$'s.
Then,
we have $\norm{\Mc x}_2^2 = \sum_i \lambda_i^2 c_i^2
\leq \sum_i \lambda_i c_i^2 = \langle x, \Mc x \rangle
= \langle x,  x \rangle - \frac{1}{2} \langle x, \Lc x \rangle
\leq (1 - \frac{\gamma_2}{2}) \norm{x}_2^2$,

where the last
inequality follows from 
$\langle x, \Lc x \rangle \geq \gamma_2 \norm{x}_2^2$,
because of the definition of $\gamma_2$ and $x \perp x_1$.

Hence, the result follows.
%
%
%
%
%
\end{proof}

\begin{proofof}{Theorem~\ref{thm:hyper-diam}}
The high level idea is based on the following observation.
Suppose $S$ is the support
of a non-negative vector $\vp$ in the measure space.
Then, applying the discretized diffusion operator $\Mo$
to $\vp$ has the effect of spreading the measure
on $S$ to vertices that are within one hop from $S$,
where two vertices $u$ and $v$ are within one hop
from each other if there is an edge $e$
that contains
both $u$ and $v$.

Therefore, to prove that a hypergraph
has hop-diameter at most $l$,
it suffices to show that,
starting from a measure vector $\vp$ whose support
consists of only one vertex,
applying the operator $\Mo$ to $\vp$ for $l$ times
spreads the support to all vertices.
Since we consider orthogonal projection,
it will be more convenient
to perform the calculation in the normalized space.

Given a vertex $u \in V$,
denote $\chi_u \in \R^V$
as the corresponding characteristic unit vector
in the normalized space.  
The goal is to show that if $l$ is large enough,
then for all vertices $u$ and $v$,
we have $\langle \chi_u, \Mc^l (\chi_v) \rangle > 0$.

We use $\Pi$ to denote the projection operator
into the subspace that is orthogonal to $x_1 := \Wh \vec{1}$.
Then, we have
$\chi_u = \frac{\sqrt{w_u}}{w(V)} \cdot x_1 + \Pi \chi_u$.

Lemma~\ref{lemma:lap_proj}
implies that
for all $x$, $\Mc(x) \perp x_1$,
and for all real $\alpha$, $\Mc(\alpha x_1 + x) = \alpha x_1 + \Mc(x)$.

Hence, we have
$\langle \chi_u, \Mc^l \chi_v \rangle
= \frac{\sqrt{w_u w_v}}{w(V)} + \langle \Pi \chi_u, \Mc^l (\Pi \chi_v) \rangle$.
Observe that the first term $\frac{\sqrt{w_u w_v}}{w(V)} \geq \frac{1}{N_w}$,
where $N_w := \max_{u \in V} \frac{w(V)}{w_u}$.

For the second term,
we have 
$\langle \Pi \chi_u, \Mc^l (\Pi \chi_v) \rangle
\leq \norm{\Pi \chi_u}_2 \cdot \norm{\Mc^l (\Pi \chi_v)}_2
\leq (1 - \frac{\gamma_2}{2})^{l/2}$,
where the first inequality
follows from Cauchy-Schwartz
and the second inequality follows
from applying Lemma~\ref{lem:hyper-higher-norms} for $l$ times.

Hence, for $l$ larger than $\frac{2 \log N_w}{\log \frac{1}{1 - \frac{\gamma_2}{2}}}
= \bigo{\frac{\log N_w}{\gamma_2}}$,
we have $\langle \Pi \chi_u, \Mc^l (\Pi \chi_v) \rangle > 0$, as required.

\end{proofof}

%% file: cheeger.tex
\section{Cheeger Inequalities for Hypergraphs}
\label{sec:cheeger}

In this section, we generalize the Cheeger inequalities
to hypergraphs.  For the basic version,
we relate the expansion of a hypergraph
with the eigenvalue $\gamma_2$ of
the Laplacian $\Lc$ defined in Section~\ref{sec:laplacian}.
However, at the moment, we cannot exploit
the higher order spectral properties of $\Lc$.
Instead, we achieve higher order Cheeger inequalities
in terms of the orthogonal minimaximizers defined in
Section~\ref{sec:higher-cheeger}.

\subsection{Basic Cheeger Inequalities for Hypergraphs}
\label{sec:basic_cheeger}

We prove the basic Cheeger inequalities for hypergraphs.

\begin{theorem}[Restatement of Theorem~\ref{thm:hyper-cheeger}]
\label{thm2:hyper-cheeger}
Given an edge-weighted hypergraph $H$, we have:
\[ \frac{\lh}{2} \leq \phi_H \leq \lh +  2 \sqrt{\frac{\lh}{\rmin}}
 \leq 2 \sqrt{ \lh},\]
where $\phi_H$ is the hypergraph expansion
and $\gamma_2$ is the eignenvalue of $\Lc$ as in
Theorem~\ref{th:hyper_lap}.
\end{theorem}
Towards proving this theorem, 
we first show that a {\em good} line-embedding of the hypergraph  
suffices to upper bound the expansion.

\begin{proposition}
\label{prop:hyper-1d}
Let $H = (V,E,w)$ be a hypergraph with edge weights $w : E \to \R^+$  
 and let $f \in \R_+^{V}$ be a non-zero vector. 
Then, there exists a set $S \subseteq {\sf supp}(f)$ such that 
\[ \phi(S) \leq \frac{\sum_{e \in E} w_e \max_{u,v \in e} \Abs{f_u - f_v} }{  \sum_u w_u f_u } . \]
\end{proposition}

\begin{proof}
The proof is similar to the proof of the 
corresponding statement for vertex expansion in graphs~\cite{lrv13}.
Observe that in the result, the upper bound on the right hand side
does not change if $f$ is multiplied by a positive scalar.
Hence, we can assume, without loss of generality, 
that $f \in [0,1]^V$.

We define a family of functions $\set{F_r : [0,1] \to \set{0,1} }_{r \in [0,1]}$ as follows.
\[ F_r(x) = \begin{cases} 1 & x \geq r \\ 0 & \textrm{otherwise}     \end{cases}. \]

For $r \geq 0$ and
a vector $f \in [0,1]^V$,
we consider the induced vector
$F_r(f) \in \{0,1\}^V$, whose coordinate corresponding
to $v$ is $F_r(f_v)$.
Let $S_r$ denote the support of the vector $F_r(f)$.
For any $a \in [0,1]$ we have
\begin{equation}
\label{eq:hyper-1d-helper1}
 \int_0^1 F_r(a)\, \dr  = a \mper 
\end{equation}

Now, observe that if $a - b \geq 0$, then $F_r(a) - F_r(b) \geq 0, \forall r \in [0,1]$; similarly,
if $a - b \leq 0$ then  $F_r(a) - F_r(b) \leq 0, \forall r \in [0,1]$. Therefore, 
\begin{equation}
\label{eq:hyper-1d-helper2}
 \int_0^1 \Abs{F_r(a) - F_r(b)} \dr = \Abs{ \int_0^1 F_r(a) \dr - \int_0^1 F_r(b) \dr }  =  \Abs{a-b} \mper 
\end{equation}
Also, for a hyperedge $e$, if $u = \arg\max_{u \in e} f_u$ and $v = \arg\min_{u \in e} f_u$, then 
\begin{equation}
\label{eq:hyper-1d-helper3}
 \Abs{F_r(f_u) - F_r(f_v)} \geq \Abs{ F_r(f_{u'}) - F_r(f_{v'}) }, \quad \forall r \in [0,1] \textrm{ and }   \forall u', v' \in e \mper 
\end{equation}

\noindent Therefore, we have
\begin{align*}
\frac{\int_0^1  \sum_e w_e \max_{u,v \in e} \Abs{F_r(f_u) - F_r(f_v)} \dr }{ \int_0^1 \sum_u w_u F_r(f_u) \dr }
& =   \frac{ \sum_e w_e \max_{u,v \in e} \int_0^1 \Abs{F_r(f_u) - F_r(f_v)} \dr }{ \int_0^1 \sum_u w_u F_r(f_u) \dr } 
	& \textrm{(Using \ref{eq:hyper-1d-helper3})}	\\
& =   \frac{ \sum_e w_e \max_{u,v \in e} \Abs{ \int_0^1 F_r(f_u) \dr - \int_0^1 F_r(f_v) \dr}}{  \sum_u w_u \int_0^1 F_r(f_u) \dr}  
	& \textrm{(Using \ref{eq:hyper-1d-helper2})}	\\
& =   \frac{ \sum_e w_e \max_{u,v \in e} \Abs{ f_u - f_v} }{  \sum_u w_u f_u }. 
	& \textrm{(Using \ref{eq:hyper-1d-helper1})}
\end{align*}

Therefore, there exists $r' \in [0,1]$ such that 
\[ \frac{  \sum_e w_e \max_{u,v \in e} \Abs{F_{r'}(f_u) - F_{r'}(f_v)} }{  \sum_u w_u F_{r'}(f_u) } 
 \leq \frac{ \sum_e w_e \max_{u,v \in e} \Abs{ f_u - f_v} }{  \sum_u w_u f_u} \mper  \]
 
\noindent Since $F_{r'}(\cdot)$ takes value in $\set{0,1}$, we have 
\[ \frac{  \sum_e w_e \max_{u,v \in e} \Abs{F_{r'}(f_u) - F_{r'}(f_v)} }{  \sum_{u \in V} w_u F_{r'}(f_u) }
 = \frac{  \sum_e w_e \cdot \Ind{ e \textrm{ is cut by } S_{r'} } }{ \sum_{u \in S_{r'}} w_u}  =  \phi(S_{r'}) \mper \] 
Therefore, 
\[ \phi(S_{r'}) \leq  \frac{ \sum_e w_e \max_{u,v \in e} \Abs{ f_u - f_v}}{  \sum_u w_u f_u} 
\qquad \textrm{and} \qquad S_{r'} \subseteq \supp(f) \mper   \]

\hfill
\end{proof}

\begin{proposition}
\label{prop:hyper-sweep-rounding}
Given an edge-weighted hypergraph $H = (V,E,w)$ and a non-zero vector $f \in \R^{V}$ such that $f \perp_w \vec{1}$, 
there exists a set $S \subset V$ such that $w(S) \leq \frac{w(V)}{2}$ and
\[ \phi(S) \leq \D_w(f) +  2 \sqrt{\frac{ \D_w(f)}{\rmin} } , \]

\noindent where $\D_w(f)=\frac{\sum_{e\in E} \w{e} \max_{u,v\in e}{(\f{u}-\f{v})^2}}{\sum_{u \in V} \w{u} \f{u}^2}$ and $\rmin = \min_{e \in E} |e|$.
\end{proposition}

\begin{proof}
Let $g = f + c \vec{1}$ for an appropriate $c \in \R$ such that both
$w(\supp(g^+))$ and $w(\supp(g^-))$ are at most $\frac{w(V)}{2}$.
For instance, sort the coordinates of $f$
such that $f(v_1) \leq f(v_2) \leq \cdots \leq f(v_n)$
and pick $c = f(v_i)$, where $i$ is the smallest
index such that $\sum_{j=1}^i w(v_j) \geq \frac{w(V)}{2}$.

Since $f \perp_w \vec{1}$,
it follows that $\langle g, \vec{1} \rangle_w = c \langle \vec{1}, \vec{1} \rangle_w$.
Hence,
we have $\langle f, f \rangle_w = 
\langle g, g \rangle_w - 2 c \langle g, \vec{1} \rangle_w +
c^2 \langle \vec{1}, \vec{1} \rangle_w
= \langle g, g \rangle_w -
c^2 \langle \vec{1}, \vec{1} \rangle_w
\leq \langle g, g \rangle_w$.

Therefore,  we have
\[ \D_w(f) =  \frac{ \sum_{e \in E} w_e \max_{u,v \in e} (g_u - g_v)^2 }{\langle f, f \rangle_w }
\geq \frac{ \sum_{e \in E} w_e \max_{u,v \in e} (g_u - g_v)^2 }{ \langle g, g \rangle_w }
= \D_w(g) \mper  \]

For any $a,b \in R$, we have 
\[ (a^+ - b^+)^2 + (a^- - b^-)^2 \leq  (a - b)^2  .\]

Therefore, we have 
\begin{eqnarray*}
\D_w(f)  & \geq & \D_w(g) = \frac{ \sum_{e \in E} w_e \max_{u,v \in e} (g_u - g_v)^2 }{ \sum_u w_u g_u^2 } \\ 
	& \geq & \frac{ \left( \sum_{e \in E} w_e \max_{u,v \in e} (g_u^+ - g_v^+)^2  \right)
	+ \left( \sum_{e \in E} w_e \max_{u,v \in e} (g_u^- - g_v^-)^2  \right) }{ \sum_u w_u (g_u^+)^2   + \sum_u w_u (g_u^-)^2 } \\ 
 & \geq & \min \set{ \frac{ \sum_{e \in E} w_e \max_{u,v \in e} (g_u^+ - g_v^+)^2 }{\sum_u w_u (g_u^+)^2} , 
	\frac{ \sum_{e \in E} w_e\max_{u,v \in e} (g_u^- - g_v^-)^2 }{\sum_u w_u (g_u^-)^2} } \\
	& = & \min \set{\D_w(g^+), \D_w(g^-)}.
\end{eqnarray*}

Let $h \in \set{g^+, g^-}$ be the vector corresponding the minimum in the previous inequality. Then, we have
\begin{align*}
 \sum_{e \in E} w_e \max_{u,v \in e} \Abs{h_u^2 - h_v^2} 
& = \sum_{e \in E} w_e \max_{u,v \in e} \Abs{h_u - h_v}(h_u + h_v) \\
& = \sum_{e \in E} w_e \max_{u,v \in e} (h_u - h_v)^2 +2 \sum_{e \in E} w_e \min_{u \in e} h_u \max_{u,v \in e} \Abs{h_u - h_v} \\  
& \leq \sum_{e \in E} w_e \max_{u,v \in e} (h_u - h_v)^2 
		+2 \sqrt{ \sum_{e \in E} w_e \max_{u,v \in e} (h_u - h_v)^2} \sqrt{\sum_{e \in E} w_e 
		\cdot \frac{ \sum_{u \in e} h_u^2}{\rmin} } \\
& = \sum_{e \in E} w_e \max_{u,v \in e} (h_u - h_v)^2
+ 2 \sqrt{ \sum_{e \in E} w_e \max_{u,v \in e} (h_u - h_v)^2} \sqrt{ \frac{ \sum_{u \in V} w_u h_u^2}{\rmin} },
\end{align*}

\noindent  where the inequality follows from the Cauchy-Schwarz's Inequality.

Using $\D_w(h) \leq \D_w(f)$,
\[ \frac{\sum_{e \in E} w_e \max_{u,v \in e} \Abs{h_u^2 - h_v^2}}{ \sum_u w_u h_u^2} 
\leq \D_w(h) + 2 \sqrt{ \frac{  \D_w(h)}{\rmin} } \leq  \D_w(f) +  2 \sqrt{\frac{ \D_w(f)}{\rmin} } \mper  \] 
Invoking Proposition~\ref{prop:hyper-1d} with vector $h^2$, we get that there exists a set $S \subset \supp \paren{h}$ such that
\[ \phi(S) \leq \D_w(f) + 2 \sqrt{\frac{ \D_w(f)}{\rmin} }  \qquad \textrm{and } \qquad  
w(S) \leq w(\supp \paren{h}) \leq \frac{w(V)}{2} \mper \]
\end{proof}

The ``hypergraph orthogonal separators'' construction due to \cite{lm14b}
can also be used to prove Proposition~\ref{prop:hyper-sweep-rounding}, albeit with a much larger
absolute constant in the bound on the expansion of the set $S$.

We are now ready to prove Theorem~\ref{thm2:hyper-cheeger}.

\begin{proofof}{Theorem~\ref{thm2:hyper-cheeger} (and \ref{thm:hyper-cheeger})}~
\begin{enumerate}
\item 
Let $S \subset V$ be any set such that $w(S) \leq \frac{w(V)}{2}$, and let $g \in \{0,1\}^V$ be the indicator vector 
of $S$. Let $f$ be the component of $g$ orthogonal to $\vec{1}$ (in the weighted space).
Then, $g = f + c \vec{1}$,
where $c = \frac{\langle g, \vec{1} \rangle_w}{\langle \vec{1}, \vec{1} \rangle_w}
= \frac{w(S)}{w(V)}$.

Moreover, as in the proof of Proposition~\ref{prop:hyper-sweep-rounding},
we have
$\langle f, f \rangle_w = \langle g, g \rangle_w - c^2 \langle \vec{1}, \vec{1} \rangle_w
= w(S) \cdot (1 - \frac{w(S)}{w(V)}) \geq \frac{w(S)}{2}$.

Then, since $g \neq \vec{1}$,
we have $0 \neq f \perp_{w} \vec{1}$ and
so we have 
\begin{align*}
\lh & \leq \D_w(f) 
 = \frac{ \sum_e w_e \max_{u,v \in e} (g_u - g_v)^2  }{ \langle f, f \rangle_w } \\
  & \leq \frac{ w(\partial S) }{ w(S)/2} = 2 \phi(S).
  \end{align*}
	
Since the choice of the set $S$ was arbitrary, we have $\frac{\lh}{2} \leq \phi_H$.

\item 
Invoking Proposition~\ref{prop:hyper-sweep-rounding} with the minimizer $h_2$ such that $\eig_2 = \D_w(h_2)$, we get that
$\phi_H \leq \lh +  2 \sqrt{\frac{\lh}{\rmin}}$.

For $\lh \leq \frac{1}{4}$, we observe that $\rmin \geq 2$ and
have $\phi_H \leq (\frac{1}{2} + \sqrt{2}) \cdot \sqrt{\lh}
\leq 2 \sqrt{\lh}$;
for $\lh > \frac{1}{4}$,
observe that we have $\phi_H \leq 1 \leq 2 \sqrt{\lh}$.

We remark that the constant 2 in the upper bound can be improved slightly by optimizing the threshold for $\lh$ in the above case analysis, and further considering cases
whether $\rmin = 2$ or $\rmin \geq 3$.
\end{enumerate}
\end{proofof}

\subsection{Higher Order Orthogonal Minimaximizers}
\label{sec:min}

As mentioned in Section~\ref{sec:higher-cheeger},
we do not yet know about higher order spectral properties
of the Laplacian $\Lc$.
Hence, to achieve results like higher order Cheeger-like inequalities,
we consider
the notion
of orthogonal minimaximizers
with respect to the discrepancy ratio.

In Section~\ref{sec:higher-cheeger},
the parameters $\xi_k$ and $\zeta_k$ are defined
in terms of the normalized space.
We can equivalently define them
in terms of the weighted space as
$\xi_k := \min_{f_1, \ldots, f_k}
\max_{i \in [k]} \D_w(f_i)$
and
$\zeta_k := \min_{f_1, \ldots, f_k} 
\max \{\D_w(f) : f \in \spn\{f_1, \ldots, f_k\}\}$,
where the minimum is over $k$ non-zero mutually
orthogonal vectors $f_1, f_2, \ldots, f_k$ in the weighted space.
The proofs shall work with either the normalized or the weighted space,
depending on which is more convenient.

We do not know an efficient method to find $k$ orthonormal vectors that achieve $\xi_k$ or $\zeta_k$.
In Section~\ref{sec:hyper-eigs-poly-alg},
we describe how approximations of these vectors can be obtained.

We prove Lemma~\ref{lemma:min} that compares
the parameters $\gamma_k$, $\xi_k$ and $\zeta_k$ by the following claims.

\begin{claim}
\label{claim:xg}
For $k \geq 1$, $\xi_k \leq \gamma_k$.
\end{claim}

\begin{proof}
Suppose the procedure produces $\{\gamma_i: i \in [k]\}$,
which is attained by orthonormal vectors $X_k := \{x_i : i \in [k]\}$
in the normalized space.
Observe that $\max_{i \in [k]} \Dc(x_i) = \Dc(x_k) = \gamma_k$,
since $x_k$ could have been a candidate in the minimum for defining $\gamma_i$ because $x_k \perp x_j$, for all $j \in [k-1]$.

Since $X_k$ is a candidate for taking the minimum
over sets of $k$ orthonormal vectors in the definition of
$\xi_k$, it follows that $\xi_k \leq \gamma_k$.
\end{proof}

\begin{claim}
\label{claim:gz1}
For $k \geq 1$, $\gamma_k \leq \zeta_k$.
\end{claim}

\begin{proof}
For $k=1$, $\gamma_1 = \zeta_1 = 0$.

For $k > 1$, suppose the $\{\gamma_i : i \in [k-1]\}$
have already been constructed with
the corresponding orthonormal minimizers $X_{k-1} := \{x_i : i \in [k-1]\}$.

Let $Y_k := \{y_i: i \in [k]\}$ be an arbitrary set of $k$ orthonormal
vectors.  Since the subspace orthogonal to $X_{k-1}$
has rank $n - k + 1$ and the span of $Y_k$
has rank $k$, there must be a non-zero $y \in \spn(Y_k) \cap X_{k-1}^\perp$.

Hence, it follows that $\gamma_k = \min_{\vec{0} \neq x \in X_{k-1}^\perp}
\Dc(x) \leq \max_{y \in \spn(Y_k)} \Dc(y)$.
Since this holds for any set $Y_k$ of $k$ orthonormal vectors,
the result follows.
\end{proof}

\begin{claim}
\label{claim:zx}
Given any $k$ orthogonal vectors $\{f_i : i \in [k]\}$ in the weighted space. We have,
$$
\zeta_k \leq k \max_{i\in[k]} \D_w(f_i).
$$
Moreover, if the $f_i$'s have disjoint support, we have
$$
\zeta_k \leq 2 \max_{i\in[k]} \D_w(f_i).
$$
\end{claim}

\begin{proof}
Here it will be convenient to consider
the equivalent discrepancy ratios for the
weighted space.

It suffices to show that
for any $h \in \spn(\{f_i : i \in [k]\})$, $\D_w(h) \leq k \max_{i \in [k]} \D_w(f_i)$.

Suppose for some scalars $\alpha_i$'s, 
$
h=\sum_{i \in [k]} \alpha_i f_i.
$

For $u,v \in V$ we have
\begin{align*}
(h(u)-h(v))^2 
&= (\sum_{i\in[k]} \alpha_i(f_i(u)-f_i(v)))^2 \\
&\leq k\sum_{i\in [k]} \alpha_i^2 (f_i(u)-f_i(v))^2, \label{eq:1}
\end{align*}

\noindent where the last inequality follows from Cauchy-Schwarz inequality.
In the case $f_i$'s have disjoint support, we have 
$$
(h(u)-h(v))^2 
\leq 2 \sum_{i\in [k]} \alpha_i^2 (f_i(u)-f_i(v))^2.
$$
For each $e\in E$ we have 
\begin{align*}
\max_{u,v\in e} (h(u)-h(v))^2 
& \leq \max_{u,v\in e} k \sum_{i\in [k]} \alpha_i^2 (f_i(u)-f_i(v))^2 \\
& \leq k \sum_{i\in [k]} \alpha_i^2 \max_{u,v\in e} (f_i(u)-f_i(v))^2.
\end{align*} 

\noindent  Therefore, we have
\begin{align*}
\D_w(h) &=  \frac{\sum_{e} \w{e} \max_{u,v\in e}{(h(u)-h(v))^2}}{\sum_{u \in V} \w{u} h(u)^2}\\
&\leq  \frac{k\sum_{i\in[k]} \alpha_i^2 \sum_{e} \w{e} \max_{u,v\in e}{(f_i(u)-f_i(v))^2}}{\sum_{i\in[k]} \alpha^2_i \sum_{u \in V} \w{u} f_i(u)^2} \\
&\leq k \max_{i \in [k]} \D_w(f_i),
\end{align*}
as required.
\end{proof}

\begin{claim}
\label{claim:gz2}
We have $\gamma_2 = \zeta_2$.
\end{claim}

\begin{proof}
From Claim~\ref{claim:gz1}, we already have $\gamma_2 \leq \zeta_2$.  Hence, it suffices to show the other direction.
We shall consider the discrepancy ratio for the weighted space.

Suppose $f \perp_w \mathbf{1}$ attains
$\D_w(f)=\gamma_2$. Then, we have
\begin{align*}
\zeta_2 &\leq \max_{g=af+b\mathbf{1}} \frac{\sum_{e \in E} \w{e} \max_{u,v\in e}{(\g{u}-\g{v})^2}}{\sum_{v \in V} \w{v} \g{v}^2} \\
&= \max_{g=af+b\mathbf{1}} \frac{\sum_{e \in E} \w{e} \max_{u,v\in e}{a^2(\f{u}-\f{v})^2}}{\sum_{v \in V} \w{v} (a \f{v}+b)^2}\\
&= \max_{g=af+b\mathbf{1}} \frac{\sum_{e \in E} \w{e} \max_{u,v\in e}{a^2(\f{u}-\f{v})^2}}{\sum_{v \in V} \w{v}(a^2 \f{v}^2+b^2)+2ab\sum_{v \in V} \w{v} \f{v}} \\
&\leq \max_{g=af+b\mathbf{1}} \frac{\sum_{e \in E} \w{e} \max_{u,v\in e}{a^2(\f{u}-\f{v})^2}}{\sum_{v \in V} a^2 \w{v} \f{v}^2}=\gamma_2.
\end{align*} 
\end{proof}

\subsection{Small Set Expansion}
\label{sec:hyper-sse}

Even though we do not have an efficient method to 
generate $k$ orthonormal vectors that attain $\xi_k$.
As a warm up, we show that an approximation can still
give us a bound on the expansion of a set of size at most $O(\frac{n}{k})$.

\begin{theorem}[Formal Statement of \ref{thm:hyper-sse-informal}]
	\label{thm:hyper-sse}
	Suppose $H = (V,E,w)$ is a hypergraph, and 
	$f_1, f_2, \ldots, f_k$ are $k$ orthonormal vectors in the weighted space
	such that $\max_{s \in [k]} \D_w(f_s) \leq \xi$.
	Then, a random set $S \subset V$ can be constructed in polynomial time such that with
	$\Omega(1)$ probability,
		$\Abs{S} \leq \frac{24 \Abs{V}}{k}$  and
	\[ \phi(S) \leq C \min \{ \sqrt{r \log k}, \ratio  \} \cdot \sqrt{\xi},  \]
	where $C$ is an absolute constant and $r$ is the size of the largest hyperedge in $E$. 
\end{theorem}

Our proof is achieved by a
randomized polynomial time Algorithm~\ref{alg:hyper-sse}
that computes a set $S$ satisfying the conditions of the theorem,
given vectors whose discrepancy ratios are at most $\xi$.  
We will use the following \emph{orthogonal separator}~\cite{lm14b} subroutine.
We say that a set $S$ \emph{cuts} another set $e$, if there exist $u,v \in e$ such that
$u \in S$ and $v \not\in S$.

\begin{fact}[Orthogonal Separator~\cite{lm14b}]
	\label{lem:gen-orth-sep}
	There exists a randomized polynomial time algorithm that, given a set of unit vectors 
	$\set{\U}_{u \in V}$, parameters $\beta \in (0,1)$ and $\tau \in \Z^+$, 
	outputs a random set 
	$\widehat{S} \subset \set{\U}_{u \in V} $ such that
	for some absolute constant $c_1$ and $\alpha = \Theta(\frac{1}{\tau})$,
	we have the following.

	\begin{enumerate}
		\item For every $\U$, $\Pr[\U \in \widehat{S}]  = \alpha$.
		
		\item For every $\U,\V$ such that $\inprod{\U,\V} \leq \beta$,
		\[ \Pr[ \U \in \widehat{S} \textrm{ and } \V \in \widehat{S} ] \leq \frac{\alpha}{\tau} \mper     \]
		
		\item For any $e \subset \set{\U}_{u \in V}$ 
		\[ \Pr[e \textrm{ is ``cut'' by  } \widehat{S} ] \leq 
		 \frac{c_1}{\sqrt{1 - \beta}} \cdot \alpha \tau \log \tau \log \log \tau \sqrt{\log \Abs{e}} \cdot \max_{\U,\V \in e}\, \norm{\U - \V}.   \]
			\end{enumerate}
\end{fact}

\begin{remark}
We remark that the vectors do not have to satisfy the
$\ell_2^2$-constraints in this version of orthogonal separators~\cite{lm14b}.
\end{remark}

	\begin{algorithm}[H]
		\caption{Small Set Expansion}
		\begin{enumerate}
			\item {\bf Spectral Embedding}. Let $f_1, \ldots, f_k$ be 
			orthonormal vectors in the \textbf{weighted space}
			such that  $\max_{s \in [k]} \D_w(f_s) \leq \xi$. 
			We map a vertex $i \in V $ to a vector $u_i \in \R^k$ defined as follows. 
			For $i  \in V$ and $s \in [k]$,
			
			\[   u_i(s) = f_s(i) \mper \]
			In other words, we map the vertex $u$ to the vector formed by taking the coordinate
			corresponding to vertex $u$
			from $f_1, \ldots, f_k$.  We consider the Euclidean $\ell_2$ norm in $\R^k$.

			\item {\bf Normalization}. For every $i \in V$, let $\tilde{u}_i = \frac{u_i}{\norm{u_i}}$.
			\item {\bf Random Projection}. 
			Using Fact~\ref{lem:gen-orth-sep} (orthogonal separator), sample a random set $\widehat{S}$ from the set of vectors $\set{\tilde{u}_i}_{i \in V}$
			with $\beta = 99/100$ and $\tau = k$, 
			and define the vector $X \in \R^V$ as follows.

			\[  X_i \defeq \begin{cases}  \norm{u_i}^2 & \textrm{if } \tilde{u}_i \in \widehat{S} \\
			0 & \textrm{otherwise}   \end{cases} \mper \]
			\label{step:hyper-algstep-helper1}
			\item {\bf Sweep Cut}. Sort the coordinates of the vector $X$ in decreasing order and output
			the prefix having the least expansion (See  Proposition~\ref{prop:hyper-1d}). 
		\end{enumerate}
		\label{alg:hyper-sse}
	\end{algorithm}

We first prove some basic facts about the spectral embedding (Lemma~\ref{lem:hyper-spectral-embedding}), where
the analogous facts for graphs are well known.

\begin{lemma}[Spectral embedding]
\label{lem:hyper-spectral-embedding}
We have the following.
	\begin{enumerate}
		\item
		\[  \frac{ \sum_{e \in E} w_e \max_{i,j \in e}  \norm{u_i - u_j}^2 }{ \sum_{i \in V} w_i \norm{u_i}^2 }  \leq \max_{s \in [k]} \D_w(f_s) \mper  \]
		\item
		\[ \sum_{i \in V} w_i \norm{u_i}^2 = k \mper \]
		\item
		\[ \sum_{i \in V} w_i \inprod{u_j,u_i}^2 = \norm{u_j}^2, \qquad \forall j \in V \mper \]
		
		\item \[\sum_{e \in E} w_e \max_{i \in e} \norm{u_i} \cdot \max_{i,j \in e} \norm{u_i - u_j}
		\leq k \cdot \sqrt{\max_{s \in [k]} \D_w(f_s)}.\]
		
	\end{enumerate}
\end{lemma}

\begin{proof}
 
	\begin{enumerate}
		\item For the first statement,
		we have
		
		$\frac{ \sum_{e \in E} w_e \max_{i,j \in e}  \norm{u_i - u_j}^2 }{ \sum_{i \in V} w_i \norm{u_i}^2 } 
		= \frac{\sum_{e \in E} w_e \max_{i,j \in e}  \sum_{s \in [k]} (f_s(i) - f_s(j))^2   }{\sum_{i \in V} w_i \sum_{s \in [k]} f_s(i)^2}$
		
		$\leq
		\frac{\sum_{s \in [k]} \sum_{e \in E} w_e \max_{i,j \in e}   (f_s(i) - f_s(j))^2   }{\sum_{s \in [k]} \sum_{i \in V} w_i  f_s(i)^2}
		\leq \max_{s \in [k]} \D_w(f_s)$.

		\item The second statement follows because
		each $f_s$ has norm $1$ in the weighted space.

		\item For the third statement,
		\begin{align*}
			\sum_{i \in V} w_i \inprod{u_j, u_i}^2 & 
			=  \sum_{i \in V} w_i \paren{ \sum_{s \in [k]} f_s(j) f_s(i)}^2 \\
			& =  \sum_{i \in V} w_i  \sum_{s,t \in [k]} f_s(j) f_s(i) f_t(j) f_t(i) \\
			& =  \sum_{s,t \in [k]} f_s(j) f_t(j) \sum_{i \in V} w_i f_s(i)  f_t(i) \\
			& =  \sum_{s,t \in [k]} f_s(j) f_t(j) \cdot \inprod{f_s, f_t}_w \\
			& =  \sum_{s,t \in [k]} f_s(j) f_t(j) \cdot \Ind{s = t}\\
			& = \sum_{s \in [k]} u_j(s)^2 \\
			& =  \norm{u_j}^2.
		\end{align*}

	\item For the fourth statement,
		using the Cauchy-Schwarz inequality,
		we have
		
		\begin{align*}
		\sum_{e \in E} w_e \max_{i \in e} \norm{u_i} \cdot \max_{i,j \in e} \norm{u_i - u_j} &
		\leq \sqrt{ \sum_{e \in E} w_e \max_{i \in e} \norm{u_i}^2 } \cdot
		\sqrt{ \sum_{e \in E} w_e \max_{i,j \in e} \norm{u_i - u_j}^2 } \\
		& =  \sum_{e \in E} w_e \max_{i \in e} \norm{u_i}^2  \cdot
		\sqrt{ \frac{\sum_{e \in E} w_e \max_{i,j \in e} \norm{u_i - u_j}^2}{\sum_{e \in E} w_e \max_{i \in e} \norm{u_i}^2} } \\
		& \leq \sum_{e \in E} w_e \max_{i \in e} \norm{u_i}^2  \cdot \sqrt{\max_{s \in [k]} \D_w(f_s)},
		\end{align*}
		
		where the last inequality follows from the first statement.
		
		To finish with the proof,
		observe that 
		
		$\sum_{e \in E} w_e \max_{i \in e} \norm{u_i}^2 \leq \sum_{i \in V} w_i \norm{u_i}^2 = k$,
		where the last equality follows from the second statement.
		\end{enumerate}
\end{proof}

We denote $D:= \frac{\tau}{\sqrt{1 - \beta}} \cdot \log \tau \log \log \tau \cdot \sqrt{\log r}$.

\paragraph{Main Analysis} 
To prove that Algorithm~\ref{alg:hyper-sse} outputs a set which meets the 
requirements of Theorem~\ref{thm:hyper-sse}, we will show that the vector $X$ meets the requirements
of Proposition~\ref{prop:hyper-1d}. We prove an upper bound on the numerator 
$\sum_{e \in E} w_e \max_{i,j \in e} |X_i - X_j|$ in
Lemma~\ref{lem:hyper-sse-num} and a lower bound on the denominator 
$\sum_{i \in V} w_i X_i$
in Lemma~\ref{lem:hyper-sse-denom}.  We first show a technical lemma.

	\begin{lemma}
		\label{lem:2vecs} For any non-zero vectors $u$ and $v$,
	$\norm{\tilde{u}-\tilde{v}} \leq 2\frac{\norm{u-v}}{\sqrt{\norm{u}^2 + \norm{v}^2}}$.
	\end{lemma}
	
	\begin{proof} Denote $a := \norm{u}$, $b := \norm{v}$
	and $\theta := \inprod{\tilde{u}, \tilde{v}}$.
	Then, we have
	
	\begin{align*}
	\norm{\tilde{u}-\tilde{v}}^2 (\norm{u}^2 + \norm{v}^2) & = (2 - 2 \theta)(a^2 + b^2) \\
	& \leq 4(a^2 - 2ab \theta + b^2) = 4 \norm{u-v}^2,
	\end{align*}
	
	where the inequality is equivalent to $(1+\theta)(a^2 + b^2) - 4ab \theta \geq 0$.
	
	To see why this is true, consider
	the function $h(\theta) := (1+\theta)(a^2 + b^2) - 4ab \theta$
	for $\theta \in [-1,1]$.  Since $h'(\theta)$ is independent of $\theta$,
	$h$ is either monotonically increasing or decreasing.
	Hence, to show that $h$ is non-negative,
	it suffices to check that both $h(-1)$ and $h(1)$ are non-negative.
	\end{proof}

\begin{lemma}
	\label{lem:hyper-sse-num}
	We have 
	$\E[ \sum_{ e \in E} w_e \max_{i, j \in e} \Abs{X_i - X_j}] \leq 
	O(D) \cdot \sqrt{\xi}$.
\end{lemma}

\begin{proof}	
	For an edge $e \in E$ we have 
	
	\begin{equation}
		\label{eq:hyper-sse-helper16}
		 \E [ \max_{i,j \in e} \Abs{X_i - X_j} ]  \leq  \max_{i,j \in e} \Abs{ \norm{u_i}^2 - \norm{u_j}^2}  \cdot
		\Pr[\tilde{u}_i \in \widehat{S} \ \forall i \in e]
		+ \max_{i \in e} \norm{u_i}^2 \cdot \Pr[  e \textrm{ is cut by } \widehat{S}].
	\end{equation}

	By Fact~\ref{lem:gen-orth-sep}~(1), the probability in the
	first term is at most $\Theta(\frac{1}{k})$.
	Hence, the
	first term is at most
	\begin{equation}
		\label{eq:hyper-sse-helper10}
		\frac{\Theta(1)}{k} \cdot \max_{i, j\in e} \Abs{ \norm{u_i}^2 - \norm{u_j}^2} \leq
		\frac{\Theta(1)}{k} \cdot \max_{i, j \in e} \norm{u_i - u_j} \cdot \norm{ u_i + u_j } \leq 
		 \frac{\Theta(1)}{k} \cdot \max_{i, j \in e} \norm{u_i - u_j} \max_{i \in e} \norm{u_i}  \mper 
	\end{equation}
	
	To bound the second term in  (\ref{eq:hyper-sse-helper16}), 
	we divide the edge set $E$ into $E_1$ and $E_2$ as follows.
	\[ E_1 \defeq \set{ e \in E :  \max_{i,j \in e} \frac{ \norm{u_i}^2 }{  \norm{u_j}^2  } \leq 2   } 
	\quad \textrm{and} \quad   
	E_2 \defeq \set{ e \in E :  \max_{i,j \in e} \frac{ \norm{u_i}^2 }{  \norm{u_j}^2  } > 2 } \mper \] 
	
	$E_1$ is the set of those edges whose vertices have roughly equal lengths and 
	$E_2$ is the set of those edges whose vertices have large disparity in lengths.
	
	\begin{claim}\label{claim:E1E2}  
	Suppose $E_1$ and $E_2$ are as defined above.  Then, the following holds.
	
	\begin{compactitem}
	
	\item[(a)] For $e \in E_1$, we have
	
	$\Pr[  e \textrm{ is cut by } \widehat{S}] \leq O(\alpha D) \cdot
	\frac{\max_{i,j \in e} \norm{u_i - u_j} }{\max_{i \in e} \norm{u_i}}$.
	
	\item[(b)] For $e \in E_2$,
	we have $\max_{i \in e} \norm{u_i}^2 \leq 4 \max_{i \in e} \norm{u_i} \max_{i,j \in e} \norm{u_i - u_j}$.
	\end{compactitem}
	\end{claim}
	
	\begin{proof} We prove the two statements.
	
	\noindent (a) For $e \in E_1$,
	using Lemma~\ref{lem:2vecs} and  Fact~\ref{lem:gen-orth-sep},
	the probability that $e$ is cut by $\widehat{S}$ is at most
	
	$$O(\alpha D) \cdot
	\max_{i,j \in e} 
		\frac{ \norm{u_i - u_j} }{\sqrt{ \norm{u_i}^2 + \norm{u_j}^2 }} 
		\leq
		O(\alpha D) \cdot
	\frac{\max_{i,j \in e} \norm{u_i - u_j} }{\max_{i \in e} \norm{u_i}},$$
	
	where the inequality follows because $e \in E_1$.

	\noindent (b) Fix any $e \in E_2$,
	and suppose the vertices in $e = [r]$ are labeled such that
	$\norm{u_1} \geq \norm{u_2} \geq \ldots \geq \norm{u_r}$. Then, 
	from the definition of $E_2$, we have 
	\[   \frac{\norm{u_1}^2 }{ \norm{u_r}^2 } > 2 \mper \]
	
	Hence, $\max_{i,j \in e} \norm{u_i - u_j} \geq \norm{u_1 - u_r} \geq (1 - \frac{1}{\sqrt{2}}) \cdot \norm{u_1}$.
	Therefore, $\max_{i \in e} \norm{u_i}^2 \leq 4 \max_{i \in e} \norm{u_i} \max_{i,j \in e} \norm{u_i - u_j}$.
	\end{proof}

	For a hyperedge  $e \in E_1$, using Claim~\ref{claim:E1E2}~(a),
	the second term 
	in  (\ref{eq:hyper-sse-helper16}) is at most
	
	$\frac{O( D ) }{k}  \max_{i \in e} \norm{u_i} \max_{i,j \in e} \norm{u_i - u_j}$.
	
	For $e \in E_2$, in the second term of (\ref{eq:hyper-sse-helper16}),
	we can just upperbound the probability trivially by $1 \leq \frac{O(D)}{k}$,
	and use Claim~\ref{claim:E1E2}~(b) to conclude that the second term is also at most
	
	$\frac{O( D ) }{k}  \max_{i \in e} \norm{u_i} \max_{i,j \in e} \norm{u_i - u_j}$.

	Hence, inequality (\ref{eq:hyper-sse-helper16}) becomes:
	
	\begin{align*}
		\E[\max_{i,j \in e} \Abs{X_i - X_j}] & \leq \frac{O(D)}{k} \cdot
		\max_{i \in e} \norm{u_i} \max_{i,j \in e} \norm{u_i - u_j}.
		\end{align*}
	
	Summing over all hyperedges $e \in E$, we have 
	\begin{align*}
		\E[\sum_{e \in E} w_e \max_{i,j \in e} \Abs{X_i - X_j}] & 
		\leq \frac{O( D)}{k} \cdot
		\sum_{e \in E} w_e \max_{i \in e} \norm{u_i} \cdot
		\max_{i,j \in e} \norm{u_i - u_j} \\
		& \leq O(D) \cdot \sqrt{\xi},
	\end{align*}
where the last inequality follows
from Lemma~\ref{lem:hyper-spectral-embedding}~(4).
\end{proof}

\begin{lemma}
	\label{lem:hyper-sse-denom}
	We have
	\[ \Pr [ \sum_{i \in V} w_i X_i  > \frac{1}{2} ] \geq \frac{1}{12} \mper  \]
\end{lemma}

\begin{proof}
	We denote $Y \defeq \sum_{i \in V} w_i X_i$.
	We first compute $\E [Y]$ as follows. 
	\begin{align*}
		\E [Y] & = \sum_{i \in V} w_i \norm{u_i}^2 \Pr[ \tilde{u} \in \widehat{S} ] \\
		& = \sum_{i \in V} w_i \norm{u_i}^2 \cdot \alpha & \textrm{(From Fact~\ref{lem:gen-orth-sep}~(1))} \\
		& = k \alpha   &  \textrm{(Using Lemma~\ref{lem:hyper-spectral-embedding}~(2))} \mper
	\end{align*}
	
	Next we give an upper  bound of $\E[Y^2]$. 
	\begin{align*}
		\E [Y^2] & = \sum_{i,j \in V} w_i w_j \norm{u_i}^2 \norm{u_j}^2 \Pr[ \tilde{u}_i, \tilde{u}_j \in \widehat{S} ] \\
		& \leq  \sum_{ \substack{i,j: \\ \inprod{ \tilde{u}_i, \tilde{u}_j} \leq \beta }  } 
		w_i w_j \norm{u_i}^2 \norm{u_j}^2 \Pr[ \tilde{u}_i, \tilde{u}_j \in \widehat{S} ]
			+ 	\sum_{ \substack{i,j: \\ \inprod{ \tilde{u}_i, \tilde{u}_j} > \beta }  } 
		w_i w_j \norm{u_i}^2 \norm{u_j}^2 \Pr[ \tilde{u}_i, \tilde{u}_j \in \widehat{S} ].
	\end{align*}
	
	We use Fact~\ref{lem:gen-orth-sep}~(2) to bound the first term, and use the trivial bound of 	$\frac{1}{k}$ (Fact~\ref{lem:gen-orth-sep}~(1)) to bound $\Pr[ \tilde{u}_i, \tilde{u}_j \in S ]$ in the second term. Therefore, 
	\begin{align*}
		\E[Y^2] & \leq \sum_{ \substack{i,j: \\ \inprod{ \tilde{u}_i, \tilde{u}_j} \leq \beta }  } 
		w_i w_j \norm{u_i}^2 \norm{u_j}^2 \cdot \frac{\alpha}{k}
		+ \sum_{ \substack{i,j: \\ \inprod{ \tilde{u}_i, \tilde{u}_j} > \beta }  } 
		w_i w_j \norm{u_i}^2 \norm{u_j}^2 \cdot \frac{ \inprod{ \tilde{u}_i, \tilde{u}_j}^2 }{\beta^2} \cdot \alpha \\
		& \leq \sum_{i,j} w_i w_j \paren{   \frac{\alpha \norm{u_i}^2 \norm{u_j}^2  }{k}  + \frac{\alpha}{\beta^2 } \inprod{u_i,u_j}^2  } \\
		& =  \frac{\alpha}{k} \paren{ \sum_{i} w_i \norm{u_i}^2 }^2 + \frac{\alpha}{\beta^2 } \sum_{i,j} w_i w_j \inprod{u_i,u_j}^2  \\
		& =  \frac{\alpha}{k} \cdot k^2 + \frac{\alpha}{\beta^2 } \cdot k  =\alpha k ( 1 + \frac{1}{\beta^2} )
		\leq  3 k \alpha.  
		& \textrm{(Using Lemma~\ref{lem:hyper-spectral-embedding})}
	\end{align*}
	
	Since $Y$ is a non-negative random variable, we get using the Paley-Zygmund inequality that 
	\[ \Pr[Y \geq \frac{1}{2} \E[Y] ] \geq \paren{\frac{1}{2}}^2 \frac{ \E[Y]^2  }{\E[Y^2]} 
	= \frac{1}{4} \cdot \frac{1}{3} = \frac{1}{12} \mper \]
	This finishes the proof of the lemma.
\end{proof}

We are now ready to finish the proof of Theorem~\ref{thm:hyper-sse}.\\

\begin{proofof}{Theorem~\ref{thm:hyper-sse}}

\noindent \textbf{(1)} We first show that Algorithm~\ref{alg:hyper-sse} gives $S \subset V$
such that $|S| = O(\frac{n}{k})$ and 
$\phi(S) = O(k \log k \log \log k \cdot \sqrt{\xi \log r})$.

	By the definition of Algorithm~\ref{alg:hyper-sse},
	\[ \E[ \Abs{ \supp(X)} ] = \frac{n}{k} \mper \]
	Therefore, by Markov's inequality,
	\begin{equation}
		\label{eq:hyper-sse-helper4}
		\Pr[ \Abs{ \supp(X)}  \leq \frac{24n}{k}   ]\geq 1 - \frac{1}{24}     \mper
	\end{equation}
	
	Using Markov's inequality and Lemma~\ref{lem:hyper-sse-num},
	for some large enough constant $C_1 > 0$,
	\begin{equation}
		\label{eq:hyper-sse-helper5}
		\Pr[ \sum_{ e \in E} w_e \max_{u,v \in e} \Abs{X_u - X_v} \leq C_1 D \cdot \sqrt{\xi} ]
		\geq 1 - \frac{1}{48}  \mper
	\end{equation}

	Therefore, using a union bound over (\ref{eq:hyper-sse-helper4}), (\ref{eq:hyper-sse-helper5})
	and Lemma~\ref{lem:hyper-sse-num}, we get that with probability at least $\frac{1}{48}$,
	the following happens.
	
	\begin{compactitem}
	\item[(1)] $\frac{ \sum_{e \in E} w_e \max_{i,j \in e} \Abs{X_i - X_j} }{ \sum_{i \in V} w_i X_i } \leq O(D) \cdot \sqrt{\xi}$, and
	\item[(2)] $\Abs{\supp(X)} \leq \frac{24n}{k}$.
	\end{compactitem}
	
	When these two events happen, from Proposition~\ref{prop:hyper-1d},
	Algorithm~\ref{alg:hyper-sse} outputs a set $S$ such that
	$\phi(S) \leq O(D) \cdot \sqrt{\xi}$
	and $\Abs{S} \leq \Abs{\supp(X)} =  O(\frac{n}{k})$, as required.

	\vspace{1cm}
	
	\noindent \textbf{(2)} We next show that algorithmic
	version~\cite{lrtv12,lot12} of Fact~\ref{thm:graph-higher-cheeger}
	for $2$-graphs can give us $S \subset V$ such that
	$|S| = O(\frac{n}{k})$ and $\phi(S) = O(\sqrt{r \xi \log k})$.

	Given edge-weighted hypergraph $H = (V,E,w)$,
	we define an edge-weighted 2-graph $G = (V,E')$ as follows.
	For each $e \in E$, where $r_e = |e|$,
	add a complete graph on $e$ with each pair having weight $\frac{w_e}{r_e - 1}$.
	Observe that eventually a pair $\{u,v\}$ in $G$ has weight derived from all $e \in E$ such that both $u$ and $v$ are in $e$.
	In this construction, each vertex $u$ has the same weight in $H$ and $G$.
	
	We first relate the discrepancy ratios of the two graphs by
	showing that $\D_w^G(f) \leq \frac{r}{2} \cdot \D_w^H(f)$.
	Since the denominators are the same, we compare the contribution of each hyperedge $e \in E$ to the numerators.  For $e \in E$ with $r_e = |e|$,
	its contribution to the numerator of $\D_w^G(f)$ is
	$\frac{w_e}{r_e - 1} \sum_{\{u,v\} \in {e \choose 2}} (f_u - f_v)^2
	\leq w_e \cdot \frac{r_e}{2} \cdot \max_{u,v \in e} (f_u - f_v)^2$,
	which is $\frac{r_e}{2}$ times the contribution of $e$ to the numerator of $\D_w^H(f)$.
	
	Hence, Fact~\ref{thm:graph-higher-cheeger} for 2-graphs implies that
	given vectors orthogonal vectors $f_1, f_2, \ldots, f_k$ in the weighted space
	(where $\max_{i \in [k]} \D_w^G(f_i) \leq \frac{r \xi}{2}$),
	there is a procedure to return $S$ such that $|S| = O(\frac{n}{k})$
	and $\phi^G(S) = O(\sqrt{r \xi \log k})$.
	
	Therefore, it suffices to prove that $\phi^H(S) \leq \phi^G(S)$.  Again, the denominators involved are the same.  Hence, we compare the numerators.
	For each hyperedge $e \in \partial S$, suppose $r_e = |e|$ and $a_e = |e \cap S|$,
	where $0 < a_e < r_e$.
	Then, the contribution of $e$ to the numerator of $\phi^G(S)$
	is $\frac{w_e}{r_e -1} \cdot a_e( r_e - a_e) \geq w_e$, which is exactly
	the contribution of $e$ to the numerator of $\phi^H(S)$.  Hence, the result follows.
	
\end{proofof}

\subsection{Higher Order Cheeger Inequalities for Hypergraphs}
\label{sec:higher_order_cheeger}

In this section,
we achieve an algorithm that,
given $k$ orthonormal vectors $f_1, f_2, \ldots, f_k$
in the weighted space such that $\max_{s \in [k]} \D_w(f_s) \leq \xi$,
returns $\Theta(k)$ non-empty disjoint subsets
with small expansion.

\begin{theorem}[Restatement of Theorem~\ref{thm:hyper-higher-cheeger-informal}]
\label{thm:hyper-higher-cheeger}
Suppose $H = (V,E,w)$ is a hypergraph.  Then, we have the following.
\begin{compactitem}

\item [(a)] Suppose
	$f_1, f_2, \ldots, f_k$ are $k$ orthonormal vectors in the weighted space
	such that $\max_{s \in [k]} \D_w(f_s) \leq \xi$.
	There is a randomized procedure that runs in polynomial time
	such that for every $ \epsilon \geq \frac{1}{k}$, with $\Omega(1)$ probability,
	returns
	$\floor{(1-\epsilon)k}$  non-empty disjoint sets $S_1, \ldots, S_{\floor{(1-\epsilon)k}} \subset V$ such that 
	\[ \max_{i \in [\floor{(1-\epsilon)k}]} \phi(S_i) = \bigo{ \frac{k^2}{\eps^{1.5}} \log \frac{k}{\eps} \log{\log \frac{k}{\eps}} \sqrt{\log r} \cdot \sqrt{\xi} } \mper \]
	
	\item[(b)]
		For any $k$ disjoint non-empty sets $S_1, \ldots, S_k \subset V$
	\[ \max_{i \in [k]} \phi(S_i) \geq \frac{\zeta_k}{2},  \]
	where $\zeta_k$ is defined in Section~\ref{sec:higher-cheeger}.
	\end{compactitem}
\end{theorem}

\begin{proofof}{Theorem~\ref{thm:hyper-higher-cheeger}~(b)}
	
	For an arbitrary collection of $k$ disjoint non-empty sets $\{S_l\}_l$, let $f_l$ be the corresponding indicator function $S_l$. Then, the vectors $f_l$'s have disjoint support, and by Claim~\ref{claim:zx}, we have
	$$
	\frac{\zeta_k}{2} \leq \max_{l \in [k]}\D_w(f_l) = \max_{l \in [k]} \phi(S_l).
	$$
\end{proofof}

For statement (a), the proof is similar to Section~\ref{sec:hyper-sse},
and we also have a similar sampling algorithm.

\begin{algorithm}[H]
	\caption{Sample algorithm}
	\label{alg_cheeger}
	\begin{algorithmic}[1]
		\State Suppose $f_1, \ldots, f_k$ are orthonormal vectors in the weighted space
		such that $\max_{s \in [k]} \D_w(f_s) \leq \xi$.
		We map a vertex $i \in V $ to a vector $u_i \in \R^k$ defined as follows. 
			For $i  \in V$ and $s \in [k]$,
			
			\[   u_i(s) = f_s(i) \mper \]
				
		\State For each $i \in V$, normalize $\tilde{u}_i \gets \frac{u_i}{\norm{u_i}}$.
		\State Using Fact~\ref{lem:gen-orth-sep} (orthogonal separator), 
		sample
		$T:=\frac{2 \log 4n}{\alpha}$ independent subsets $S_1, \ldots, S_T \subset V$ with the set of vectors
		$\set{\tilde{u}_i}_{i \in V}$,
		$\beta = 1 - \frac{\epsilon}{72}$ and $\tau = \frac{16k}{\epsilon}$.

		\State Define measure $\mu(S) := \sum_{i\in S} w_i \norm{u_i}^2$.
		
		\noindent For each $l\in[T]$, define $S'_l$ as follows: 
		$$
		S'_l = 
		\begin{cases} 
		S_l &\mbox{if } \mu(S_l) \leq 1+ \frac{\epsilon}{4}; \\ 
		\emptyset & \mbox{otherwise.}
		\end{cases} 
		$$
		
		\State For each $l\in [T]$, let $S''_l=S'_l \backslash (\cup_{j\in [l-1]} S'_j)$.
		
		\State Arbitrarily merge sets from $\{S''_l\}$ to form sets having $\mu$-measure in $[\frac{1}{4}, 1 + \frac{\epsilon}{4}]$ (while discarding sets with total measure at most $\frac{1}{4}$). We name the resulting
		sets to be $B=\{B_1, \ldots, B_t\}$.
		
		\State For each $j\in[t]$, set $\hat{B}_j=\{ i \in B_j : \norm{u_i}^2 \geq r_j\}$, where $r_j$ is chosen to minimize $\phi(\hat{B}_j)$.
		\State Output the non-empty sets $\hat{B}_j$ with the smallest expansion $\phi(\hat{B}_j)$, for $j \in [t]$.
	\end{algorithmic}
\end{algorithm}

\noindent \textbf{Forming Disjoint Subsets.}  The algorithm first uses
orthogonal separator to generate subsets $S_l$'s independently.  If the $\mu$-measure
of a subset is larger than $1 + \frac{\eps}{4}$, then it is discarded.
We first show that with high probability, each vertex is contained in some subset
that is not discarded.

\begin{lemma}[Similar to {\cite[Lemma 2.5]{lm14}}]
	\label{lemma_pro_all_cover}
	For every vertex $i\in V$, and $l\in [T]$, we have 
	$$\Pr[i \in S'_l] \geq \frac{\alpha}{2}.$$
\end{lemma}

\begin{proof}
Recall that we sample $S_l$ using Fact~\ref{lem:gen-orth-sep} with
$\beta = 1 - \frac{\epsilon}{72}$ and $\tau = \frac{16k}{\epsilon}$.

	Fix $i \in V$.  If $i \in S_l$, then $i \in S'_l$ unless $\mu(S_l) > 1+ \frac{\epsilon}{4}$.
	Hence ,
	we only need to show that $\Pr[\mu(S_l) > 1+ \frac{\epsilon}{4} | i \in S_l] \leq \frac{1}{2}$.

	Define the sets $V_1$ and $V_2$ as follows
	$$
	V_1 = \{j \in V : \inprod{\tilde{u}_i,\tilde{u}_j} > \beta \}
	$$
	and 
	$$
	V_2 = \{j \in V : \inprod{\tilde{u}_i,\tilde{u}_j} \leq \beta \}.
	$$
	
	We next give an upper bound for $\mu(V_1)$. From Fact~\ref{lem:hyper-spectral-embedding}~(3), we have
	
	$$1 = \sum_{j \in V} w_j \norm{u_j}^2 \inprod{\tilde{u}_i,\tilde{u}_j}^2
	\geq \beta^2 \sum_{j \in V_1} w_j \norm{u_j}^2 = \beta^2 \cdot \mu(V_1).$$
	
	Hence, $\mu(V_1) \leq \beta^{-2} \leq 1 + \frac{\epsilon}{8}$.

	For any $j \in V_2$, we have $\inprod{\tilde{u}_i,\tilde{u}_j} \leq \beta$. 
	Hence, by Fact~\ref{lem:gen-orth-sep}~(2) of orthogonal separators, 
	$$
	\Pr[j \in S_l | i \in S_l] \leq \frac{1}{\tau}.
	$$
	
	Therefore,
	$$
	E[\mu(S_l \cap V_2) | i \in S_l] \leq \frac{\mu(V_2)}{\tau} \leq \frac{\mu(V)}{\tau} = \frac{\epsilon}{16},
	$$
	
	where the equality holds because $\mu(V) = k$ and $\tau = \frac{16k}{\epsilon}$.
	
	By Markov's inequality, $\Pr[\mu(S_l \cap V_2)\geq \frac{\epsilon}{8} | i \in S_l] \leq \frac{1}{2}$.
	
	Since $\mu(S_l)=\mu(S_l \cap V_1)+\mu(S_l \cap V_2)$, 
	we get 
	
	$\Pr[\mu(S_l) > 1+ \frac{\epsilon}{4} | i \in S_l]
	\leq \Pr[\mu(S_l \cap V_2)\geq \frac{\epsilon}{8} | i \in S_l] \leq \frac{1}{2}$,
	as required.
\end{proof}

\begin{lemma}
\label{lemma:enough_subsets}
With probability at least $\frac{3}{4}$,
every vertex is contained in at least one $S_l'$.
Moreover, when this happens, Algorithm~\ref{alg_cheeger} returns
at least $t \geq \floor{k(1 - \eps)}$ non-empty disjoint subsets.
\end{lemma}

\begin{proof}
From Lemma~\ref{lemma_pro_all_cover},
the probability that a vertex is not included
in $S_l'$ for all $l \in [T]$
is at most $(1 - \frac{\alpha}{2})^T \leq \exp(-\frac{\alpha T}{2}) \leq \frac{1}{4n}$.
Hence, by the union bound,
the probability that there exists a vertex
not included in at least one $S_l'$ is at most $\frac{1}{4}$.

When every vertex is included in some $S_l'$, then
the total $\mu$-measure  of the $S_l''$'s is exactly $\mu(V) = k$.
Since we merge the $S_l''$'s to form subsets of $\mu$-measure
in the range $[\frac{1}{4}, 1 + \frac{\eps}{4}]$,
at most a measure of $\frac{1}{4}$ will be discarded.

Hence, the number of subsets formed
is at least $t \geq \frac{k-\frac{1}{4}}{1+\frac{\epsilon}{4}} \geq (1-\epsilon)k$,
where the last inequality holds because $\frac{1}{k} \leq \eps < 1$.
\end{proof}

\noindent \textbf{Bounding Expansion.}  After we have
shown that the algorithm returns enough number of subsets (each
of which having $\mu$-measure at least $\frac{1}{4}$),
it remains to show that their expansion is small.
In addition to measure $\mu$,
we also consider
measure 

$\nu(S) :=
\sum_{e \subset S} w_e \max_{i,j \in e} (\norm{u_i}^2-\norm{u_j}^2)
	+\sum_{e\in \partial S } w_e \max_{i\in S\cap e} \norm{u_i}^2.$

The next lemma shows that there is a non-empty subset of $S$
having expansion at most $\frac{\nu(S)}{\mu(S)}$.

\begin{lemma}
\label{lemma:mu_nu_ratio}
Suppose $S$ is a subset of $V$.
For $r \geq 0$,
denote $S_r := \{i \in S: \norm{u_i}^2 \geq r\}$.
Then, there exists $r > 0$ such that $S_r \neq \emptyset$ and
$\phi(S_r) \leq \frac{\nu(S)}{\mu(S)}$.
\end{lemma}

\begin{proof}
Suppose $r$ is sampled uniformly from
the interval $(0, M)$,
where $M := \max_{i \in S} \norm{u_i}^2$.
Observe that for $r \in (0,M)$, $S_r$ is non-empty.

Then, it follows that an edge $e$ can be in $\partial S_r$
only if $e \subset S$ or $e \in \partial S$.

For $e \subset S$,
$e \in \partial S_r$ \emph{iff}
there exists $i,j \in e$ such that
$\norm{u_i}^2 < r \leq \norm{u_j}^2$.

On the other hand, if $e \in \partial S$,
then $e \in \partial S_r$ \emph{iff}
$r \leq \max_{i \in S \cap e} \norm{u_i}^2$.

Hence, $\E[w(\partial S_r)] = \frac{\nu(S)}{M}$.

Similarly, $i \in S$ is in $S_r$ \emph{iff} $r \leq \norm{u_i}^2$.
Hence, $\E[w(S_r)] = \frac{\mu(S)}{M}$.

Therefore, there exists $M > \rho > 0$
such that $\phi(S_\rho) = \frac{w(\partial S_\rho)}{w(S_\rho)}
\leq \frac{\E[w(\partial S_r)]}{\E[w(S_r)]} =  \frac{\nu(S)}{\mu(S)}$.
\end{proof}

In view of Lemma~\ref{lemma:mu_nu_ratio},
it suffices to show that the algorithm generates
subsets with small $\nu$-measure.

\begin{lemma}\label{lemma_bounding_cut}
	Algorithm~\ref{alg_cheeger} produces subsets $B_j$'s such that

		$$
		\mathbb{E}[\max_{l\in[t]} \nu(B_l)] \leq O(D) \cdot k \sqrt{\xi_k},
		$$ 
		where $D=\frac{\tau}{\sqrt{ 1 - \beta}} \cdot \log \tau  \log \log \tau  \sqrt{\log r}$, and $r=\max_{e\in E} |e|$.
\end{lemma}

\begin{proof}
		
		Let $E_{cut} := \cup_{l \in [t]} \partial B_l$ be the set of edges cut by $B_1,\ldots, B_t$.
		Then, for all $l \in [t]$,
		
		$\nu(B_l) \leq \sum_{e\in E_{cut}} w_e  \max_{i \in e} \norm{u_i}^2
		+ \sum_{e\in E} w_e \max_{i,j \in e} (\norm{u_i}^2-\norm{u_j}^2)$.
		
		Hence, $\max_{l \in [t]} \nu(B_j)$ also has the same upper bound.
		Taking expectation, we have
		
		\begin{equation}
		\mathbb{E}[\max_{l \in [t]} \nu(B_j)]
		\leq 
		\mathbb{E}[\sum_{e\in E_{cut}} w_e  \max_{i \in e} \norm{u_i}^2 ] +
		\sum_{e\in E} w_e \max_{i,j \in e} (\norm{u_i}^2-\norm{u_j}^2). 
		\label{eq:max_nu}
		\end{equation}

		The second term in (\ref{eq:max_nu}) is 
		\begin{align*}
		\sum_{e\in E} w_e \max_{i,j \in e} (\norm{u_i}^2-\norm{u_j}^2)
		&\leq \sum_{e\in E} w_e \max_{i,j \in e} \norm{u_i-u_j} \cdot \norm{u_i+u_j} \\
		&\leq 2\sum_{e\in E} w_e \max_{i,j \in e} \norm{u_i-u_j} \max_{i\in e} \norm{u_i}.
		\end{align*}
		
		To bound the first term in (\ref{eq:max_nu}), we divide the edge set $E$ into two parts $E_1$ and $E_2$ as follows
		$$
		E_1 =\{e\in E : \max_{i,j\in e} \frac{\norm{u_i}^2}{\norm{u_j}^2} \leq 2  \} 
		\text{\quad and \quad }
		E_2 =\{e\in E : \max_{i,j\in e} \frac{\norm{u_i}^2}{\norm{u_j}^2} > 2 \}.
		$$

		The first term in (\ref{eq:max_nu}) is 
		\begin{align} \label{eq_edge_cut_expectation}
		\mathbb{E}[\sum_{e\in E_{cut}} w_e  \max_{i\in e}\norm{u_i}^2 ]
		\leq \sum_{e\in E_1} \Pr[e \in \cup_{l \in [t]} \partial B_l ] \cdot w_e \max_{i\in e} \norm{u_i}^2 + \sum_{e\in E_2} w_e \max_{i\in e} \norm{u_i}^2.  
		\end{align}
		
		We next bound the contribution from edges in $E_1$.  Fix an edge $e \in E_1$.
		Recall that for $l \in [T]$, the set $S_l$ is generated
		independently by the orthogonal separator (Lemma~\ref{lem:gen-orth-sep}).
		For $l \in [T]$,
		we define $\mathcal{E}_l$ to be the event that
		for $l' \in [l-1]$, $S_{l'} \cap e = \emptyset$ and $e \in \partial S_l$.

		Observe that $e \in \cup_{l \in [t]} \partial B_l$
		implies that there exists $l \in [T]$ such that the event $\mathcal{E}_l$ happens.
		Next, if $\widehat{S}$ is sampled from the orthogonal separator in
		Lemma~\ref{lem:gen-orth-sep},
		then Lemma~\ref{lemma_pro_all_cover} implies that $\Pr[\widehat{S} \cap e = \emptyset] \leq 1 - \frac{\alpha}{2}$,
		and Claim~\ref{claim:E1E2}~(a) states that
		
		$\Pr[e \in \partial \widehat{S}] \leq O(\alpha D) \cdot
	\frac{\max_{i,j \in e} \norm{u_i - u_j} }{\max_{i \in e} \norm{u_i}}.$

		\noindent Therefore, we have
		
		\begin{align*}
		\Pr[e \in \cup_{l \in [t]} \partial B_l] 
		&\leq \sum_{l \in[T]} \Pr[\mathcal{E}_l] \\
		&\leq \sum_{l \in [T]} (1-\frac{\alpha}{2} )^{l-1} \cdot \Pr[e \in \partial \widehat{S}] \\
		& \leq \frac{2}{\alpha} \cdot \Pr[e \in \partial \widehat{S}]  \\
		&\leq O(D) \cdot
	\frac{\max_{i,j \in e} \norm{u_i - u_j} }{\max_{i \in e} \norm{u_i}}.
		\end{align*}
		
		Hence, the first term in (\ref{eq_edge_cut_expectation})
		is
		
		$\sum_{e\in E_1} \Pr[e \in \cup_{l \in [t]} \partial B_l ] \cdot w_e \max_{i\in e} \norm{u_i}^2 \leq
		\sum_{e\in E_1} w_e \max_{i,j \in e} \norm{u_i - u_j} \cdot \max_{i \in e} \norm{u_i}.$
		
		For $e \in E_2$, Claim~\ref{claim:E1E2}~(b)
		implies that the second term in (\ref{eq_edge_cut_expectation})
		is
		
		$\sum_{e\in E_2} w_e \max_{i\in e} \norm{u_i}^2
		\leq \sum_{e\in E_2} 4 w_e  \max_{i \in e} \norm{u_i} \max_{i,j \in e} \norm{u_i - u_j}$.
		
		Therefore, it follows that
		
		\begin{align*}
		\mathbb{E}[\max_{l \in [t]} \nu(B_l)] 
		&= O(D) \cdot  \sum_{e\in E} w_e \max_{i\in e} \norm{u_i} \max_{i,j\in e} \norm{u_i - u_j}  \\
		&\leq  O(D) \cdot k \sqrt{\max_{s \in [k]} \D_w(f_s)} \\
		& \leq  O(D) \cdot k \cdot \sqrt{\xi},
		\end{align*}
		where the second to last inequality comes from Lemma~\ref{lem:hyper-spectral-embedding}~(4).
\end{proof}

\begin{proofof}{Theorem~\ref{thm:hyper-sse}~(a)}
	We run Algorithm~\ref{alg_cheeger}. By Lemma~\ref{lemma:enough_subsets},
	with probability at least $\frac{3}{4}$,
	it produces at least
	$t \geq (1-\epsilon)k$ subsets $B_1, \ldots, B_t$,
	each of which has $\mu$-measure at least $\frac{1}{4}$. 
	
	Using Markov's inequality and Lemma~\ref{lemma_bounding_cut},
	with probability at least $\frac{3}{4}$,
	 we have $\max_{l \in [t]}\nu(B_l) \leq
	 4 \E[\max_{l \in [t]}\nu(B_l)] = O(D k) \cdot \sqrt{\xi}$.
	 
	 By union bound, with probability at least $\frac{1}{2}$,
	 the algorithm produces at least $t \geq (1-\epsilon)k$ 
	 disjoint subsets $B_l$, each of
	 which satisfies $\nu(B_l) = O(D k) \cdot \sqrt{\xi}$
	 and $\mu(B_l) \geq \frac{1}{4}$.
	
	Hence, Lemma~\ref{lemma:mu_nu_ratio} implies that each such $B_l$
	contains a non-empty subset $\hat{B}_l$
	such that $\phi(\hat{B}_j) \leq \frac{\nu(B_l)}{\mu(B_l)} = O(D k) \cdot \sqrt{\xi}$, as required.

\end{proofof}

%% file: vertex_expansion.tex
\section{Vertex Expansion in $2$-Graphs and Hardness}
\label{sec:vert-exp}

As mentioned in Section~\ref{sec:vertex-results},
vertex expansion in $2$-graphs
is closely related to hyperedge expansion.
Indeed, Reduction~\ref{red:hyper-vert}
implies that vertex expansion in $d$-regular graphs
can be reduced to hyperedge expansion.
We show that this reduction also relates
the parameter $\linf$ (see (\ref{eq:linf})) defined by
Bobkov \etal~\cite{bht00}
with the parameter $\gamma_2$ associated
with the Laplacian we define (in Section~\ref{sec:disp}) for hypergraphs.

\begin{theorem}[Restatement of Theorem~\ref{thm:linf-eig2}]
Let $G = (V,E)$ be a undirected $d$-regular $2$-graph with parameter~$\linf$, and let $H = (V,E')$ be the hypergraph obtained
in Reduction~\ref{red:hyper-vert} having parameter~$\eig_2$. Then,
\[  \frac{\eig_2}{4} \leq  \frac{\linf}{d} \leq \eig_2 \mper    \]
\end{theorem}

\begin{proof}
Using Theorem~\ref{th:hyper_lap}
for hypergraphs, the parameter
$\eig_2$ of $H$ can be reformulated in terms of the weighted space as:
\[  \eig_2 = \min_{f \perp \one} 
	\frac{ \sum_{u \in V}  \max_{  i,j \in \paren{ \set{u} \cup N(u)}} \paren{f_i - f_j}^2 }{d \sum_{u \in V} f_u^2} \mper    \]
Therefore, it follows that $  \frac{\linf}{d} \leq \eig_2$.

Next, using $(x + y)^2 \leq 4 \max \set{x^2,y^2}$ for any $x,y \in \R$, we get  
\[  \eig_2 = \min_{f \perp \one} 
	\frac{ \sum_{u \in V}  \max_{  i,j \in \paren{ \set{u} \cup N(u)}} \paren{f_i - f_u + f_u - f_j}^2  }{d \sum_{u \in V} f_u^2} 
\leq \min_{f \perp \one} 
	\frac{ \sum_{u \in V} 4 \max_{ v \sim u} \paren{f_v - f_u}^2 }{d \sum_{u \in V} f_u^2}
	= \frac{4 \linf}{d}  \mper    \]
\end{proof}

\subsection{Hardness via the \SSEH}

We state the \SSEH proposed by Raghavendra and Steurer \cite{rs10}.

\begin{hypothesis}[Small-Set Expansion (\sse) Hypothesis]
	\label{hyp:sse}
	For every constant $\eta > 0$, there exists sufficiently small $\delta>0$
	such that, given a graph $G$ (with unit edge weights), it is NP-hard to distinguish the
	following two cases:
	
	\begin{description}\item[\yes:] 
		there exists a vertex set $S$ with $\delta \leq \frac{|S|}{n} \leq 10\delta$ and
		edge expansion 
		$\phi(S)\le \eta$,
		\item[\no:]  all vertex sets $S$  with $\delta \leq \frac{|S|}{n} \leq 10\delta$ have expansion
		$\phi(S)\ge 1-\eta$.
	\end{description}
\end{hypothesis} 

\paragraph{\SSEH}
Apart from being a natural optimization problem, the small-set expansion problem is closely tied to the Unique
Games Conjecture.  Recent work by Raghavendra-Steurer
\cite{rs10} established the reduction from the small-set expansion problem to the well known Unique
Games problem, thereby showing that \SSEH implies the Unique Games Conjecture.  
We refer the reader to \cite{rst12} for a comprehensive discussion on the 
implications of \SSEH.
We shall use the following hardness result
for vertex expansion based on \SSEH.

	\begin{fact}[\cite{lrv13}]
	\label{fact:lrv13}
		For every $\eta > 0$, there exists an absolute constant $C_1$ such that $\forall \varepsilon>0 $ it is \sse-hard to distinguish 
		between the following two cases for a given graph $G = (V,E,w)$ with maximum degree $d \geq 100/\varepsilon$ and minimum 
		degree $c_1 d$ (for some absolute constant $c_1$).
		\begin{description}
			\item[\yes] : There exists a set $S \subset V$ of size $\Abs{S} \leq \Abs{V}/2$ such that 
			\[ \phiv(S) \leq \varepsilon \]
			\item[\no] : For all sets $S \subset V$, 
			\[  \phiv(S) \geq  \min \set{10^{-10}, C_1 \sqrt{\varepsilon\log d}} - \eta.\]
		\end{description}
	\end{fact}

Reduction~\ref{red:hyper-vert} implies
that vertex expansion in $2$-graphs is closely
related to hyperedge expansion.
Therefore, the hardness of vertex expansion as
stated in Fact~\ref{fact:lrv13}
should imply the hardness of hyperedge expansion.
We formalize this intuition in the following theorems.

\begin{theorem}[Formal statement of \ref{thm:hyper-expansion-hardness-informal}]
	\label{thm:hyper-expansion-hardness}
	For every $\eta> 0$, there exists an absolute constant $C$ such that
	for all $\widehat{\ve} >0 $ it is \sse-hard to distinguish 
	between the following two cases for a given hypergraph $H = (V,E,w)$ with maximum hyperedge size $r$ such that $\widehat{\ve} r \log r \in [\eta^2, c_2]$ (for some absolute constant $c_2$)
	and $\rmin \geq c_1 r$ (for some absolute constant $c_1$).
	\begin{description}
		\item[\yes] : There exists a set $S \subset V$ such that 
		\[ \phi_H(S) \leq \widehat{\ve} \]
		\item[\no] : For all sets $S \subset V$,
		\[  \phi_H(S) \geq  C \sqrt{\widehat{\ve} \cdot \frac{\log r}{r}}.\]
	\end{description}
\end{theorem}

\begin{proof}
Given an undirected graph $G$ with maximum degree $d$
and minimum degree $\Omega(d)$ as in Fact~\ref{fact:lrv13},
we apply Reduction~\ref{red:hyper-vert} to obtain
a hypergraph $H$ with maximum edge cardinality $r = d+1$.
Then,
Fact~\ref{thm:hyper-vert-exp}
implies that for any subset $S$ of vertices,
$c_i \cdot \phi_H(S) \leq \frac{\phiv_G(S)}{d+1} \leq \phi_H(S)$.

Fix some small enough $\eta > 0$ and corresponding $C_1 > 0$
as in Fact~\ref{fact:lrv13}.
Let $\ve > \frac{100}{d+1} = \frac{100}{r}$.

Under the \yes case of vertex expansion in Fact~\ref{fact:lrv13},
there is some subset $S$ such that $|S| \leq \frac{|V|}{2}$ and
$\phiv_G(S) \leq \ve$.
This implies that $\phi_H(S) \leq \frac{\ve}{c_1 r}$,
and we denote $\widehat{\ve} := \frac{\ve}{c_1 r} > \frac{100}{c_1 r^2}$.

Under the \no case of vertex expansion in Fact~\ref{fact:lrv13},
we have the fact that any $S \subset V$ has vertex expansion

$\phiv_G(S) \geq  \min \set{10^{-10}, C_1 \sqrt{\varepsilon \log d}} - \eta$.

This implies that for some constant $C'$ depending on $C_1$ and $c_1$,

$\phi_H(S) \geq \frac{\phiv_G(S)}{r} \geq
\min \set{\frac{10^{-10}}{r}, C' \sqrt{\widehat{\ve} \cdot  \frac{\log r}{r}}} - \frac{\eta}{r}$.

Observe that this lower bound is non-trivial under the case

$\frac{10^{-10}}{r} \geq C' \sqrt{\widehat{\ve} \cdot  \frac{\log r}{r}}
\geq 2 \cdot \frac{\eta}{r}$,
which is equivalent to $\widehat{\ve} r \log r \in [\eta^2, c_2]$,
for some constant $c_2$ depending on $C_1$ and $c_1$.
Hence, under this case,
we have $\phi_H(S) \geq \frac{C'}{2} \cdot \sqrt{\widehat{\ve} \cdot  \frac{\log r}{r}}$.

Hence, the \sse-hardness in Fact~\ref{fact:lrv13} finishes the proof.
\end{proof}

\begin{theorem}[Formal statement of \ref{thm:hyper-eigs-lower-informal}]
	\label{thm:hyper-eigs-lower}
	For every $\eta > 0$, there exists an absolute constant $C$ such that $\forall \overline{\ve} >0 $ it is \sse-hard to distinguish 
	between the following two cases for a given hypergraph $H = (V,E,w)$ with maximum hyperedge size $r$
	such that $\overline{\ve} r \log r \in [\eta^2, c_2]$ (for some absolute constant $c_2$),
	 $\rmin \geq c_1 r$ (for some absolute constant $c_1$) and $\gamma_2 \leq \frac{1}{r}$
	where $\gamma_2$ is the parameter associated with $H$
	as in Theorem~\ref{thm2:hyper-cheeger}.
	
	\begin{description}
		\item[\yes] : $\gamma_2 \leq \overline{\ve}$.
		
		\item[\no] : $\gamma_2 \geq C \overline{\ve} \log r$.
	\end{description}
\end{theorem}

\begin{proof}
We shall use
the hardness result in Theroem~\ref{thm:hyper-expansion-hardness},
and the Cheeger inequality for hyeprgraphs in
Theorem~\ref{thm2:hyper-cheeger} and
Proposition~\ref{prop:hyper-sweep-rounding}.

Given a hypergraph $H$, we have

$\frac{\gamma_2}{2} \leq \phi_H \leq \gamma_2 + 2 \sqrt{\frac{\gamma_2}{\rmin}}
\leq O(\sqrt{\frac{\gamma_2}{r}})$,
where the last inequality follow
because $\rmin = \Omega(r)$ and $\gamma_2 \leq \frac{1}{r}$.

	Hence, the \yes case in Theorem~\ref{thm:hyper-expansion-hardness}
	implies that $\gamma_2 \leq 2 \widehat{\ve}$.
	
	The \no case in Theorem~\ref{thm:hyper-expansion-hardness}
	implies that $\gamma_2 = \Omega(\widehat{\ve} \log r)$.
	
Therefore, the hardness result in
Theorem~\ref{thm:hyper-expansion-hardness} finishes the proof. 	
\end{proof}

%% file: approxminimizer.tex
\section{Polynomial Time Approximation Algorithm for Procedural Minimizers}
\label{sec:hyper-eigs-poly-alg}

Observe the procedures in Section~\ref{sec:cheeger}
take $k$ orthonormal vectors $f_1, f_2, \ldots, f_k$ in the weighted
space such that $\max_{i \in [k]} \D_w(f_i)$ is small.
However, we do not know of an efficient algorithm to generate
such $k$ vectors to attain the minimum $\xi_k$.
In this section, we consider an approximation algorithm to
produce these vectors.

\begin{theorem}[Restatement of Theorem~\ref{th:approx_xik_informal}]
\label{th:approx_xik}
There exists a randomized polynomial time algorithm that, given a hypergraph $H = (V,E,w)$
and a parameter $k < \Abs{V}$, outputs $k$ orthonormal vectors $f_1, \ldots, f_k$
in the weighted space
such that with high probability, for each $i \in [k]$,
\[ \D_w(f_i) \leq \bigo{i \log r\,\cdot \xi_i}.  \]
\end{theorem}

Observe that Theorem~\ref{th:approx_xik} gives
a way to generate $k$ orthonormal vectors in the weighted
space such that the maximum discrepancy ratio $\D_w(\cdot)$ is 
at most $k \log r \cdot \xi_k$.  Hence, these vectors
can be used as inputs for the procedures in Theorem~\ref{thm2:hyper-cheeger}
(more precisely, we use an approximate $f_2$ in Proposition~\ref{prop:hyper-1d}),
Theorems~\ref{thm:hyper-sse} and~\ref{thm:hyper-higher-cheeger}
to give approximation algorithms as described
in Corollaries~\ref{cor:hyper-sparsest-informal},
\ref{cor:hyper-sse-informal}
and~\ref{cor:hyper-higher-cheeger-informal}.

The approximate algorithm in Theorem~\ref{th:approx_xik}
achieves the $k$ vectors by starting with $f_1 \in \spn(\vec{1})$,
and repeatedly using the algorithm in the following
theorem to generate approximate procedural minimizers.

\begin{theorem}[Restatement of Theorem~\ref{thm:hyper-eigs-alg}]
\label{thm:hyper-eigs-alg_formal}
Suppose for $k \geq 2$,
$\{f_i\}_{i \in [k-1]}$ is a set of orthonormal vectors
in the weighted space, and
define $\gamma := \min \{ \D_w(f) : \vec{0} \neq f \perp_w \{f_i:
i \in [k-1]\}\}$.
Then, there is a randomized procedure that
produces a non-zero vector $f$ that
is orthogonal to $\{f_i\}_{i \in [k-1]}$
in polynomial time,
such that with high probability,
$\D_w(f) = \bigo{\gamma \log r}$, 
where $r$ is the size of the largest hyperedge.
\end{theorem}

\begin{proofof}{Theorem~\ref{th:approx_xik}}
On a high level, we start with $f_1 := \frac{\vec{1}}{\norm{\vec{1}}_w}$.
For $1 < i \leq k$, assuming that orthonormal vectors $\{f_l : l \in [i-1]\}$
are already constructed, we apply Theorem~\ref{thm:hyper-eigs-alg_formal}
to generate $f_i$.
Hence, it suffices to show that $\D_w(f_i) \leq \bigo{i \log r \cdot \xi_i}$.

We prove that if $\xi := \min \{ \D_w(f) : \vec{0} \neq f \perp_w \{f_l: l \in [i-1]\}\}$,
then $\xi \leq i \cdot \xi_i$.  Hence, Theorem~\ref{thm:hyper-eigs-alg_formal}
implies that $\D_w(f_i) \leq \bigo{\xi \log r }\leq \bigo{ i \log r \cdot \xi_i}$.

Therefore, it remains to show $\xi \leq i \cdot \xi_i$.
Suppose $g_1, g_2, \ldots, g_i$ are orthonormal vectors in
the weighted space that attain $\zeta_i$ (which
is defined in Section~\ref{sec:min}).

Since $\spn(\{g_1, g_2, \ldots, g_i\})$ has dimension~$i$, 
there exists non-zero $g \in \spn(\{g_1, g_2, \ldots, g_{i}\})$
such that $g \perp_w \{f_1, f_2, \ldots, f_{i-1}\}$.
By the definition of $\zeta_i$,
we have $\D_w(g) \leq \zeta_i \leq i \xi_i$,
where the last inequality follows from Claim~\ref{claim:zx}.
Hence, we have $\xi \leq i \xi_i$, as required.
\end{proofof}

We next give an $\sdp$ relaxation (\ref{sdp:eig-k}) and a rounding
algorithm (Algorithm~\ref{alg:hyper-eigs-rounding})
to prove Theorem~\ref{thm:hyper-eigs-alg_formal}.

\subsection{An $\sdp$ Relaxation to Approximate Procedural Minimizers: Proof of
Theorem~\ref{thm:hyper-eigs-alg_formal}}
\label{sec:sdp_relax}

We present SDP~\ref{sdp:eig-k} to compute a vector 
in the weighted space that is
orthogonal to $f_1, \ldots, f_{k-1}$ having the least discrepancy ratio $\D_w(\cdot)$. 
In the SDP, for each $u \in V$,
the vector $\vec{g_u}$ represents the $u$-th coordinate of the 
vector $f \in \R^V$ that we try to compute.
The objective function of the SDP and equation~(\ref{eq:hyper-sdp-norm}) seek to minimize
the discrepancy ratio $\D_w(\cdot)$. 
We shall see that
equation~(\ref{eq:hyper-sdp-orth}) ensures
that after rounding, the resulting vector $f$ is
orthogonal to $f_1, \ldots, f_{k-1}$ in the weighted space
and achieves $\bigo{\log r}$-approximation with constant probability.

\begin{mybox}
\begin{SDP}
\label{sdp:eig-k}
\[ \sdpval \defeq \min \sum_{e \in E} w_e \max_{u,v \in e} \norm{\vec{g_u} - \vec{g_v}}^2  \] 
\subjectto
\begin{equation}
\label{eq:hyper-sdp-norm}
 \sum_{v \in V} w_v \norm{\vec{g_v}}^2  =  1 
\end{equation}

\begin{equation}
\label{eq:hyper-sdp-orth}
\sum_{v \in V} w_v f_i(v) \, \vec{g_v}  = \vec{0} \qquad \forall i \in [k-1] 
\end{equation}

\end{SDP}
\end{mybox}

\begin{algorithm}[H]
\caption{Rounding Algorithm for Computing Eigenvalues}
\label{alg:hyper-eigs-rounding}
\begin{algorithmic}[1]
\State Solve \sdp \ref{sdp:eig-k} 
to generate vectors $\vec{g_v} \in \R^n$ for $v \in V$.

\State Sample a random Gaussian vector $\vec{z} \sim \cN(0,1)^n$. For $v \in V$, set $f(v) \defeq \inprod{\vec{g_v},\vec{z}}$.

\State Output $f$.
\end{algorithmic}
\end{algorithm}

\begin{lemma}[Feasibility]
\label{lemma:sdp_feas}
With probability 1, Algorithm~\ref{alg:hyper-eigs-rounding} outputs a non-zero vector $f$ such that $f \perp_w \{f_1, f_2, \ldots, f_{k-1}\}$.
\end{lemma}

\begin{proof}
Because of equation~(\ref{eq:hyper-sdp-norm}),
there exists $v \in V$ such that $\vec{g_v} \neq \vec{0}$.
Hence, when $\vec{z}$ is sampled from $\cN(0,1)^n$,
the probability that $f(v) := \langle \vec{z}, \vec{g_v} \rangle$
is non-zero is 1.

For any $i \in [k-1]$, we use equation~\ref{eq:hyper-sdp-orth} to achieve:

\[ \inprod{f,f_i}_w = \sum_{v \in V} w_v \inprod{\vec{g_v},\vec{z}} f_i(v) =  \inprod{\sum_{v \in V} w_v f_i(v) \vec{g_v}, \vec{z}}  = 0 \mper \]
\end{proof}

\begin{lemma}[Approximation Ratio]
\label{lem:hyper-eigs-rounding}
With probability at least $\frac{1}{24}$, Algorithm~\ref{alg:hyper-eigs-rounding} outputs a vector $f$ such that 
$\D_w(f) \leq  384 \log r\,\cdot \sdpval$.
\end{lemma}

\begin{proof}
To give an upper bound on $\D_w(f_k)$,
we prove an upper bound on the numerator and a lower bound on the denominator in
the definition of $\D_w( \cdot )$.

For the numerator, we have
\begin{align*}
\Ex{\sum_{e \in E} w(e)  \max_{u,v \in e} (f(u) - f(v))^2 } &
= \sum_{e \in E} w(e) \cdot \Ex{ \max_{u,v \in e} (f(u) - f(v))^2 } \\
& \leq 8 \log r \sum_{e \in E} w(e) \max_{u,v \in e} \norm{\vec{g_u} - \vec{g_v}}^2 \, \,  \textrm{(Using Fact~\ref{fact:appGauss})}  \\
& = 8 \log r\, \cdot \sdpval,  \\
\end{align*}
where the inequality follows from Fact~\ref{fact:appGauss} in the following manner.
For each $e \in E$,
observe that the $\max_{u,v \in e}$ is over a set of cardinality
${r \choose 2} \leq \frac{r^2}{2}$.  Moreover
for $u,v \in e$, $f(u) - f(v) = \langle \vec{g_u} - \vec{g_v}, \vec{z} \rangle$
is a normal distribution with variance $\norm{\vec{g_u} - \vec{g_v}}^2$ and mean 0.
Hence, Fact~\ref{fact:appGauss} 
implies that $\Ex{\max_{u,v \in e} (f(u) - f(v))^2 }
\leq 8 \log r \cdot \max_{u,v \in e} \norm{\vec{g_u} - \vec{g_v}}^2$.

Therefore, by Markov's Inequality,
\begin{equation}
\label{eq:sdp-hyper-num}
\Pr\left[ \sum_{e \in E} w(e) \max_{u,v \in e} (f(u) - f(v))^2   \leq 192 \log r\, \cdot \sdpval  \right] \geq 1 - \frac{1}{24} \mper
\end{equation}

For the denominator, using linearity of expectation, we get 
\[ \Ex{ \sum_{v\in V} w_v f(v)^2} = \sum_{v \in V} w_v \Ex{ \inprod{\vec{g_v}, \vec{z}}^2} = \sum_{v \in V} w_v \norm{\vec{g_v}}^2 = 1 \qquad \textrm{(Using Equation~\ref{eq:hyper-sdp-norm})}. \]

Now applying Fact~\ref{lem:squaregaussian} to the denominator we conclude
\begin{equation} 
\label{eq:sdp-hyper-denom}
\Pr\left[\sum_{v\in V} w_v f(v)^2 \geq \frac{1}{2} \right] \geq \frac{1}{12} \mper  
\end{equation}
Using the union bound on Inequality~(\ref{eq:sdp-hyper-num}) and Inequality~(\ref{eq:sdp-hyper-denom}) we get that 
\[ \Pr\left[ \D_w(f) \leq  384 \log r \, \cdot \sdpval \right]  \geq \frac{1}{24} .\]
\end{proof}

\begin{fact}[Variant of Massart's Lemma]
\label{fact:appGauss}
Suppose $Y_1, Y_2, \ldots, Y_d$ are normal random variables
that are not necessarily independent.
For each $i \in [d]$, suppose $\Ex{Y_i} = 0$ and $\Ex{Y_i^2} = \sigma_i^2$.
Denote $\sigma := \max_{i \in [d]} \sigma_i$.
Then, we have

\begin{compactitem}
\item[1.] $\Ex{\max_{i \in [d]} Y_i^2} \leq 4 \sigma^2 \ln d $, and
\item[2.] $\Ex{\max_{i \in [d]} |Y_i|} \leq 2 \sigma \cdot \sqrt{\ln d}$.
\end{compactitem}
\end{fact}

\begin{proof}
For $i \in [d]$, we write $Y_i = \sigma_i Z_i$,
where $Z_i$ has the standard normal distribution $\cN(0,1)$.
Observe that for any real numbers
$x_1, x_2, \ldots, x_d$, for any positive integer $p$,
we have
$\max_{i \in [d]} x_i^2 \leq (\sum_{i \in [d]} x_i^{2p})^{\frac{1}{p}}$.
Hence, we have

\begin{eqnarray*}
\Ex{\max_{i \in [d]} Y_i^2} & \leq & \Ex{ \left( \sum_{i \in [d]} Y_i^{2p} \right)^{\frac{1}{p}} }
  \leq  \left(\Ex{ \sum_{i \in [d]} Y_i^{2p} } \right)^{\frac{1}{p}} \quad 
	\textrm{ (by Jensen's Inequality, because $t \mapsto t^{\frac{1}{p}}$ is concave )} \\
	& \leq & \sigma^2 \left(\Ex{ \sum_{i \in [d]} Z_i^{2p} } \right)^{\frac{1}{p}} =
	\sigma^2 \left(\sum_{i \in [d]} \frac{(2p)! }{(p)! 2^{p} } \right)^{\frac{1}{p}}
	\quad \textrm{(For $Z_i \sim \cN(0,1)$, $\Ex{Z_i^{2p}} = \frac{(2p)! }{(p)! 2^{p} }$)} \\
 & \leq  & \sigma^2 p d^{\frac{1}{p}}. \quad \textrm{(using $\frac{(2p)!}{p!} \leq (2p)^{p} $ )} \\
\end{eqnarray*}

Picking $p = \ceil{\log d}$ gives
the first result $\Ex{\max_{i \in [d]} Y_i^2} \leq 4 \sigma^2 \log d$.
Moreover, the inequality $\Ex{|Y|} \leq \sqrt{\Ex{Y^2}}$ immediately gives the second
result.
\end{proof}

\begin{fact} 
\label{lem:squaregaussian}
Let $Y_1,\ldots, Y_m$ be normal random variables (that are not necessarily independent) 
having mean 0
such that $ \Ex{ \sum_i Y_i^2} = 1$ then
\[ \Pr\left[\sum_i Y_i^2 \geq \frac{1}{2} \right] \geq \frac{1}{12} \mper \]
\end{fact}

\begin{proof}
We will bound the second moment of the random variable $R := \sum_i Y_i^2$ as follows.

\begin{align*}
\Ex{R^2} &= \sum_{i,j} \Ex{Y_i^2 Y_j^2} \\ 
& \leq \sum_{i,j} \paren{ \Ex{Y_i^4} }^{\frac{1}{2}} \paren{\Ex{Y_j^4} }^{\frac{1}{2}} & \textrm{ (Using Cauchy-Schwarz Inequality)} \\
& = \sum_{i,j} 3\, \Ex{Y_i^2} \Ex{Y_j^2}  & \textrm{ (Using } \Ex{Z^4} = 3 \paren{\Ex{Z^2}}^2 \textrm{for Gaussian $Z$)} \\
& = 3\paren{ \sum_{i}  \Ex{Y_i^2} }^2 = 3.  
\end{align*}

By the Paley-Zygmund inequality,
\[ \Pr\left[ R \geq \frac{1}{2} \cdot \Ex{R} \right] \geq \left(\frac{1}{2}\right)^2 \cdot \frac{\Ex{R}^2}{\Ex{R^2}} \geq \frac{1}{12} \mper \]
\end{proof}

\ignore{
We now have all the ingredients to prove \ref{thm:hyper-eigs-alg}.

\begin{proofof}[Theorem~\ref{thm:hyper-eigs-alg}]
We will prove this theorem inductively. 
For the basis of induction, we have the first vector 
$x_1 = \Wh \vec{1}$ and $f_1 = \Wmh x_1$. 
We assume that we have computed $x_1,\ldots, x_{k-1}$ satisfying 
$\Dc(x_i) \leq \bigo{i \log r\, \xi_i}$. We now show how to compute $x_k$. 

Proposition~\ref{prop:hyper-CF} implies that for \sdp~\ref{sdp:eig-k},
\[ \sdpval \leq k\, \xi_k \mper \]
Therefore, from Lemma~\ref{lem:hyper-eigs-rounding}, we get that Algorithm~\ref{alg:hyper-eigs-rounding}
will output a unit vector which is orthogonal to all $x_i$ for $i \in [k-1]$ and 
\[ \Dc(x_k) \leq 192 k \log r\, \xi_k \mper \]
\\
\end{proofof}
}

%% file: sparsest_cut.tex
\section{Sparsest Cut with General Demands}
\label{sec:hyper-sparsestcut}

In this section, 
we study the Sparsest Cut with General Demands problem (defined in
Section~\ref{sec:sparest_overview}) and give an approximation algorithm for it
(Theorem~\ref{thm:hyper-sparsest-nonuniform}).

\begin{theorem}[Restatement of Theorem~\ref{thm:hyper-sparsest-nonuniform}]
There exists a randomized polynomial time algorithm that given 
an instance of the hypergraph Sparsest Cut problem with hypergraph  $H=(V,E,w)$ and $k$ demand pairs in 
$T = \{(\{s_i,t_i\}, D_i): i \in [k]\}$, 
outputs a set $S \subset V$ such that with high probability,
\[ \Phi(S) \leq \bigo{\sqrt{\log k \log r} \log \log k } \Phi_H ,\]
where $r = \max_{e \in E} \Abs{e} $.
\end{theorem}

\begin{proof}
We prove this theorem by giving an \sdp relaxation for this problem (SDP~\ref{sdp:hgensc})
and a rounding algorithm for it (Algorithm~\ref{alg:hgensc}).
We introduce a variable $\U$ for each vertex $u \in V$. 
Ideally, we would want all vectors $\U$ to be in the set $\set{0,1}$ so that we can identify the cut,
in which case $\max_{u,v \in e} \norm{\U - \V}^2 $ will indicate whether the edge $e$ is cut 
or not.
Therefore, our objective function will be $\sum_{e \in E} w(e) \max_{u,v \in e} \norm{\U - \V}^2$.
Next, we add (\ref{eq:hyper-gensc-1}) as a scaling constraint. 
Finally, we add $\ell_2^2$-triangle inequality constraints~(\ref{eq:hyper-gensc-3}) between all triplets of vertices, 
as all integral solutions of the relaxation will trivially satisfy this. 
Therefore, SDP~\ref{sdp:hgensc} is a relaxation
of the problem and its objective value is at most $\Phi_H$. 

\begin{mybox}
\begin{SDP}
\label{sdp:hgensc}
\[  \min \sum_{e \in E} w_e \max_{u,v \in e} \norm{ \U - \V  }^2   \]
\subjectto 
\begin{equation}
\label{eq:hyper-gensc-1}
\sum_{i \in [k]} D_i \cdot  \norm{\bar{s_i} - \bar{t_i}}^2   = 1  
\end{equation}
\begin{equation}
\label{eq:hyper-gensc-3}
\norm{\U - \V}^2 + \norm{\V - \bar{w}}^2 \geq  \norm{\U - \bar{w}}^2 \qquad \forall u,v,w \in V 
\end{equation}
\end{SDP}
\end{mybox}

Our main ingredient is the following result due to \cite{aln05}. 
\begin{fact}[\cite{aln05}]
\label{thm:aln}
Let $(V,d)$ be an arbitrary metric space, and let $U \subset V$ be any $k$-point subset.
If the space $(V,d)$ is a metric of the negative type, then there exists a $1$-Lipschitz map
$f : V \to \ell_2$ such that the map $f|_U : U \to \ell_2$ has distortion $\bigo{\sqrt{\log k} \log \log k}$.
\end{fact}

\begin{algorithm}[H]
\caption{Rounding Algorithm for Sparsest Cut}
\label{alg:hgensc}
\begin{algorithmic}[1]

\State Solve SDP~\ref{sdp:hgensc}.

\State Compute the map $f : (V, \ell_2^2) \to \R^n$ using Fact~\ref{thm:aln},
with $U$ being the set of vertices that appear in the demand pairs in $T$.

\State Sample $\vec{z} \sim \cN(0,1)^n$ and define $x \in \R^V$ such that $x(v) \defeq \inprod{\vec{z},f(v)}$ for each $v \in V$.

\State Arrange the vertices of $V$ as $v_1, \ldots, v_n$
such that $x(v_j) \leq x(v_{j+1})$  for each  $1\leq j \leq n-1$. 
Output the sparsest cut of the form 
\[ \paren{ \set{v_1, \ldots, v_i} , \set{v_{i+1}, \ldots, v_n }  } \mper  \]
\label{step:hgensc-four}
\end{algorithmic}
\end{algorithm}

Without loss of generality, we may assume that the map $f$ is such that $f|_U$ has the least distortion (on vertices in demand pairs) among all $1$-Lipschitz
maps $f:(V, \ell_2^2) \to \ell_2$ (\cite{aln05} gives a polynomial time algorithm to compute such a map.)
For the sake of brevity, let $\Lambda = \bigo{\sqrt{\log k} \log \log k}$ denote the distortion factor
guaranteed in Fact~\ref{thm:aln}.
Since SDP~\ref{sdp:hgensc} is a relaxation of $\Phi_H$, we also get that objective value of the 
\sdp is at most $\Phi_H$.  Suppose $x \in \R^V$ is the vector produced
by the rounding algorithm.

We next analyze the following quantity.  The numerator is related to  the objective function, and the denominator is related to
the expression in~(\ref{eq:hyper-gensc-1}):

\begin{align}
\label{eq:sparsest_x}
\varphi(x) := \frac{\sum_{e \in E} w_e \max_{u,v \in e} \Abs{x(u) - x(v)}}{\sum_{i \in [k]} D_i \cdot |x(s_i) - x(t_i)|}.
\end{align}

The following analysis is 
similar to the proof of Lemma~\ref{lem:hyper-eigs-rounding}.

For each edge $e$, obsever that for $u, v \in e$,
$x_u - x_v$ is a random variable having normal distribution
with mean 0 and variance $\norm{f(u) - f(v)}^2$.
Hence,
using Fact~\ref{fact:appGauss}~(2), we get 
\[ \Ex{\max_{u,v \in e} \Abs{x(u) - x(v)}} \leq 4 \sqrt{\log r} \max_{u,v \in e} \norm{f(u) - f(v) } \leq 
4 \sqrt{\log r} \max_{u,v \in e} \norm{\U - \V }^2, \]
where the last inequality follows because $f: (V, \ell_2^2) \rightarrow \ell_2$ is 1-Lipschitz.

The expectation of the
numerator of (\ref{eq:sparsest_x}) is

$\Ex{\sum_{e \in E} w_e \max_{u,v \in e} \Abs{x(u) - x(v)}}
\leq 4 \sqrt{\log r} \cdot \Phi_H$.

Using Markov's inequality, we have
\begin{equation}
\label{eq:hgensc-help-1}
 \Pr[ \sum_{e \in E} w_e \max_{u,v \in e} \Abs{x(u) - x(v)} \leq 96 \sqrt{\log r}\, \cdot \Phi_H ]  \geq 1 - \frac{1}{24} \mper     
\end{equation}

For the denominator,
observing that $x(s_i) - x(t_i)$
has a normal distribution with mean~0 and variance $\norm{f(s_i) - f(t_i)}^2$
and
a random variable $Z$ having
distribution $\cN(0,1)$ satisfies $\Ex{|Z|} = \sqrt{\frac{2}{\pi}}$,
we have

\[ \Ex{ \sum_{i \in [k]} D_i \cdot \Abs{x(s_i) - x(t_i)} }  
 = \sqrt{\frac{2}{\pi}}\sum_{i \in [k]} D_i \norm{ f(s_i) -f(t_i)} 
\geq \sqrt{\frac{2}{\pi}} \cdot \frac{1}{\Lambda} \sum_{i \in [k]} D_i \cdot \norm{\bar{s_i} - \bar{t_i}}^2  =
\sqrt{\frac{2}{\pi}} \cdot \frac{1}{\Lambda} , \]
where the inequality follows from the distortion of $f|_U$ as guaranteed
by Fact~\ref{thm:aln}, and
the last equality follows from (\ref{eq:hyper-gensc-1}).

We next prove a variant of Fact~\ref{lem:squaregaussian}.

\begin{claim}
\label{claim:abs_gauss}
Let $Y_1,\ldots, Y_m$ be normal random variables (that are not necessarily independent) 
having mean 0. Denote $R := \sum_i |Y_i|$.
Then,
\[ \Pr\left[R \geq \frac{1}{2} \Ex{R} \right] \geq \frac{1}{12} \mper \]
\end{claim}

\begin{proof}
For each $i$, let $\sigma_i^2 =  \Ex{Y_i}$.
Then, $\Ex{R} = \sqrt{\frac{2}{\pi}} \sum_i \sigma_i$.

Moreover, we have

$\Ex{R^2} = \sum_{i,j} \Ex{|Y_i| \cdot |Y_j|}
\leq \sum_{i,j} \sqrt{\Ex{Y_i^2} \cdot \Ex{Y_j^2}}
= \sum_{i,j} \sigma_i \sigma_j = \frac{\pi}{2} \cdot \Ex{R}^2$,

where the inequality follows from Cauchy-Schwarz.

Finally, using the Paley-Zygmund Inequality,
we have
\[ \Pr\left[ R \geq \frac{1}{2} \cdot \Ex{R} \right] \geq \left(\frac{1}{2}\right)^2 \cdot \frac{\Ex{R}^2}{\Ex{R^2}} \geq \frac{1}{12} \mper \]
\end{proof}

Hence, using Fact~\ref{lem:squaregaussian}, we get 
\begin{equation}
\label{eq:hgensc-help-2}
\Pr[ \sum_{i \in [k]} D_i \Abs{x(s_i) - x(t_i)} 
\geq \sqrt{\frac{1}{2\pi}} \cdot \frac{1}{\Lambda}] \geq \frac{1}{12}.
\end{equation}

Using (\ref{eq:hgensc-help-1}) and (\ref{eq:hgensc-help-2}), we get that 
with probability at least $\frac{1}{24}$,
\[\varphi(x) = \frac{\sum_e w_e \max_{u,v \in e} \Abs{x(u) - x(v)}}{ \sum_{i \in [k]} D_i \cdot \Abs{x(s_i) - x(t_i)} }
\leq \bigo{\sqrt{\log r}} \cdot\, \Lambda\, \Phi_H \mper \]

We next apply an analysis that is similar to Proposition~\ref{prop:hyper-1d}.
For $r \in \R$, define $S_r := \{v \in V: x(v) \leq r\}$.
Observe that if $r$ is sampled uniformly at random from the interval $[\min_v x(v), \max_v x(v)]$,
then two vertices $u$ and $v$ are separated by
the cut $(S_r, \overline{S_r})$ with probability proportional to
$|x(u) - x(v)|$.

Hence, an averaging argument implies
that there exists $r \in \R$ such that
$\Phi(S_r) \leq \varphi(x) \leq \bigo{\sqrt{\log k \log r}  \log \log k } \Phi_H$,
as required in the output
of Step~\ref{step:hgensc-four}.
\end{proof}

%% file: ack.tex
\paragraph{Acknowledgments}
The dispersion process associated with our Laplacian operator was suggested to us by Prasad Raghavendra
in the context of understanding vertex expansion in graphs, and was the
starting point of this project.
We would like to thank Ravi Kannan,  
Konstantin Makarychev, Yury Makarychev, 
Yuval Peres, 
Prasad Raghavendra, Nikhil Srivastava, Piyush Srivastava, Prasad Tetali, Santosh Vempala, 
David Wilson and Yi Wu for helpful discussions.
We would also like to thank Jingcheng Liu and Alistair Sinclair for
giving us valuable comments on an earlier version of this paper.

%% file: tensor.tex
\section{Hypergraph Tensor Forms}
\label{app:hyper-tensor}
Let $A$ be an $r$-tensor. 
For any suitable norm $\normb{\cdot}$, e.g. $\norm{.}^2_2$, $\norm{.}^r_r$, 
we define tensor eigenvalues as follows. 

\begin{definition}
We define $\lambda_1$, the largest eigenvalue of a tensor $A$  
as follows.
\[ \lambda_1 \defeq \max_{X \in \R^n} \frac{ \sum_{i_1, i_2, \ldots, i_r} A_{i_1 i_2 \ldots i_r} X_{i_1} X_{i_2} \ldots X_{i_r} }{\normb{X}}, \]

\[v_1 \defeq \argmax_{X \in \R^n} \frac{\sum_{i_1, i_2, \ldots, i_r} A_{i_1 i_2 \ldots i_r} X_{i_1} X_{i_2} \ldots X_{i_r}  }{\normb{X}}.
  \]

We inductively define successive eigenvalues $\lambda_2 \geq \lambda _3 \geq \ldots$ as follows.
\[ \lambda_k \defeq \max_{X \perp \set{v_1, \ldots, v_{k-1}} } 
	\frac{ \sum_{i_1, i_2, \ldots, i_r} A_{i_1 i_2 \ldots i_r} X_{i_1} X_{i_2} \ldots X_{i_r} }{\normb{X}}, \]
	
\[v_k \defeq \argmax_{x \perp \set{v_1, \ldots, v_{k-1}} }
	\frac{ \sum_{i_1, i_2, \ldots, i_r} A_{i_1 i_2 \ldots i_r} X_{i_1} X_{i_2} \ldots X_{i_r}  }{\normb{X}}.   \]

\end{definition}

Informally, the Cheeger's Inequality states that a graph has a sparse cut if and only if the 
gap between the two largest eigenvalues of the adjacency matrix is small; in particular, a graph
is disconnected if and only if its top two eigenvalues are equal.  In the case of the hypergraph tensors,
we show that there exist hypergraphs having 
no gap between many top eigenvalues while still being connected. This shows that the tensor 
eigenvalues are not relatable to expansion in a Cheeger-like manner.

\begin{proposition}
For any $k \in \N$, there exist connected hypergraphs such that $\lambda_1 = \ldots = \lambda_k$.
\end{proposition}

\begin{proof}
Let $r = 2^w$ for some $w \in \Z^+$. Let $H_1$ be a large enough complete $r$-uniform hypergraph. 
We construct $H_2$ from two copies of $H_1$, say $A$ and $B$, as follows. 
Let $a \in E(A)$ and $b \in E(B)$ be any two hyperedges.   
Let $a_1 \subset a $ (resp. $b_1 \subset b$) be a set of any $r/2$ vertices. 
We are now ready to define $H_2$.
\[ H_2 \defeq \left( V(H_1) \cup V(H_2), (E(H_1) \setminus \set{a}) \cup (E(H_2) \setminus \set{b}) 
\cup \set{ (a_1 \cup b_1), (a_2 \cup b_2)} \right)  \] 
Similarly, one can recursively define $H_i$ by joining two copies of $H_{i-1}$ (this can be done as  
long as $r > 2^{2i}$). The construction of $H_k$ can be viewed as a {\em hypercube of hypergraphs}. 

Let $A_H$ be the tensor form of hypergraph $H$.
For $H_2$, it is easily verified that $v_1 = \one$.
Let $X$ be the vector which has $+1$ on the vertices 
corresponding to $A$ and $-1$ on the vertices corresponding to $B$. 
By construction, for any hyperedge $\set{i_1, \ldots, i_r} \in E$
\[ X_{i_1} \ldots X_{i_r} = 1    \]
and therefore, 
\[  \frac{ \sum_{i_1, i_2, \ldots, i_r} A_{i_1 i_2 \ldots i_r} X_{i_1} X_{i_2} \ldots X_{i_r} }{\normb{X}} = \lambda_1 \mper  \]
Since $\inprod{X,\one} = 0$, we get $\lambda_2 = \lambda_1$ and $v_2 = X$.
Similarly, one can show that $\lambda_1 = \ldots = \lambda_k$ for $H_k$. 
This is in sharp contrast to the fact that $H_k$ is, by construction,  a connected hypergraph. 
\end{proof}

%% file: example.tex
\section{Examples}
\label{sec:example}

We give examples of hypergraphs to show that
some properties are not satisfied.  For convenience,
we consider the properties in terms of the weighted space.
We remark that the examples could also be formulated
equivalently in the normalized space.
In our examples, the procedural minimizers are discovered by trial-and-error using programs.  However, we only
describe how to use Mathematica to 
verify them.  Our source code
can be downloaded at the following link:

\url{http://i.cs.hku.hk/~algth/project/hyper_lap/main.html}

\noindent \textbf{Verifying Procedural Minimizers.}
In our examples, we need to verify that
we have the correct value for
$\gamma_k := \min_{\vec{0} \neq f \perp_w \{f_1, f_2, \ldots, f_{k-1}\}} \D_w(f)$,
and a certain non-zero vector $f_k$ attains the minimum.

We first check that $f_k$ is perpendicular
to $\{f_1, \ldots, f_{k-1}\}$ in the weighted space,
and $\D_w(f_k)$ equals $\gamma_k$.

Then, it suffices to check that
for all $\vec{0} \neq f \perp_w \{f_1, f_2, \ldots, f_{k-1}\}$,
$\D_w(f) \geq \gamma_k$.  As the numerator in  the definition of $\D_w(f)$ involves
the maximum operator, we use a program to consider all cases of the relative
order of the vertices with respect to $f$.

For each permutation $\sigma: [n] \ra V$,
for $e \in E$, we define $S_\sigma(e) := \sigma(\max \{i: \sigma(i) \in e\})$
and $I_\sigma(e) := \sigma(\min \{i: \sigma(i) \in e\})$.

We consider the mathematical program
$P(\sigma) := \min \sum_{e \in E} w_e  \cdot (f(S_\sigma(e)) - f(I_\sigma(e)))^2
- \gamma_k \cdot \sum_{u \in V} w_u f(u)^2$ subject to
$f(\sigma(n)) \geq f(\sigma(n-1)) \geq \cdots f(\sigma(1))$
and $\forall i \in [k-1], \langle f_i, f \rangle = 0$.
Since the objective function is a polynomial, and all constraints
are linear, the Mathematica function \textsf{Minimize} can solve the
program.

Moreover, the following two statements are equivalent.
\begin{compactitem}
\item[1.] For all $\vec{0} \neq f \perp_w \{f_1, f_2, \ldots, f_{k-1}\}$,
$\D_w(f) \geq \gamma_k$.
\item[2.] For all permutations $\sigma$, $P(\sigma) \geq 0$.
\end{compactitem}

Hence, to verify the first statement, 
it suffices to use Mathematica to solve $P(\sigma)$ for all permutations $\sigma$.

\begin{example}
\label{ex:gamma_nonunique}
The sequence $\{\gamma_k\}$ generated
by the procedural minimizers is not unique.
\end{example}

\begin{proof}
Consider the following hypergraph with $5$ vertices and $5$ hyperedges each with unit weight.
\begin{figure}[H]
	\centering
	\begin{tikzpicture}
	\node (v5) at (5,1) {};
	\node (v1) at (4,1) {};
	\node (v2) at (3,1) {};
	\node (v3) at (2,1) {};
	\node (v4) at (1,1) {};

	\begin{scope}[fill opacity=0]
	\filldraw[fill=blue!70] ($(v4)+(-0.5,0)$)
	to[out=90,in=180] ($(v4)+(0,0.45)$)
	to[out=0,in=90] ($(v4)+(0.5,0)$)
	to[out=270,in=0] ($(v4)+(0,-0.45)$)
	to[out=180,in=270] ($(v4)+(-0.5,0)$);
	\filldraw[fill=blue!70] ($(v4)+(-0.55,0)$)
	to[out=90,in=180] ($(v4)+(0.5,0.75)$)
	to[out=0,in=90] ($(v4)+(1.5,0)$)
	to[out=270,in=0] ($(v4)+(0.5,-0.75)$)
	to[out=180,in=270] ($(v4)+(-0.55,0)$);
	\filldraw[fill=blue!70] ($(v4)+(-0.6,0)$)
	to[out=90,in=180] ($(v4)+(1,1)$)
	to[out=0,in=90] ($(v4)+(2.5,0)$)
	to[out=270,in=0] ($(v4)+(1,-1)$)
	to[out=180,in=270] ($(v4)+(-0.6,0)$);
	\filldraw[fill=blue!70] ($(v4)+(-0.65,0)$)
	to[out=90,in=180] ($(v4)+(1.5,1.25)$)
	to[out=0,in=90] ($(v4)+(3.5,0)$)
	to[out=270,in=0] ($(v4)+(1.5,-1.25)$)
	to[out=180,in=270] ($(v4)+(-0.65,0)$);
	\filldraw[fill=blue!70] ($(v4)+(-0.7,0)$)
	to[out=90,in=180] ($(v4)+(2,1.5)$)
	to[out=0,in=90] ($(v4)+(4.5,0)$)
	to[out=270,in=0] ($(v4)+(2,-1.5)$)
	to[out=180,in=270] ($(v4)+(-0.7,0)$);
	\end{scope}

	\foreach \v in {1,2,...,5} {
		\fill (v\v) circle (0.1);
	}
	
	\fill (v1) node [below] {$d$};
	\fill (v2) circle (0.1) node[inner sep = 6pt] [below] {$c$};
	\fill (v3) circle (0.1) node [below] {$b$};
	\fill (v4) circle (0.1) node[inner sep = 6pt] [below] {$a$};
	\fill (v5) circle (0.1) node[inner sep = 6pt] [below] {$e$};
	
	\node at ($(v4)+(-0.2,0.2)$) {$e_1$};
	\node at ($(v3)+(-0.2,0.2)$) {$e_2$};
	\node at ($(v2)+(-0.2,0.2)$) {$e_3$};
	\node at ($(v1)+(-0.2,0.2)$) {$e_4$};
	\node at ($(v5)+(-0.2,0.2)$) {$e_5$};
	\end{tikzpicture}	
\end{figure}

We have verified that different minimizers for $\gamma_2$
can lead to different values for $\gamma_3$.

\begin{table}[H]
	\centering
	\begin{tabular}{c|cc|cc}
		\toprule
		{$i$} & $\gamma_i$ & $f_i^\T$ & $\gamma_i'$ & $f_i'^\T$\\
		\hline
		1 & 0 & $(1,1,1,1,1)$ & 0 & $(1,1,1,1,1)$ \\
		2 & 5/6 & $(1,1,1,-4,-4)$ & 5/6 & $(2,2,-3,-3,-3)$ \\
		3 & 113/99 & $(2,2,-6,3,-6)$ & 181/165 & $(4,-5,-5,5,5)$ \\
		\bottomrule
	\end{tabular}
	
\end{table}
\end{proof}

\begin{example}
\label{ex:xg}
There exists a hypergraph such that $\xi_2 < \gamma_2$.
\end{example}

\begin{proof}
Consider the following hypergraph $H = (V,E)$ with
$V = \{a, b, c, d\}$ and $E = \{e_i : i \in [5]\}$.
For $i \neq 3$, edge $e_i$ has weight 1, and edge $e_3$ has weight 2.
Observe that every vertex has weight $3$.

\begin{figure}[H]
	\centering
	\begin{tikzpicture}
	\node (v1) at (0,2) {};
	\node (v2) at (2,2.8) {};
	\node (v3) at (2,1.2) {};
	\node (v4) at (4.5,2) {};
	
	\begin{scope}[fill opacity=0]
	\filldraw ($(v2)+(+0.5,+0.5)$) 
	to[out=-45,in=45] ($(v3)+(0.5,-0.5)$) 
	to[out=225,in=270] ($(v1)+(-1.2,0)$)
	to[out=90,in=135] ($(v2)+(+0.5,+0.5)$);
	\end{scope}
	
	\path[every node/.style={font=\sffamily\small}]
	(v1) edge node [above left] {$e_1$} (v2)
	(v2) edge node [above right] {$e_2$} (v4)
	(v3) edge node [below right] {$e_3$} (v4)
	(v1) edge [loop left] node  {$e_4$} (v1);
	
	\foreach \v in {1,2,3,4} {
		\fill (v\v) circle (0.1);
	}
	
	\fill (v1) circle (0.1) node [below right] {$a$};
	\fill (v2) circle (0.1) node [above left] {$b$};
	\fill (v3) circle (0.1) node [below right] {$c$};
	\fill (v4) circle (0.1) node [below right] {$d$};
	
	\node at (0.8,0.8) {$e_5$};
	\end{tikzpicture}
\end{figure}

We can verify 
that $\gamma_2 = \frac{2}{3}$
with the corresponding vector $f_2 := (1, 1, -1, -1)^\T$.

Recall that $\xi_2=\min_{g_1,g_2} \max_{i\in [2]} \D_w(g_i)$,
where the minimum is over all non-zero $g_1$ and $g_2$ such that
$g_1 \perp_w g_2$.
We can verify that
$\xi_2 \leq \frac{1}{3}$ by considering the
the two orthogonal vectors 
$g_1=(0,0,1,1)^{\T}$ and $g_2=(1,1,0,0)^{\T}$
in the weighted space.
\end{proof}

\begin{example}[Issues with Distributing Hyperedge Weight Evenly]
\label{eg:Louis}
Suppose $\overline{\Lo}_w$ is the operator on the weighted space
that is derived from the Figure~\ref{fig:hyper_diffusion}
by distributing the weight $w_e$ evenly among $S_e(f) \times I_e(f)$.
Then, there exists a hypergraph
such that any minimizer $f_2$ attaining
$\gamma_2 := \min_{\vec{0} \neq f \perp_w \vec{1}} \D_w(f)$
is not an eigenvector of $\overline{\Lo}_w$ or
even $\Pi^w_{\vec{1}^{\perp_w}} \overline{\Lo}_w$.
\end{example}

\begin{proof}
We use the same hypergraph as in Example~\ref{ex:xg}.
Recall that
$\gamma_2 = \frac{2}{3}$
with the corresponding vector $f_2 := (1, 1, -1, -1)^\T$.

We next show that $f_2$ is the only minimizer, up to scalar multiplication,
attaining $\gamma_2$.

According to the definition,
$$\gamma_2 = \min_{(a,b,c,d)\perp_w 1} 
\frac{(a-b)^2+(b-d)^2+2(c-d)^2+\max_{x,y\in e_5}(x-y)^2 }{3(a^2+b^2+c^2+d^2)}.$$

Without loss of generality, we only need to consider the following three cases:

\begin{enumerate}
\item[1.]$a\geq b \geq c$: Then, by substituting $a = -b-c-d$,
\begin{align*}
& \frac{(a-b)^2+(b-d)^2+2(c-d)^2+ (a-c)^2}{3(a^2+b^2+c^2+d^2)} \geq \frac{2}{3} \\
\iff & (c-d)^2+2(b+c)^2 \geq 0,
\end{align*}
and the equality is attained only when $a=b=-c=-d$.
\item[2.]$a\geq c\geq b$: Then, by substituting $d = -a-b-c$,
\begin{align*}
&  \frac{(a-b)^2+(b-d)^2+2(c-d)^2+ (a-b)^2}{3(a^2+b^2+c^2+d^2)} \geq \frac{2}{3} \\
\iff & (a+2b+c)^2+8c^2+4(a-c)(c-b) \geq 0,
\end{align*}
and the equality cannot be attained.
\item[3.]$b\geq a\geq c$: Then, by substituting $d = -a-b-c$,
\begin{align*}
&  \frac{(a-b)^2+(b-d)^2+2(c-d)^2+ (b-c)^2}{3(a^2+b^2+c^2+d^2)} \geq \frac{2}{3} \\
\iff & 4(b+c)^2+2(a+c)^2+2(b-a)(a-c) \geq 0,
\end{align*}
and the equality is attained only when $a=b=-c=-d$.
\end{enumerate}

Therefore, all minimizers attaining $\gamma_2$ must be in $\spn(f_2)$.

We next showt that $f_2$ is not an eigenvector of $\Pi^w_{\vec{1}^{\perp_w}} \overline{\Lo}_w$.
Observe that only the hyperedge $e_5 = \{a, b, c\}$
involves more than 2 vertices.  In this case,
the weight of $e_5$ is distributed evenly between
$\{a, c\}$ and $\{b,c\}$.  All other edges keep their weights.
Hence, the resulting weighted adjacency matrix $\overline{A}$ and
$\I - \Wm \overline{A}$ are as follows:

$\overline{A} = \begin{pmatrix}
\frac{3}{2} & 1 & \frac{1}{2} & 0 \\
1 & \frac{1}{2} & \frac{1}{2} & 1 \\
\frac{1}{2} & \frac{1}{2} & 0 & 2 \\
0 & 1 & 2 & 0
\end{pmatrix}$
and 
$\I - \Wm \overline{A} = \begin{pmatrix*}[r]
\frac{1}{2} & -\frac{1}{3} & -\frac{1}{6} & 0 \\
-\frac{1}{3} & \frac{5}{6} & -\frac{1}{6} & -\frac{1}{3} \\
-\frac{1}{6} & -\frac{1}{6} & 1 & -\frac{2}{3} \\
0 & -\frac{1}{3} & -\frac{2}{3} & 1
\end{pmatrix*}.$

Hence, $\overline{\Lo}_w f_2 = (\I - \Wm \overline{A}) f_2
= (\frac{1}{3}, 1, -\frac{2}{3}, - \frac{2}{3})^\T \notin \spn(f_2)$.
Moreover, $\Pi^w_{\vec{1}^{\perp_w}} \overline{\Lo}_w f_2 =
(\frac{1}{3}, 1, -\frac{2}{3}, -\frac{2}{3})^\T
\notin \spn(f_2)$.

In comparison, in our approach,
since $b$ is already connected to $d$ with edge $e_2$ of weight 1,
it follows that the weight of $e_5$ should all go to the pair $\{a,c\}$.  Hence, the resulting adjacency matrix is:

$A = \begin{pmatrix}
1 & 1 & 1 & 0 \\
1 & 1 & 0 & 1 \\
1 & 0 & 0 & 2 \\
0 & 1 & 2 & 0
\end{pmatrix}.$

One can verify that $\Lo_w f_2 = (\I - \Wm A)f_2 = \frac{2}{3} f_2$, as claimed in Theorem~\ref{th:hyper_lap}.
\end{proof}

\begin{example}[Third minimizer not eigenvector of Laplacian]
\label{eg:gamma3}
There exists a hypergraph such that for all procedural minimizers $\{(f_i, \gamma_i)\}_{i\in [3]}$
of $\D_w$, the vector~$f_3$ 
is not an eigenvector of $\Lo_w$ or even
$\Pi^w_{F_2^{\perp_w}} \Lo_w$,
where $\Lo_w$ is the operator on the weighted space
defined in Lemma~\ref{lemma:define_lap},
and $F_2 := \{f_1, f_2\}$.
\end{example}

\begin{proof}
Consider the following hypergraph with $4$ vertices and $2$ hyperedges each with unit weight.
\begin{figure}[H]
	\centering
	\begin{tikzpicture}
	\node (v1) at (0,0) {};
	\node (v2) at (1.5,0.866) {};
	\node (v3) at (1.5,-0.866) {};
	\node (v4) at (-2,0) {};
	
	\draw (1,0) circle (1.725);	
	
	\path[every node/.style={font=\sffamily\small}]
		(v1) edge node [above left] {$e_1$} (v4);
	
	\foreach \v in {1,2,3,4} {
		\fill (v\v) circle (0.1);
	}
	
	\fill (v1) circle (0.1) node [below right] {$b$};
	\fill (v2) circle (0.1) node [below right] {$c$};
	\fill (v3) circle (0.1) node [below right] {$d$};
	\fill (v4) circle (0.1) node [below right] {$a$};
	
	\node at (0.6,-1) {$e_2$};
	\end{tikzpicture}
\end{figure}	

We can verify the first 3 procedural minimizers.

\begin{table}[H]
	\centering
	\begin{tabular}{l|cc}
		\toprule
		{$i$} & $\gamma_i$ & $f_i^\T$ \\
		\hline
		1 & 0 & $(1,1,1,1)$  \\
		2 & $\frac{5-\sqrt{5}}{4}$ & $(\sqrt{5}-1,\frac{3-\sqrt{5}}{2},-1,-1)$  \\
		3 & $\frac{11+\sqrt{5}}{8}$ & $(\sqrt{5}-1, -1, 4-\sqrt{5}, -1)$  \\
		3$'$& $\frac{11+\sqrt{5}}{8}$ & $(\sqrt{5}-1, -1, -1, 4-\sqrt{5})$ \\
		\bottomrule
	\end{tabular}
\end{table}

We next show that $f_3$ and $f_{3'}$ are the only minimizers, up to scalar multiplication,
attaining $\gamma_3$.

According to the definition,
$$\gamma_2 = \min_{(a,b,c,d)\perp_w 1}\frac{(a-b)^2+\max_{x,y \in e_2}(x-y)^2}{a^2+2b^2+c^2+d^2}.$$

Observe that $c$ and $d$ are symmetric, we only need to consider the following two cases,
\begin{enumerate}
	\item[1.]$c\geq b \geq d$: Then, by substituting $a=-2b-c-d$,
	\begin{align*}
		& \frac{(a-b)^2+(c-d)^2}{a^2+2b^2+c^2+d^2} \geq 1 \\
		\iff & 5b^2 + 2(c-b)(b-d) \geq 0.
	\end{align*}
	\item[2.]$b \geq c \geq d$: Then, by substituting $a=-2b-c-d$,
	\begin{align*}
		& \frac{(a-b)^2+(b-d)^2}{a^2+2b^2+c^2+d^2} \geq \frac{5-\sqrt{5}}{4} \\
		\iff & (5+3\sqrt{5})b^2+(\sqrt{5}-3)c^2+(\sqrt{5}-1)d^2 + (2\sqrt{5}+2)bc+(2\sqrt{5}-2)bd+(\sqrt{5}-1)cd \geq 0.
	\end{align*}
	Let $f(b,c,d)$ denotes the function above.
	Since $f$ is a quadratic function of $c$ and the coefficient of $c^2$ is negative, the minimum must be achieved when $c=b \text{ or } d$.
	In other words,$f(b,c,d)\geq \min \{f(b,b,d), f(b,d,d)\}.$ Note that
	\begin{align*}
		& f(b,b,d) = (6\sqrt{5}+4)b^2+(3\sqrt{5}-3)bd+(\sqrt{5}-1)d^2 \geq 0 \\
		\text{and } & f(b,d,d) = (5+3\sqrt{5})b^2 +4\sqrt{5}bd+(3\sqrt{5}-5)d^2 \geq 0.
	\end{align*}
	and the equality holds only when $c=d=-\frac{3+\sqrt{5}}{2}b$.
\end{enumerate}
Therefore, $\gamma_2 = \frac{5-\sqrt{5}}{4}, f_2^{T}=(\sqrt{5}-1,\frac{3-\sqrt{5}}{2},-1,-1)$.

Now we are ready to calculate $\gamma_3.$
$$\gamma_3 = \min_{(a,b,c,d)\perp_w 1,f_2}\frac{(a-b)^2+\max_{x,y \in e_2}(x-y)^2}{a^2+2b^2+c^2+d^2}.$$
Note that,
$$(a,b,c,d)\perp_w \vec{1},f_2 \iff
\begin{cases}
a+2d+c+d = 0 \\
(\sqrt{5}-1)a+(3-\sqrt{5})b-c-d=0
\end{cases} \iff
\begin{cases}
a = (1-\sqrt{5})b \\
c+d = (\sqrt{5}-3)b
\end{cases}
$$
\begin{enumerate}
	\item[1.]$c\geq b \geq d$: which is equivalent to $c \geq -\frac{\sqrt{5}+3}{4}(c+d) \geq d,$ then
	\begin{align*}
		& \frac{(a-b)^2+(c-d)^2}{a^2+2b^2+c^2+d^2} \geq \frac{11+\sqrt{5}}{8} \\
		\iff & (c-\frac{\sqrt{5}+3}{4}(c+d))(d-\frac{\sqrt{5}+3}{4}(c+d)) \leq 0.
	\end{align*} 
	\item[2.]$b\geq c \geq d$: which is equivalent to $(4-\sqrt{5})b+d \geq 0 \geq (3-\sqrt{5})b+2d,$ then
	\begin{align*}
		& \frac{(a-b)^2+(b-d)^2}{a^2+2b^2+c^2+d^2} \geq \frac{11+\sqrt{5}}{8} \\
		\iff & ((4-\sqrt{5})b + d)((3+\sqrt{5})((3-\sqrt{5})b+2d)- (\sqrt{5}-1)((4-\sqrt{5})b+d)) \leq 0.
	\end{align*}
\end{enumerate}
Therefore, $\gamma_3=\frac{11+\sqrt{5}}{8}$, and the corresponding $f_3^{T}=(\sqrt{5}-1, -1, 4-\sqrt{5}, -1) \text{ or } (\sqrt{5}-1, -1, -1, 4-\sqrt{5}).$

We let $f= f_3 = (\sqrt{5}-1, -1, 4-\sqrt{5}, -1)^\T$,
and we apply the procedure described in Lemma~\ref{lemma:define_lap} to compute $\Lo_w f$.

Observe that $w_a = w_c = w_d = 1$ and $w_b=2$,
and $f(b) = f(d) < f(a) < f(c)$.

For edge $e_1$, $\Delta_1 = f(a) - f(b) = \sqrt{5}$ and $c_1 = w_1 \cdot \Delta_1 = \sqrt{5}$.
For edge $e_2$, $\Delta_2 = f(c) - f(b) = 5 - \sqrt{5}$,
and $c_2 = w_2 \cdot \Delta_2 = 5 - \sqrt{5}$.
Hence, $r_a = - \frac{c_1}{w_a}$, $r_c = - \frac{c_2}{w_c}$,
and $r_b = r_d = \frac{c_1+c_2}{w_b + w_d}$.

Therefore, $\Lo_w f = - r = 
(\sqrt{5}, - \frac{5}{3}, 5 - \sqrt{5},
- \frac{5}{3})^\T$.

Moreover, $\Pi^w_{F_2^{\perp_w}} \Lo_w f =
(-\frac{1}{2} + \frac{7}{6} \cdot \sqrt{5},
-\frac{4}{3} - \frac{1}{6} \cdot \sqrt{5},
\frac{59}{12} - \frac{11}{12} \cdot \sqrt{5},
-\frac{7}{4} + \frac{1}{12} \cdot \sqrt{5})^\T
\notin \spn(f)$.

The case when $f_3 = (\sqrt{5}-1, -1, -1, 4-\sqrt{5})^\T$
is similar, with the roles of $c$ and $d$ reversed.
\end{proof}

%% file: main.bbl
\newcommand{\etalchar}[1]{$^{#1}$}
\begin{thebibliography}{DBH{\etalchar{+}}06}

\bibitem[ABS10]{abs10}
Sanjeev Arora, Boaz Barak, and David Steurer.
\newblock Subexponential algorithms for unique games and related problems.
\newblock In {\em Foundations of Computer Science (FOCS), 2010 51st Annual IEEE
  Symposium on}, pages 563--572. IEEE, 2010.

\bibitem[AC88]{ac88}
Noga Alon and Fan~RK Chung.
\newblock Explicit construction of linear sized tolerant networks.
\newblock {\em Annals of Discrete Mathematics}, 38:15--19, 1988.

\bibitem[AK95]{ak95}
Charles~J Alpert and Andrew~B Kahng.
\newblock Recent directions in netlist partitioning: a survey.
\newblock {\em Integration, the VLSI journal}, 19(1):1--81, 1995.

\bibitem[ALN08]{aln05}
Sanjeev Arora, James Lee, and Assaf Naor.
\newblock Euclidean distortion and the sparsest cut.
\newblock {\em Journal of the American Mathematical Society}, 21(1):1--21,
  2008.

\bibitem[Alo86]{a86}
Noga Alon.
\newblock Eigenvalues and expanders.
\newblock {\em Combinatorica}, 6(2):83--96, 1986.

\bibitem[AM85]{am85}
Noga Alon and V.~D. Milman.
\newblock $\lambda_{\mbox{1}}$, isoperimetric inequalities for graphs, and
  superconcentrators.
\newblock {\em J. Comb. Theory, Ser. B}, 38(1):73--88, 1985.

\bibitem[ARV09]{arv09}
Sanjeev Arora, Satish Rao, and Umesh Vazirani.
\newblock Expander flows, geometric embeddings and graph partitioning.
\newblock {\em Journal of the ACM (JACM)}, 56(2):5, 2009.

\bibitem[BBC{\etalchar{+}}15]{BalalauBCGS15}
Oana~Denisa Balalau, Francesco Bonchi, T.{-}H.~Hubert Chan, Francesco Gullo,
  and Mauro Sozio.
\newblock Finding subgraphs with maximum total density and limited overlap.
\newblock In {\em Proceedings of the Eighth {ACM} International Conference on
  Web Search and Data Mining, {WSDM} 2015, Shanghai, China, February 2-6,
  2015}, pages 379--388, 2015.

\bibitem[BHT00]{bht00}
Sergey Bobkov, Christian Houdr{\'e}, and Prasad Tetali.
\newblock $\lambda_{\infty}$ vertex isoperimetry and concentration.
\newblock {\em Combinatorica}, 20(2):153--172, 2000.

\bibitem[BS94]{bs94}
Stephen~T Barnard and Horst~D Simon.
\newblock Fast multilevel implementation of recursive spectral bisection for
  partitioning unstructured problems.
\newblock {\em Concurrency: Practice and Experience}, 6(2):101--117, 1994.

\bibitem[BTL84]{bt84}
Sandeep~N Bhatt and Frank Thomson~Leighton.
\newblock A framework for solving vlsi graph layout problems.
\newblock {\em Journal of Computer and System Sciences}, 28(2):300--343, 1984.

\bibitem[BV09]{bv09}
S~Charles Brubaker and Santosh~S Vempala.
\newblock Random tensors and planted cliques.
\newblock In {\em Approximation, Randomization, and Combinatorial Optimization.
  Algorithms and Techniques}, pages 406--419. Springer, 2009.

\bibitem[CA99]{ca99}
Umit~V Catalyurek and Cevdet Aykanat.
\newblock Hypergraph-partitioning-based decomposition for parallel
  sparse-matrix vector multiplication.
\newblock {\em Parallel and Distributed Systems, IEEE Transactions on},
  10(7):673--693, 1999.

\bibitem[CD12]{cd12}
Joshua Cooper and Aaron Dutle.
\newblock Spectra of uniform hypergraphs.
\newblock {\em Linear Algebra and its Applications}, 436(9):3268--3292, 2012.

\bibitem[CDP10]{cdp10}
L~Elisa Celis, Nikhil~R Devanur, and Yuval Peres.
\newblock Local dynamics in bargaining networks via random-turn games.
\newblock In {\em Internet and Network Economics}, pages 133--144. Springer,
  2010.

\bibitem[Cha00]{Charikar00}
Moses Charikar.
\newblock Greedy approximation algorithms for finding dense components in a
  graph.
\newblock In {\em Approximation Algorithms for Combinatorial Optimization,
  Third International Workshop, {APPROX} 2000, Saarbr{\"{u}}cken, Germany,
  September 5-8, 2000, Proceedings}, pages 84--95, 2000.

\bibitem[Chu93]{c93}
F~Chung.
\newblock The laplacian of a hypergraph.
\newblock {\em Expanding graphs (DIMACS series)}, pages 21--36, 1993.

\bibitem[Chu97]{chung97}
Fan Chung.
\newblock {\em Spectral Graph Theory}.
\newblock American Mathematical Society, 1997.

\bibitem[CTZ15]{ChanTZ15}
T.{-}H.~Hubert Chan, Zhihao~Gavin Tang, and Chenzi Zhang.
\newblock Spectral properties of laplacian and stochastic diffusion process for
  edge expansion in hypergraphs.
\newblock {\em CoRR}, abs/1510.01520, 2015.

\bibitem[DBH{\etalchar{+}}06]{dbh06}
Karen~D Devine, Erik~G Boman, Robert~T Heaphy, Rob~H Bisseling, and Umit~V
  Catalyurek.
\newblock Parallel hypergraph partitioning for scientific computing.
\newblock In {\em Parallel and Distributed Processing Symposium, 2006. IPDPS
  2006. 20th International}, pages 10--pp. IEEE, 2006.

\bibitem[Din07]{d07}
Irit Dinur.
\newblock The pcp theorem by gap amplification.
\newblock {\em Journal of the ACM (JACM)}, 54(3):12, 2007.

\bibitem[EN14]{en14}
Alina Ene and Huy~L Nguyen.
\newblock From graph to hypergraph multiway partition: Is the single threshold
  the only route?
\newblock In {\em Algorithms-ESA 2014}, pages 382--393. Springer Berlin
  Heidelberg, 2014.

\bibitem[FHL08]{fhl08}
Uriel Feige, MohammadTaghi Hajiaghayi, and James~R Lee.
\newblock Improved approximation algorithms for minimum weight vertex
  separators.
\newblock {\em SIAM Journal on Computing}, 38(2):629--657, 2008.

\bibitem[FW95]{fw95}
Joel Friedman and Avi Wigderson.
\newblock On the second eigenvalue of hypergraphs.
\newblock {\em Combinatorica}, 15(1):43--65, 1995.

\bibitem[GGLP00]{gglp00}
Patrick Girard, L~Guiller, C~Landrault, and Serge Pravossoudovitch.
\newblock Low power bist design by hypergraph partitioning: methodology and
  architectures.
\newblock In {\em Test Conference, 2000. Proceedings. International}, pages
  652--661. IEEE, 2000.

\bibitem[HL95]{hl95}
Bruce Hendrickson and Robert Leland.
\newblock An improved spectral graph partitioning algorithm for mapping
  parallel computations.
\newblock {\em SIAM Journal on Scientific Computing}, 16(2):452--469, 1995.

\bibitem[HL13]{hl13}
Christopher~J Hillar and Lek-Heng Lim.
\newblock Most tensor problems are np-hard.
\newblock {\em Journal of the ACM (JACM)}, 60(6):45, 2013.

\bibitem[HLW06]{hlw06}
Shlomo Hoory, Nathan Linial, and Avi Wigderson.
\newblock Expander graphs and their applications.
\newblock {\em Bulletin of the American Mathematical Society}, 43(4):439--561,
  2006.

\bibitem[HQ13]{hq13}
Shenglong Hu and Liqun Qi.
\newblock The laplacian of a uniform hypergraph.
\newblock {\em Journal of Combinatorial Optimization}, pages 1--36, 2013.

\bibitem[HQ14]{hq14}
Shenglong Hu and Liqun Qi.
\newblock The eigenvectors associated with the zero eigenvalues of the
  laplacian and signless laplacian tensors of a uniform hypergraph.
\newblock {\em Discrete Applied Mathematics}, 169:140--151, 2014.

\bibitem[KAKS99]{kaks99}
George Karypis, Rajat Aggarwal, Vipin Kumar, and Shashi Shekhar.
\newblock Multilevel hypergraph partitioning: applications in vlsi domain.
\newblock {\em Very Large Scale Integration (VLSI) Systems, IEEE Transactions
  on}, 7(1):69--79, 1999.

\bibitem[KKL14]{kkl14}
Tali Kaufman, David Kazhdan, and Alexander Lubotzky.
\newblock Ramanujan complexes and bounded degree topological expanders.
\newblock In {\em Foundations of Computer Science (FOCS), 2014 IEEE 55th Annual
  Symposium on}, pages 484--493. IEEE, 2014.

\bibitem[KLL{\etalchar{+}}13]{DBLP:conf/stoc/KwokLLGT13}
Tsz~Chiu Kwok, Lap~Chi Lau, Yin~Tat Lee, Shayan~Oveis Gharan, and Luca
  Trevisan.
\newblock Improved cheeger's inequality: analysis of spectral partitioning
  algorithms through higher order spectral gap.
\newblock In {\em STOC}, pages 11--20, 2013.

\bibitem[Lei80]{l80}
Charles~E Leiserson.
\newblock Area-efficient graph layouts.
\newblock In {\em Foundations of Computer Science, 1980., 21st Annual Symposium
  on}, pages 270--281. IEEE, 1980.

\bibitem[LM12]{lm12}
John Lenz and Dhruv Mubayi.
\newblock Eigenvalues and quasirandom hypergraphs.
\newblock {\em arXiv preprint arXiv:1208.4863}, 2012.

\bibitem[LM13]{lm13b}
John Lenz and Dhruv Mubayi.
\newblock Eigenvalues of non-regular linear quasirandom hypergraphs.
\newblock {\em arXiv preprint arXiv:1309.3584}, 2013.

\bibitem[LM14a]{lm14}
Anand Louis and Konstantin Makarychev.
\newblock Approximation algorithm for sparsest k-partitioning.
\newblock In {\em SODA}, pages 1244--1255. SIAM, 2014.

\bibitem[LM14b]{lm14b}
Anand Louis and Yury Makarychev.
\newblock {Approximation Algorithms for Hypergraph Small Set Expansion and
  Small Set Vertex Expansion}.
\newblock In {\em Approximation, Randomization, and Combinatorial Optimization.
  Algorithms and Techniques (APPROX/RANDOM 2014)}, volume~28, pages 339--355,
  2014.

\bibitem[LM15]{lm15}
John Lenz and Dhruv Mubayi.
\newblock Eigenvalues and linear quasirandom hypergraphs.
\newblock In {\em Forum of Mathematics, Sigma}, volume~3, page~e2. Cambridge
  Univ Press, 2015.

\bibitem[LOT12]{lot12}
James~R Lee, Shayan OveisGharan, and Luca Trevisan.
\newblock Multi-way spectral partitioning and higher-order cheeger
  inequalities.
\newblock In {\em Proceedings of the 44th symposium on Theory of Computing},
  pages 1117--1130. ACM, 2012.

\bibitem[Lou15]{Louis15}
Anand Louis.
\newblock Hypergraph markov operators, eigenvalues and approximation
  algorithms.
\newblock In {\em Proceedings of the Forty-Seventh Annual {ACM} on Symposium on
  Theory of Computing, {STOC} 2015, Portland, OR, USA, June 14-17, 2015}, pages
  713--722, 2015.

\bibitem[LRTV11]{lrtv11}
Anand Louis, Prasad Raghavendra, Prasad Tetali, and Santosh Vempala.
\newblock Algorithmic extensions of cheeger's inequality to higher eigenvalues
  and partitions.
\newblock In {\em Approximation, Randomization, and Combinatorial Optimization.
  Algorithms and Techniques}, pages 315--326. Springer, 2011.

\bibitem[LRTV12]{lrtv12}
Anand Louis, Prasad Raghavendra, Prasad Tetali, and Santosh Vempala.
\newblock Many sparse cuts via higher eigenvalues.
\newblock In {\em Proceedings of the 44th symposium on Theory of Computing},
  pages 1131--1140. ACM, 2012.

\bibitem[LRV13]{lrv13}
Anand Louis, Prasad Raghavendra, and Santosh Vempala.
\newblock The complexity of approximating vertex expansion.
\newblock In {\em Foundations of Computer Science (FOCS), 2013 IEEE 54th Annual
  Symposium on}, pages 360--369. IEEE, 2013.

\bibitem[LT80]{lt80}
Richard~J Lipton and Robert~Endre Tarjan.
\newblock Applications of a planar separator theorem.
\newblock {\em SIAM journal on computing}, 9(3):615--627, 1980.

\bibitem[Mer69]{merton1969lifetime}
Robert~C Merton.
\newblock Lifetime portfolio selection under uncertainty: The continuous-time
  case.
\newblock {\em The review of Economics and Statistics}, pages 247--257, 1969.

\bibitem[Mer71]{merton1971optimum}
Robert~C Merton.
\newblock Optimum consumption and portfolio rules in a continuous-time model.
\newblock {\em Journal of economic theory}, 3(4):373--413, 1971.

\bibitem[MT06]{mt06}
Ravi~R Montenegro and Prasad Tetali.
\newblock {\em Mathematical aspects of mixing times in Markov chains}.
\newblock Now Publishers Inc, 2006.

\bibitem[{\O}ks14]{oksendalSDE}
Bernt {\O}ksendal.
\newblock {\em {Stochastic Differential Equations: An Introduction with
  Applications (Universitext)}}.
\newblock Springer, 6th edition, 2014.

\bibitem[Par13]{p13}
Ori Parzanchevski.
\newblock Mixing in high-dimensional expanders.
\newblock {\em arXiv preprint arXiv:1310.6477}, 2013.

\bibitem[PR12]{pr12}
Ori Parzanchevski and Ron Rosenthal.
\newblock Simplicial complexes: spectrum, homology and random walks.
\newblock {\em arXiv preprint arXiv:1211.6775}, 2012.

\bibitem[PRT12]{prt12}
Ori Parzanchevski, Ron Rosenthal, and Ran~J Tessler.
\newblock Isoperimetric inequalities in simplicial complexes.
\newblock {\em arXiv preprint arXiv:1207.0638}, 2012.

\bibitem[PSSW09]{pssw09}
Yuval Peres, Oded Schramm, Scott Sheffield, and David Wilson.
\newblock Tug-of-war and the infinity laplacian.
\newblock {\em Journal of the American Mathematical Society}, 22(1):167--210,
  2009.

\bibitem[Rod09]{r09}
J.A. Rodriguez.
\newblock Laplacian eigenvalues and partition problems in hypergraphs.
\newblock {\em Applied Mathematics Letters}, 22(6):916 -- 921, 2009.

\bibitem[RS10]{rs10}
Prasad Raghavendra and David Steurer.
\newblock Graph expansion and the unique games conjecture.
\newblock In {\em Proceedings of the 42nd ACM symposium on Theory of
  computing}, pages 755--764. ACM, 2010.

\bibitem[RST12]{rst12}
Prasad Raghavendra, David Steurer, and Madhur Tulsiani.
\newblock Reductions between expansion problems.
\newblock In {\em Computational Complexity (CCC), 2012 IEEE 27th Annual
  Conference on}, pages 64--73. IEEE, 2012.

\bibitem[SJ89]{js89}
Alistair Sinclair and Mark Jerrum.
\newblock Approximate counting, uniform generation and rapidly mixing markov
  chains.
\newblock {\em Information and Computation}, 82(1):93--133, 1989.

\bibitem[SKM14]{skm14}
John Steenbergen, Caroline Klivans, and Sayan Mukherjee.
\newblock A cheeger-type inequality on simplicial complexes.
\newblock {\em Advances in Applied Mathematics}, 56:56--77, 2014.

\bibitem[SM00]{sm00}
Jianbo Shi and Jitendra Malik.
\newblock Normalized cuts and image segmentation.
\newblock {\em Pattern Analysis and Machine Intelligence, IEEE Transactions
  on}, 22(8):888--905, 2000.

\bibitem[SS96]{ss96}
Michael Sipser and Daniel~A. Spielman.
\newblock Expander codes.
\newblock {\em IEEE Transactions on Information Theory}, 42:1710--1722, 1996.

\end{thebibliography}
